\documentclass[10pt,aps,prx,twocolumn,footinbib,floatfix]{revtex4-2}
\usepackage{graphicx}
\usepackage{indentfirst}
\usepackage{siunitx}
\usepackage{physics}
\usepackage{braket}
\usepackage{float}
\usepackage{mathtools}
\usepackage{epstopdf}
\usepackage{footnote}
\usepackage{esint}
\usepackage{comment}
\usepackage{color}
\usepackage[T1]{fontenc}
\usepackage{amsfonts}
\usepackage{amsmath}
\usepackage{braket}
\usepackage{footmisc}
\usepackage{scrextend}
\usepackage{multirow}
\usepackage{diagbox}
\usepackage[hyperfootnotes=false]{hyperref}
\usepackage[acronym]{glossaries}
\usepackage[english]{babel}
\usepackage{url}
\usepackage{bm}
\usepackage{tikz}
\usepackage{etoolbox}
\usepackage{enumerate}
\usepackage{soul}
\usepackage{amssymb}
\definecolor{darkblue}{rgb}{0,0,0.5}
\hypersetup{
    colorlinks=true,
    linkcolor=black,
    filecolor=blue,
    citecolor=darkblue,  
    urlcolor=black,
}

\usepackage[normalem]{ulem}

\usepackage{tikzit}

\tikzstyle{gate}=[fill={rgb,255: red,222; green,222; blue,222}, draw=black, shape=rectangle, minimum width=0.75cm, minimum height=0.75cm]
\tikzstyle{medgate}=[fill={rgb,255: red,222; green,222; blue,222}, draw=black, shape=rectangle, minimum width=0.5cm, minimum height=0.6cm]
\tikzstyle{minigate}=[fill={rgb,255: red,222; green,222; blue,222}, draw=black, shape=rectangle, minimum width=0.5cm, minimum height=0.5cm]
\tikzstyle{state}=[fill={rgb,255: red,222; green,222; blue,222}, regular polygon, regular polygon sides=3, draw, text width=1em, inner sep=0.5mm, outer sep=0mm, shape border rotate=90]
\tikzstyle{ministate}=[fill={rgb,255: red,222; green,222; blue,222}, regular polygon, regular polygon sides=3, draw, font={\footnotesize}, text width=0.7em, inner sep=0.2mm, outer sep=0mm, shape border rotate=90]
\tikzstyle{miniket}=[fill={rgb,255: red,222; green,222; blue,222}, regular polygon, regular polygon sides=3, draw, font={\footnotesize}, text width=0.7em, inner sep=0.2mm, outer sep=0mm, shape border rotate=270]
\tikzstyle{tinystate}=[fill={rgb,255: red,222; green,222; blue,222}, regular polygon, regular polygon sides=3, draw, font={\footnotesize}, text width=0.6em, inner sep=0.2mm, outer sep=0mm, shape border rotate=90]
\tikzstyle{tinyket}=[fill={rgb,255: red,222; green,222; blue,222}, regular polygon, regular polygon sides=3, draw, font={\footnotesize}, text width=0.6em, inner sep=0.2mm, outer sep=0mm, shape border rotate=270]
\tikzstyle{permred}=[fill={rgb,255: red,248; green,206; blue,204}, draw=black, shape=circle, font={\scriptsize}, text width=0.7em, inner sep=0.3mm, outer sep=0mm]
\tikzstyle{permblue}=[fill={rgb,255: red,218; green,232; blue,252}, draw=black, shape=circle, font={\scriptsize}, text width=0.7em, inner sep=0.3mm, outer sep=0mm]
\tikzstyle{permgreen}=[fill={rgb,255: red,218; green,252; blue,232}, draw=black, shape=circle, font={\scriptsize}, text width=0.7em, inner sep=0.3mm, outer sep=0mm]
\tikzstyle{plauette}=[fill={rgb,255: red,169; green,196; blue,235}, regular polygon, regular polygon sides=3, draw, font={\footnotesize}, text width=1.3em, inner sep=0.3mm, outer sep=0mm, shape border rotate=90]
\tikzstyle{plauetterotate}=[fill={rgb,255: red,169; green,196; blue,235}, regular polygon, regular polygon sides=3, draw, font={\footnotesize}, text width=1.3em, inner sep=0.3mm, outer sep=0mm, shape border rotate=270]
\tikzstyle{miniplauette}=[fill={rgb,255: red,169; green,196; blue,235}, regular polygon, regular polygon sides=3, draw, font={\footnotesize}, text width=0.5em, inner sep=0.2mm, outer sep=0mm, shape border rotate=90]

\apptocmd{\sloppy}{\hbadness 9999\relax}{}{}

\DeclareMathOperator{\E}{\mathbb{E}}

\DeclareMathOperator{\Prob}{\mathbb{P}}

\DeclarePairedDelimiterX\parvertent[2]{\lparen}{\rparen}%
{#1\,\delimsize\vert\,\mathopen{}#2}

\usepackage{stmaryrd}
\usepackage{trimclip}
\usepackage[normalem]{ulem}

\makeatletter
\DeclareRobustCommand{\shortto}{%
  \mathrel{\mathpalette\short@to\relax}%
}

\newcommand{\short@to}[2]{%
  \mkern2mu
  \clipbox{{.5\width} 0 0 0}{$\m@th#1\vphantom{+}{\shortrightarrow}$}%
  }
\makeatother

\newcommand*\diff{\mathop{}\!\mathrm{d}}

\newtheorem{theorem}{Theorem}

\newtheorem{corollary}[theorem]{Corollary}

\newtheorem{lemma}[theorem]{Lemma}

\newenvironment{proof}[1][Proof]{\noindent\textbf{#1.} }{\ \rule{0.5em}{0.5em}}

\newcommand{\calC}{{\cal C}}

\newcommand{\calE}{{\cal E}}
\newcommand{\calF}{{\cal F}}
\newcommand{\calI}{{\cal I}}

\newcommand{\calN}{{\cal N}}
\newcommand{\calO}{{\cal O}}

\newcommand{\calH}{{\cal H}}

\newcommand{\calU}{{\cal U}}

\newcommand{\1}{^{(1)}}

\newcommand{\bI}{\boldsymbol I}

\newcommand{\bmz}{{\bm z}}

\newlength{\ketketwidth}
\newlength{\ketwidth}

\newcommand{\kettstylesep}[3]{
    \settowidth{\ketwidth}{$#2\left|#1\right\rangle$}
    \settowidth{\ketketwidth}{$#2\left.\left|#1\right\rangle\right\rangle$}
    \left|#1\right\rangle#3\hspace{\ketwidth}\hspace{-\ketketwidth}
}
\newcommand{\kett}[1]{
    \left.\mathchoice
        {\kettstylesep{#1}{\displaystyle}{\hspace{0.3em}}}
        {\kettstylesep{#1}{\textstyle}{\hspace{0.3em}}}
        {\kettstylesep{#1}{\scriptstyle}{\hspace{0.3em}}}
        {\kettstylesep{#1}{\scriptscriptstyle}{\hspace{0.25em}}}
    \right\rangle
}
\newcommand{\bbrastylesep}[3]{
    \settowidth{\ketwidth}{$#2\left\langle#1\right|$}
    \settowidth{\ketketwidth}{$#2\left\langle\left\langle#1\right|\right.$}
    #3\hspace{\ketwidth}\hspace{-\ketketwidth}\left\langle#1\right|
}
\newcommand{\bbra}[1]{
    \left\langle\mathchoice
        {\bbrastylesep{#1}{\displaystyle}{\hspace{0.3em}}}
        {\bbrastylesep{#1}{\textstyle}{\hspace{0.3em}}}
        {\bbrastylesep{#1}{\scriptstyle}{\hspace{0.3em}}}
        {\bbrastylesep{#1}{\scriptscriptstyle}{\hspace{0.25em}}}
    \right.
}

\newcommand{\bbrakettstylesep}[4]{
    \settowidth{\ketwidth}{$#3\left\langle#1\middle|#2\right\rangle$}
    \settowidth{\ketketwidth}{$#3\left\langle\left\langle#1\middle|#2\right\rangle\right.$}
    #4\hspace{\ketwidth}\hspace{-\ketketwidth}\left\langle#1\middle|#2\right\rangle#4\hspace{\ketwidth}\hspace{-\ketketwidth}
}
\newcommand{\bbrakett}[2]{
    \left\langle\mathchoice
        {\bbrakettstylesep{#1}{#2}{\displaystyle}{\hspace{0.3em}}}
        {\bbrakettstylesep{#1}{#2}{\textstyle}{\hspace{0.3em}}}
        {\bbrakettstylesep{#1}{#2}{\scriptstyle}{\hspace{0.3em}}}
        {\bbrakettstylesep{#1}{#2}{\scriptscriptstyle}{\hspace{0.25em}}}
    \right\rangle
}

\def\be{\begin{equation}}
\def\ee{\end{equation}}
\def\ba{\begin{eqnarray}}
\def\ea{\end{eqnarray}}

\newcommand{\QZ}[1]{{{\textcolor{black}{#1}}}}
\newcommand{\BZ}[1]{{{\textcolor{black}{#1}}}}

\usepackage[normalem]{ulem}

\usepackage[normalem]{ulem}
\usepackage{xr}

\begin{document}


\title{\QZ{Anticoncentrated $n$-bit distribution from $\log(n)$ qubits}}


\author{Bingzhi Zhang$^{1,\dag}$}
\author{Quntao Zhuang$^{1,2,\ddag}$}
\affiliation{
$^{1}$Ming Hsieh Department of Electrical and Computer Engineering, University of Southern California, Los
Angeles, California 90089, USA
\\
$^{2}$Department of Physics and Astronomy, University of Southern California, Los
Angeles, California 90089, USA\\
$^\dag$bingzhiz@usc.edu;
$^\ddag$qzhuang@usc.edu
}

\begin{abstract}
\QZ{
Random circuit sampling (RCS) is a leading approach to demonstrate quantum advantage, with its believed classical hardness rooted in anticoncentration of output distributions and average-case hardness of probability estimation. Here we show that this association is not fundamental. We introduce holographic random circuit sampling (HRCS), a spatiotemporal protocol that interleaves random unitary evolution with mid-circuit measurements. We prove that $n$ classical bits exhibiting $\epsilon$-approximate anticoncentration of Haar random states can be generated using only $\calO(\log n)$ physical qubits and linear depth, establishing a precise space–time trade-off and indicating efficient classical simulation. Our analyses is built upon exact formulas for collision probability and higher-order power sums. \BZ{Our experimental validation on IBM Quantum devices demonstrates sampling up to $200$ classical bits using only $20$ qubits.}
}
\end{abstract}


 \maketitle


\section{Introduction}
\QZ{Demonstrating quantum advantage with near-term devices is a central goal in quantum science. 
Random circuit sampling (RCS)~\cite{hangleiter2023} has emerged as a leading paradigm for this purpose, serving both as a proposed demonstration of quantum supremacy and as a benchmark for near-term quantum processors~\cite{boixo2018characterizing}. Experimental realizations across multiple platforms~\cite{arute2019quantum,wu2021strong,zhu2022quantum,morvan2024phase, gao2025establishing,deCross2025, abanin2025observation, ransford2025helios} have established RCS as a central testbed for quantum computational complexity.}

\QZ{The believed classical hardness of RCS relies on two structural ingredients: anticoncentration of the output distribution and the average-case hardness of approximating output probabilities~\cite{bouland2019complexity}. \BZ{Anticoncentration, the spreading of probability mass over exponentially many outcomes, is widely viewed as a hallmark of simulation hardness and is often implicitly associated with the quantum state entanglement, one of the key quantum resources.} While exact probability evaluation for typical polynomial-depth circuits has been proven average-case hard~\cite{movassagh2023hardness}, anticoncentration and constant-order moment matching have been rigorously established even for shallow random circuits~\cite{barak2020spoofing, dalzell2022random,schuster2025random}.}

\QZ{A common implicit assumption is that generating $n$ classical bits exhibiting anticoncentration of Haar random states requires an $n$-qubit quantum system. 
\BZ{In other words, the attainable level of anticoncentration is believed to (inversely) scale with the spatial Hilbert space dimension.}}
\QZ{Here, we overturn this assumption by introducing \emph{holographic random circuit sampling} (HRCS), a spatiotemporal protocol that interleaves random unitary evolution with mid-circuit measurements. \BZ{For a fixed number of physical qubits, HRCS allows an expansion of sampling size in exponential scaling of number of qubits while maintaining Haar-level anticoncentration, showing a huge gap of coherent quantum memory against independent sampling. In turn, we prove that the distribution of $n$ classical bits exhibiting $\epsilon$-approximate Haar-level anticoncentration can be generated using only $\calO(\ln n)$ physical qubits at the cost of $n$ temporal steps for a total depth of $\calO(n\ln n)$, establishing a precise space–time trade-off in random circuit sampling. Experimentally, we implement HRCS on IBM Quantum devices and demonstrate effective sampling up to $200$ classical bits, while using only $20$ physical qubits.}
}


\QZ{Our results reveal a conceptual separation between structural randomness and computational hardness: anticoncentration and finite-moment statistics do not by themselves require large quantum systems or classical intractability. \BZ{We also demonstrate HRCS for benchmarking the combined quality of circuit gates and mid-circuit measurements in experiments.}
}

\begin{figure*}[t]
    \centering
    \includegraphics[width=\textwidth]{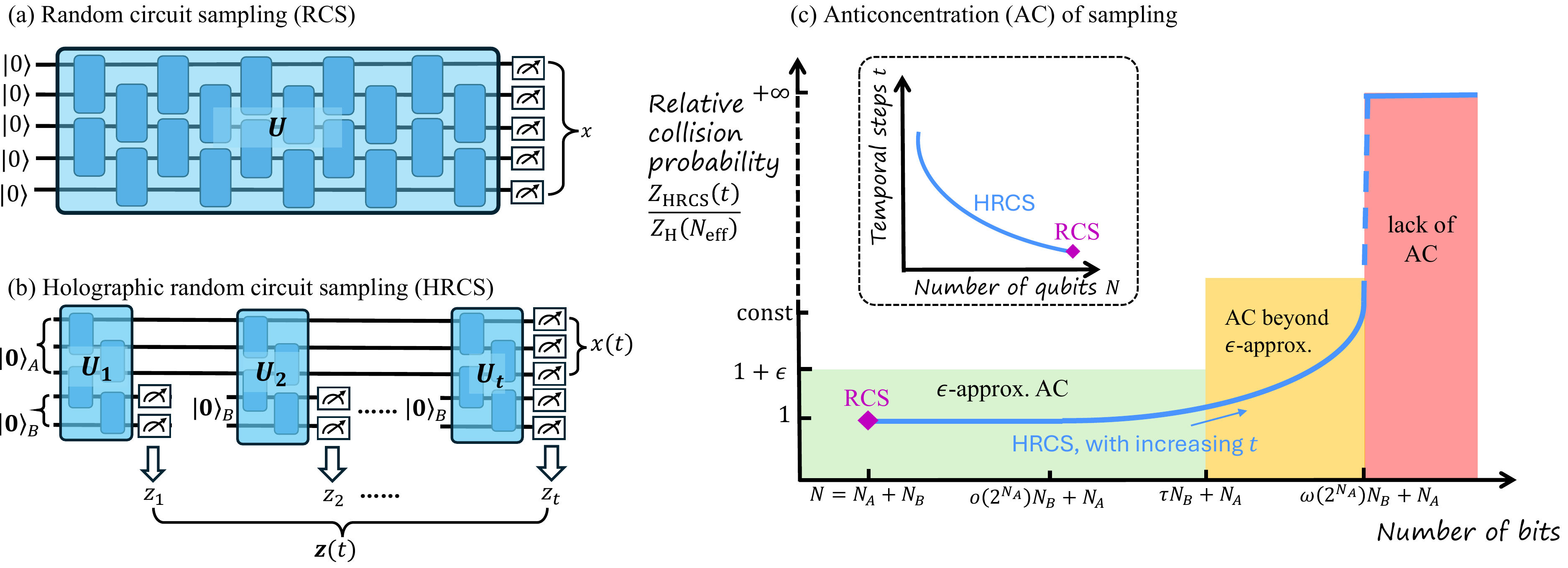}
    \caption{A schematic for quantum circuit sampling methods and main results. In (a) random circuit sampling (RCS), a random unitary circuit is applied on a trivial initial state and one performs computational basis measurements on the output state. (b) In holographic random circuit sampling (HRCS), computational basis measurements are performed on the bath $B$ in each step following a unitary circuit, and on the system $A$ in the end. The sampling task is targeted at the joint distribution of temporal mid-circuit measurements $\bmz(t)$ and the final state measurements $x(t)$. In (c), we summarize the anticoncentration (AC) of sampling in HRCS in the asymptotic limit of $N_A \to \infty$ for increasing effective system size in HRCS. While RCS is a single operating point (magenta diamond), HRCS allows the increase of effective system size (blue curve). The inset shows the spacetime tradeoff of HRCS holding AC. }
    \label{fig:concept}
\end{figure*}



\

\section{Random Circuit Sampling problem}
As shown in Fig.~\ref{fig:concept}a, in RCS, the quantum system of $N$ qubits is initialized in a trivial product state, i.e. $\ket{\bm 0}=\ket{0}^{\otimes N}$, and a random unitary circuit $U$ is applied on the system. Finally, one performs measurements in a fixed set of basis, i.e. the computational basis, on the output state and collects the corresponding measurement result $x \in \{0,1\}^{N}$. In this regard, each measurement outcome appears with probability $p(x)= |\braket{x|U|\bm 0}|^2$.


A key ingredient and necessary condition for the hardness of RCS is the anticoncentration (AC) of the measurement distribution. For shallow circuits like constant-depth ones, due to the information lightcone, the output state is only supported on a few computational bases. As a result, the simulation of sampling from $p(x)$ remains relatively easy. 
However, when the circuit depth increases and the unitary becomes complex, $p(x)$ develops AC---widely supported over a number of bases exponential in the number of qubits, leading to a hard task for reproduction on classical devices. Indeed, it has been proven that the computation of $p(x) $ for a typical Haar random unitary is an NP-hard problem~\cite{bouland2019complexity}.

To quantify AC, we consider the widely-adopted~\cite{barak2020spoofing, dalzell2022random, dalzell2024random, fefferman2024effect} metric of collision probability (CP) defined by
\be
    Z \coloneqq \sum_x p(x)^2.
    \label{eq:CP_def}
\ee
Given a shallow quantum circuit, large values of CP indicate concentration---the corresponding RCS is easy. With the increase of layers of gates, the CP of an $N$-qubit state approaches the Haar value of 
\be 
Z_{\rm H}(N)= 2/\left(2^N+1\right), 
\label{ZH_overview}
\ee 
inversely proportional to the system dimension $d=2^N$, leading to a hard RCS problem~\footnote{Note that $Z=Z_H$ is the sweet spot for RCS complexity, as further smaller CP such as the extremely case of $Z_{\rm uni}=1/d$ for a uniform distribution, RCS becomes simple again.}.



\QZ{The central question we address is: to generate $n$ random bits following distribution $\epsilon$-close to AC of Haar random states, $Z \le (1+\epsilon)Z_{\rm H}(n)$, what is the minimum required resource (number of qubits, and circuit depth) for a quantum algorithm?} \BZ{Conventional RCS requires $N = n$ qubits and $\ln(n/\epsilon)$ circuit depth to achieve the $\epsilon$-approximate AC~\cite{barak2020spoofing,dalzell2022random}, but can a relatively small quantum system make use of a large quantum circuit depth to achieve AC?}

\begin{figure}[t]
    \centering
    \includegraphics[width=0.45\textwidth]{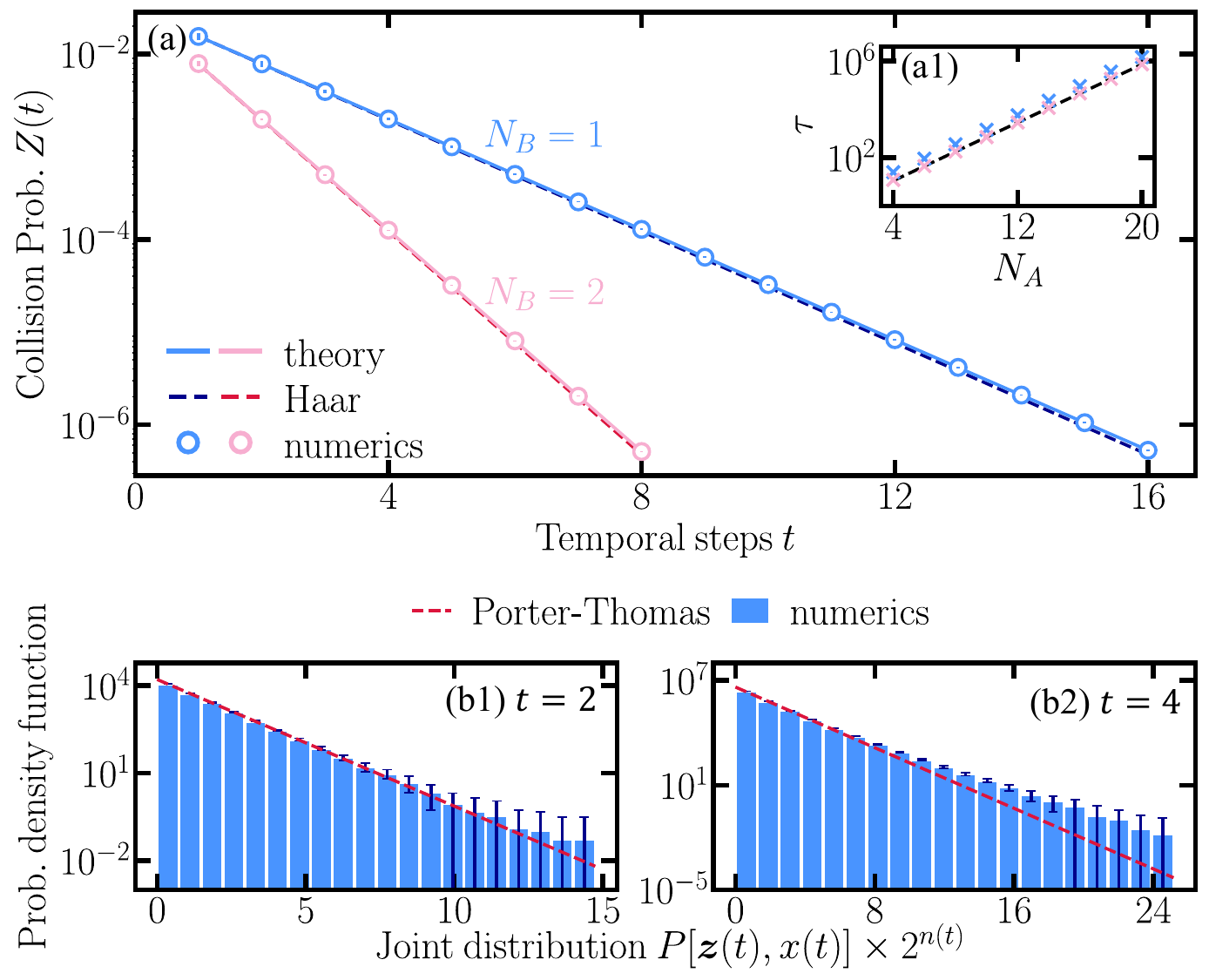}
    \caption{
    \BZ{
    (a) Ensemble-averaged collision probability (CP) $Z(t)$ for HRCS versus temporal steps in a system of $N_A=6$, $N_B=1, 2$ (blue and pink circles) qubits. Colored solid lines represent theoretical result of $Z_{\rm HRCS}(t)$ of Eq.~\eqref{eq:HRCS_CP_spt} in Theorem~\ref{HRCS_CP_spt}. Dark-colored dashed lines are CP for Haar random states $Z_{\rm H}(n)$ in Eq.~\eqref{ZH_overview} with $n = N_A + t N_B$ qubits. Inset (a1) shows the growth of critical temporal steps $\tau$ versus system size $N_A$, with $\epsilon=1$. Blue and pink crosses represent numerical solutions of $Z_{\rm HRCS}(t) = (1+\epsilon)Z_{\rm H}(N_{\rm eff})$ for $N_B=1$ and $N_B=8$ separately. The black line is Eq.~\eqref{eq:max_steps_overview}.
    (b1-b2) Ensemble-averaged probability density function of the joint sampling distribution $P[\bmz(t), x(t)]$ in HRCS of $N_A=6, N_B=4$ qubits at $t=2$ and $t=4$. Red dashed line is the Porter-Thomas distribution corresponding to Haar case.
    The ensemble average is over $50$ HRCS instances with corresponding error bars.
    }
    }
    \label{fig:HRCS_theory}
\end{figure}

\ 

\section{Holographic random circuit sampling}
In this work, we propose holographic random circuit sampling (HRCS) to go beyond the limitation of conventional RCS and \QZ{allows a trade-off between space and time resources---only using as small as $\ln(n)$ qubits to enable anticoncentration of $n$ bits.}

As shown in Fig.~\ref{fig:concept}b, HRCS adopts layers of quantum circuits $U_1,U_2,\cdots, U_t$ jointly on a system of $N_A$ qubits and a bath of $N_B$ qubits, interleaved with mid-circuit measurements on the bath to obtain measurement results, in addition to the final measurements on all qubits. After each measurement, one can optionally choose reset or not reset the bath, as the unitaries are assumed to be random. Besides collecting the `temporal' measurement results $\bmz(t)= (z_1,\cdots,z_t)$ from bath, the final measurement on the system $A$ gives additional `spatial' contribution to the measurement results $x(t)$. The entire measurement outcome, $\bmz(t)$ and $x(t)$, \BZ{has in total of $n(t) \coloneqq N_A + t N_B$ classical bits} growing with the number of rounds of mid-circuit measurements $t$. 

Our main result is that AC remains for the \BZ{$2^{n}$}-dimensional measurement results up to an exponential time scale of $t\sim 2^{N_A}$, as sketched in blue line of Fig.~\ref{fig:concept}c. 
\BZ{Moreover, HRCS allows a spacetime tradeoff between the number of qubits $N$ and temporal steps $t$ achieving $\epsilon$-approximate AC, illustrated in the inset. On the other hand, conventional RCS can only achieve $Z_{\rm H}(N)$ for a fixed constant number of bits $N = N_A+N_B$ with $N$ qubits and $t=1$ step (magenta diamonds). Therefore, HRCS is able to effective utilize the `quantum volume' of a device to exponentially increase the sampling dimension while keeping AC.} 



Formally, we have the following theorem.
\begin{theorem}
\label{HRCS_CP_spt}
    For holographic random circuit sampling with each unitary $U_t$ in $2$-design, the ensemble-averaged collision probability at step $t\ge 1$ is
    \be
        Z_{\rm HRCS}(t) = \frac{2(d_A+1)^{t-1}}{ (1+d_A d_B)^t},
        \label{eq:HRCS_CP_spt}
    \ee
    where $d_A = 2^{N_A}, d_B = 2^{N_B}$ are Hilbert space dimensions of the system and bath. In the large-system limit $d_A \gg 1$,
    \begin{align}
        &Z_{\rm HRCS}(t) \nonumber\\
        &= Z_{\rm H}(\BZ{n})\exp\left[\frac{t\left(1-d_B^{-1}\right) + d_B^{-t}-1}{d_A} + \calO\left(\frac{1}{d_A^2}\right)\right].
        \label{eq:HRCS_CP_spt_asymp}
    \end{align}
\end{theorem}

At $t=1$, both Eqs.~\eqref{eq:HRCS_CP_spt} and \eqref{eq:HRCS_CP_spt_asymp} recover the RCS result in Eq.~\eqref{ZH_overview} for a system of $N_A+N_B$ qubits. We emphasize that it is crucial to sample from the joint spatiotemporal distribution---marginal distributions are close to uniform (see Appendix \ref{app:marginal_sampling}). In Fig.~\ref{fig:HRCS_theory}a, we verify Eq.~\eqref{eq:HRCS_CP_spt} (solid) with numerical simulations (circles) where the CP exponentially decays with temporal steps $t$, for both $N_B=1$ (blue) and $N_B=2$ (pink). We also see that the CP of HRCS is close to the Haar results $Z_{\rm H}(n)$ in Eq.~\eqref{ZH_overview}, depicted as the dashed lines. 
The effective sampling size $n(t) = N_A+tN_B$ increases with growing $t$. However, such an increase \BZ{with close approximation to Haar value} does not last forever, as expected. 

To understand the maximum \BZ{sampling size} achievable by HRCS, we define the $\epsilon$-approximate Haar AC by
$Z_{\rm HRCS}(t)\le (1+\epsilon)Z_{\rm H}(n)$, for an arbitrary constant $\epsilon>0$. From Eq.~\eqref{eq:HRCS_CP_spt_asymp}, $\epsilon$-approximate AC of Haar level is satisfied when the number of temporal steps in HRCS (green region in Fig.~\ref{fig:concept}c) is bounded by
\begin{align}
    \BZ{t \le \tau \coloneqq \frac{d_A d_B}{d_B-1}\log(1+\epsilon)\simeq  d_A\ln(1+\epsilon)}.
    \label{eq:max_steps_overview}
\end{align}
Therefore, $\epsilon$-approximate AC survives until an exponentially large time scale \BZ{$t \sim \calO\left(2^{N_A}\right)$}, leading to an effective sampling size of $\calO\left(2^{N_A}\right)N_B+N_A$ exponential in $N_A$, as we illustrate in Fig.~\ref{fig:concept}c. We verify the critical temporal steps $\tau$ with numerical solutions of $Z_{\rm HRCS}(\tau) = (1+\epsilon)Z_{\rm H}(\BZ{n})$ in Fig.~\ref{fig:HRCS_theory}a1, where we see the convergence toward a universal scaling of $2^{N_A}\ln(1+\epsilon)$ as shown by the black line, with increasing $N_B$ (blue to pink crosses). 

Even at $t>\tau$, when $t/2^{N_A} \le c$ remains bounded by a constant, the sampled distribution in HRCS remains AC, despite with a larger constant approximation error (yellow region in Fig.~\ref{fig:concept}c). However, if $t/2^{N_A} \sim \omega(1)$ (red region), then $Z_{\rm HRCS}(t)/Z_{\rm H}(n)$ grows with system size $N_A$, leading to the lack of AC. 


We further confirm the above observation from CP via other statistical tools. We begin with a simple numerical check on the probability density of the joint distribution probability (PoP) in Fig.~\ref{fig:HRCS_theory}b1-b2 for sampling at different steps. The empirical PoP of $P[\bmz(t), x(t)]$ (bars) for a system of $N_A=6, N_B=4$ aligns closely with the Porter-Thomas distribution for Haar random outcomes of dimension $2^{N_A+tN_B}$ (red dashed line) at all steps, up to small deviations in rare outcomes due to finite-size effect. \BZ{The finite-size effect is more visible for larger $t$ as expected from Eq.~\eqref{eq:HRCS_CP_spt_asymp}.
More formally, we extend the analysis to the $K$-th power sum of sampling distribution $Z^{(K)} \coloneqq \sum_x p(x)^K$ focusing on properties of higher-order statistical moments. We show that the $K$-th power sum of HRCS is $Z_{\rm HRCS}^{(K)}(t) \lesssim Z_{\rm H}^{(K)}(n) \exp\left[K^2 t/(2d_A)\right]$ with $Z_{\rm H}^{(K)} \simeq K!/2^{nK}$ as the $K$-th power sum of $n$-qubit Haar random states (see formal results in Methods). Therefore, HRCS holds $\epsilon$-approximate Haar value of $K$-th power sum till $\tau^{(K)} \simeq \calO(\tau/K^2)$, which exhibits the same scaling as CP in Eq.~\eqref{eq:max_steps_overview} but suppressed by statistical moment order $\calO(1/K^2)$. Note that $\epsilon$-approximate $K$-th power sum directly implies approximation of $K'$-th power sums with $K'<K$.
}

\section{Space-time trade-off}
\BZ{
Here we analyze the space-time trade-off in HRCS for sampling $n$ bits, between the number of necessary physical qubits (space) $N \coloneqq N_A + N_B$ and the number of temporal steps $t$ (time). 
To generate the $n$ classical bits, while conventional RCS needs $N=n$ qubits, HRCS only requires the condition $n=N_A t+N_B$ to be satisfied. In the asymptotic limit of $n\to \infty$, for $\epsilon$-approximate Haar AC, we have the following approximate lower bound for necessary number of qubits as
\be
    N \gtrsim \frac{n}{t} + \frac{t-1}{t}\log_2\left(\frac{t}{\ln(1+\epsilon)}\right) + \calO\left(\frac{\ln n}{n}\right),
    \label{eq:N_t_tradeoff}
\ee
which can be derived following the maximum time steps in Ineq.~\eqref{eq:max_steps_overview} (see details in Appendix~\ref{app:tradeoff}). When $t = 1$, we recover the conventional RCS of $N = n$. As the time $t$ increases, the space resource $N$ initially enjoys a $1/t$ decrease; as we increase time more and more to reduce space, the second logarithmic term dominates and increases the space resource. At $t^\star =\ln(2) n + \calO(\ln n/n)$, which corresponds to $N_{B,\rm min} = 1/\ln(2) + \calO(\ln n/n)\simeq 1.4$, $N$ achieves the minimum of 
\be 
N_{\rm min} = \log_2\left(\frac{n}{\log_2(1+\epsilon)}\right) + \frac{1}{\ln(2)} +\calO\left(\frac{\ln n}{n}\right).
\label{eq:N_min}
\ee 
However, as $N_B$ is an integer, we expect the minimum $N_{\rm min}$ to be achieved either with $N_B=1$ or $N_B=2$. We want to highlight that this result only relies on the global $2$-design unitary in each step of Theorem~\ref{HRCS_CP_spt}, which in general holds for arbitrary spatial dimension.}

\BZ{Alternatively, if we simply perform conventional RCS on $N$ qubits for $n/N$ times and collect the final joint distribution of $n$ bits, the overall CP is $Z_{\rm prod}(N)\coloneqq \left[Z_{\rm H}(N)\right]^{n/N} \simeq 2^{n/N-1}Z_{\rm H}(n)$. When $N \sim \calO(\ln n)$ qubits, the corresponding deviation scales as $Z_{\rm prod}(N)/Z_{\rm H}(n) \sim n^{\omega (1)}$, super-polynomial in $n$, which breaks AC condition and shows a sharp contrast against HRCS, thus highlighting the importance of coherent quantum memory. Formally speaking, for $\epsilon$-approximate AC, it requires $N \sim \calO(n/(1+\ln(1+\epsilon)))$, only a constant factor smaller than $n$. This partition approach is also equivalent to the patch circuit method considered in conventional RCS experiment~\cite{arute2019quantum, wu2021strong,zhu2022quantum, gao2025establishing}.}

\begin{figure}[t]
    \centering
    \includegraphics[width=0.45\textwidth]{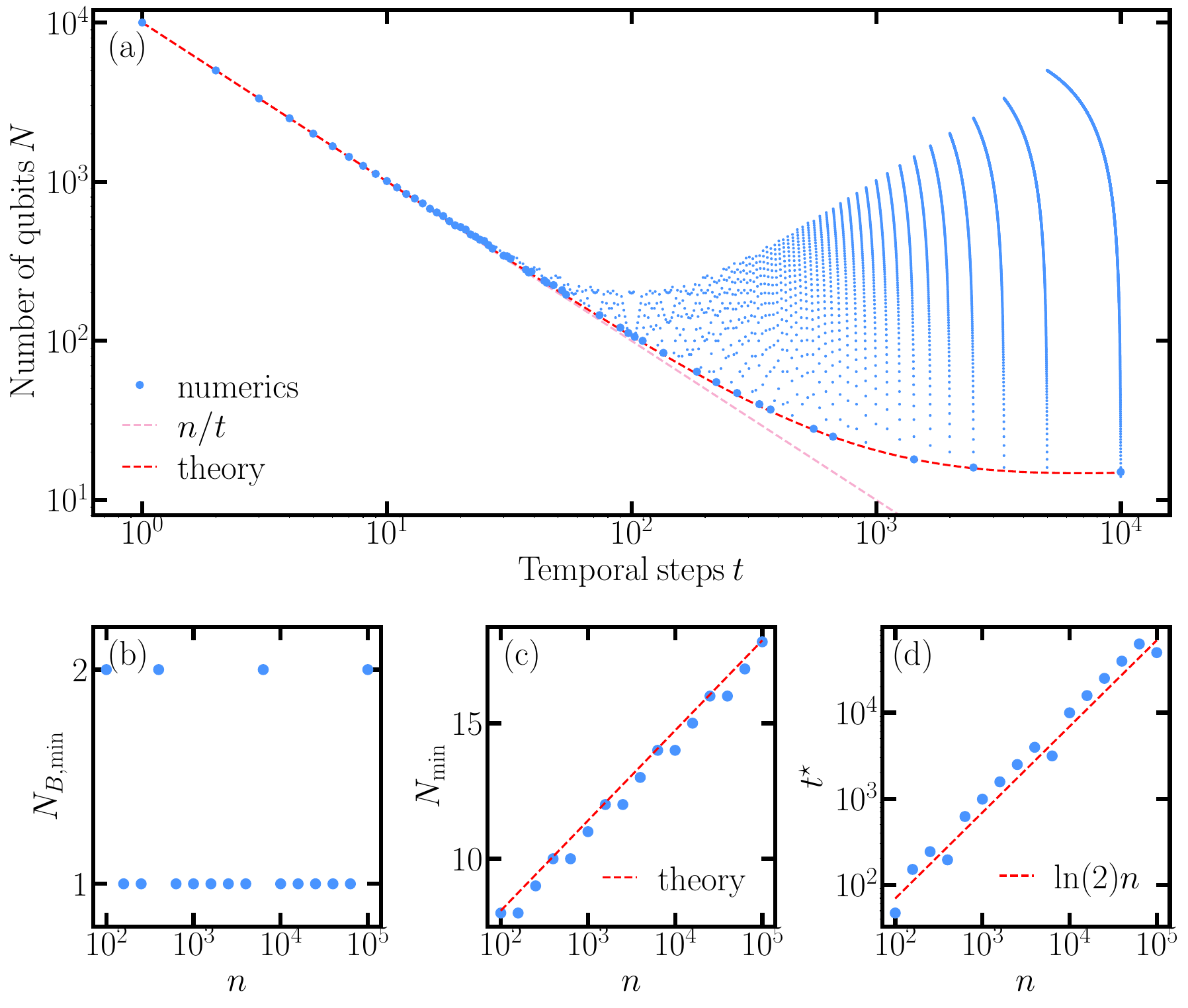}
    \caption{\BZ{(a) Tradeoff between physical qubit number $N$ versus temporal steps $t$ for generating $n = 10^4$ bits. Blue dots represent numerical solutions, and the bottom envelop are enlarged. Red and pink dashed lines show the theoretical approximate lower bound in Eq.~\eqref{eq:N_t_tradeoff} and the asymptotic form of $n/t$, respectively. In (b)-(d), we plot the scaling of the minimum bath size $N_{B, \rm min}$, minimum qubit number $N_{\rm min}$, and corresponding time steps $t^\star$ versus number of bits $n$ separately. Red dashed lines in (c) and (d) represent corresponding theoretical results of Eq.~\eqref{eq:N_min} and $\ln(2)n$. In all cases, we set $\epsilon = 1$.}
    }
    \label{fig:spacetime}
\end{figure}

\BZ{In Fig.~\ref{fig:spacetime}a, we plot the numerical solutions of the minimum number of physical qubits $N$ for every achievable time $t$ satisfying $Z_{\rm HRCS}(t) \le (1+\epsilon)Z_{\rm H}(n)$ with $\epsilon = 1$ and $n = 10^4$. We impose the integer constraints in numerical evaluations as $N_A, N_B, t$ must be positive integers and $n = N_A + t N_B$. For small number of temporal steps (i.e., $t\lesssim 50$), the numerical results (blue circles) agree with our theory of Eq.~\eqref{eq:N_t_tradeoff} (red dashed line) with an asymptotic scaling of $N\sim n/t$ (orange dashed line). On the other hand, for large number of temporal steps (i.e., $t \gtrsim 10^2$), the minimum qubit number exhibits a quasi-comb structure due to the integer constraints. More importantly, our theory of Eq.~\eqref{eq:N_t_tradeoff} (red dashed line) still characterizes the bottom envelop of minimum number of physical qubits highlighted by larger blue circles. Regarding the minimum space resources, we show the numerical minimum number of bath size $N_{B, \rm min}$ in Fig.~\ref{fig:spacetime}b, which is either one or two depending on $n$ and is consistent our theory prediction of $1/\ln(2)$ leveraging integer constraints. In Fig.~\ref{fig:spacetime}c-d, we further confirm the scaling of minimum number of physical qubits $N_{\rm min}$ and corresponding time steps $t^\star$ with our theoretical results in Eq.~\eqref{eq:N_min} and $t^\star = \ln(2) n$ (red dashed lines in b2 and b3). Overall, for generating random classical bit string, HRCS provides an exponential advantage in used physical qubits at the cost of longer time. 
}

\BZ{We can also extend the above analysis to $K$-th order power sums. For generating $n$ bits holding $\epsilon$-approximate Haar $K$-th power sums with $K \lesssim 2^{n/2}$, the minimum number of physical qubits is simply $N_{\rm min}^{(K)} \simeq N_{\rm min} + \log_2(K^2/2)$ with $N_{\rm min}$ defined in Eq.~\eqref{eq:N_min}. The minimum number of physical qubits also scales logarithmically with $n$ but involves an extra correction from the statistical moment order (see formal results in Methods).}

\BZ{We next turn to the total number of gate layers in HRCS assuming the 1D brickwork unitary in each step consisting of local 2-qubit gates on nearest neighbors. To illustrate it, we first have the upper bound for the CP as
\begin{lemma}
\label{lemma_HRCS_CP}
    For holographic random circuit sampling with each unitary $U_t$ to be $L$ layers 1D brickwork circuit where every $2$-qubit unitary satisfies $2$-design, the ensemble-averaged collision probability at step $t \ge 1$ is upper bounded by
    \be
        Z(t) \le \left[(2t-1)\exp\left(\frac{N}{2} \left(\frac{4}{5}\right)^L\right) - (2t-2)\right]Z_{\rm HRCS}(t)
        \label{eq:HRCS_CP_brickwork},
    \ee
    where $Z_{\rm HRCS}(t)$ is the ensemble-averaged CP with global 2-design unitaries in Eq.~\eqref{eq:HRCS_CP_spt}.
\end{lemma}
The proof of lemma utilizes the mapping to a statistical physics model proposed in Ref.~\cite{hunter2019unitary, barak2020spoofing, dalzell2022random} and is briefly sketched in Methods and fully presented in Appendix~\ref{app:brickwork}. Note that the finite circuit depth $L$ introduces an additional factor in Eq.~\eqref{eq:HRCS_CP_brickwork} compared to the global $2$-design unitary result of $Z_{\rm HRCS}(t)$. In particular, for CP remains constant deviation from $Z_{\rm HRCS}(t)$, we require $L \sim \ln(Nt)$.
}

\BZ{In Fig.~\ref{fig:spacetime_d}a, we show the tradeoff between total circuit depth $tL$ versus number of physical qubits $N$. In the large $N$ region, where $t\sim \calO(1)$, we can estimate the necessary layer of gates for each step is $L \sim \ln(n)$, and is confirmed with the pink dashed line of $tL \sim n \ln(n)/N$ following Eq.~\eqref{eq:N_t_tradeoff}. On the other hand, for minimum spacial resources, we expect the solution Eq.~\eqref{eq:N_min} still holds (verified in Appendix~\ref{app:brickwork}), and the corresponding circuit depth for each step is verified in Fig.~\ref{fig:spacetime_d}b with logarithmic scaling of $L^\star \sim \ln(n)$ (red dashed line). We can also conclude similar results from unitary design. When the unitary in each step of HRCS is about $\ln(1+\epsilon)/t$-approximate $2$-design, the ensemble-averaged CP at step $t\ge 1$ exhibits a relative error $\epsilon$ from Haar value. For $t \sim n$ to minimize spatial resources, we thus require the circuit depth in each step to be $L\sim \ln(Nt/\ln(1+\epsilon))\sim \ln(n)$ utilizing the shallow random unitary construction proposed in Ref.~\cite{schuster2025random} though it is slightly different from the brickwork circuit. Same conclusions can be extended to $K$-th order power sums as well.
}

\begin{figure}
    \centering
    \includegraphics[width=0.45\textwidth]{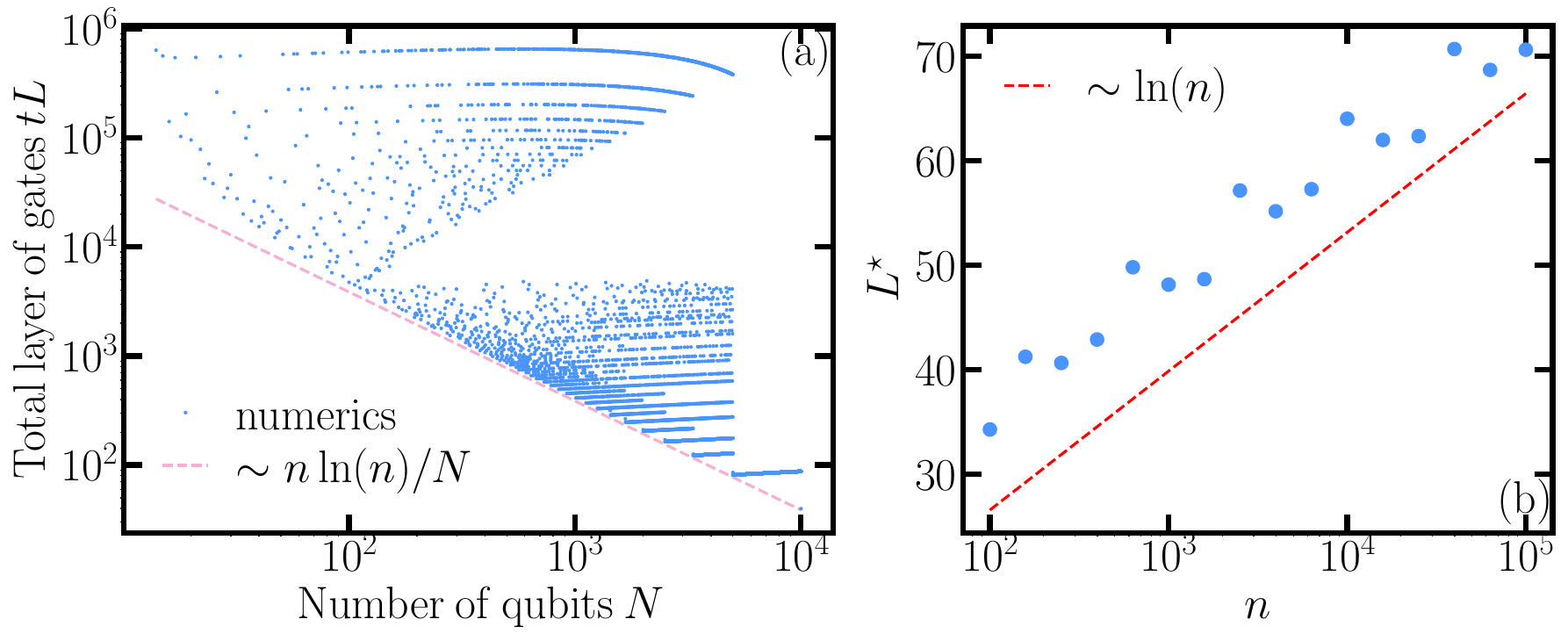}
    \caption{\BZ{(a) The tradeoff between total circuit depth $tL$ versus number of physical qubits $N$. Blue dots represent numerical solutions and pink dashed lines show the asymptotical behavior of $tL \sim n\ln(n)/N$. Here we choose $n = 10^4$. In (b), we plot the scaling of layer of gates in each step $L^\star$ for minimum $N$ versus number of bits $n$. The red dashed line represents $L^\star \sim \ln(n)$. In all cases we set $\epsilon = 1$.}}
    \label{fig:spacetime_d}
\end{figure}

\ 

\section{Experimental validation}
We experimentally implement HRCS on IBM quantum devices, with a first 10-qubit experiment and a second 20-qubit experiment. \QZ{Conventional RCS provides a way to benchmark quantum devices, it fails to capture the quality of mid-circuit measurement that is essential to quantum error correction. Meanwhile, HRCS provides a way to benchmark quantum devices combining quality of gates and mid-circuit measurements.}

\BZ{In conventional RCS, the vanilla cross-entropy benchmarking (XEB) is adopted to verify the experimental results~\cite{boixo2018characterizing, bouland2019complexity,neill2018blueprint}, which directly compares the empirical probability of measured bitstrings in experiments and the ideal one from classical simulation. For measured bitstrings of $n$ bits, the vanilla XEB is defined as
\be
    \chi \coloneqq 2^n \sum_{i} P(x_i)\hat{P}(x_i) - 1,
    \label{eq:XEB_def}
\ee
where $P(x_i), \hat{P}(x_i)$ are the ideal and empirical probability of bitstring $x_i$ from noiseless circuits and experiments, respectively. When the experiment is entirely faithful to the ideal quantum circuit, 
$
\chi=2^{n} Z-1,
$ 
where $Z$ is the CP defined in Eq.~\eqref{eq:CP_def}. For the exact Haar distribution, $\chi_{\rm H} \simeq 1$; while a random uniform sampling independent of the implemented quantum circuit leads to $\chi_{\rm uni}=0$. 
}

\BZ{In experiments, the hardware noise like depolarizing noise can significantly suppress the XEB values from ideal theory. To characterzie the noise effect on the XEB, we develop a coarse-grained noisy circuit model by adding depolarizing channels $\calN_{A \backslash B}(\rho) = \gamma_{A\backslash B} \rho + (1-\gamma_{A\backslash B}) \bI/d$ on system $A$ and bath $B$ separately following the random unitary $U_t$ in every step. In the asymptotic limit of $d_A, d_B \gg 1$, we derive the vanilla XEB for HRCS under the noisy circuit model as (see details and full formula in Appendix \ref{app:noisy_theory})
\begin{align}
    \chi_{\rm HRCS}(t) =\gamma_A^t \gamma_B^t + \frac{\gamma_A \gamma_B (2  + \gamma_B -\gamma_A \gamma_B)}{d_A(1-\gamma_A \gamma_B)},
    \label{eq:noisy_xeb}
\end{align}
In general, $\gamma_A, \gamma_B$ are free parameters, and here we further simplify it by $\gamma_A = q^{N_A}, \gamma_B = q^{N_B}$ where $q$ can be regarded as the fidelity of a single qubit through gates and measurements in one step. In the early time of $t \lesssim N_A$, the XEB decays exponentially with temporal steps $\chi_{\rm XEB}(t) \sim \gamma_A^t \gamma_B^t$, which can be interpreted as the exponential suppression of fidelity through time similar to the conventional RCS~\cite{boixo2018characterizing, arute2019quantum}. The vanilla XEB also coincides with the noisy state fidelity on $N_A + t N_B$ qubits (see proof in Appendix~\ref{app:noisy_theory}), which also implies the power for benchmarking gate and mid-circuit measurement fidelity. However, in the late time $t\gtrsim N_A$, the XEB converges to an approximate constant, and fails to stand as the proxy for any state fidelity due to large accumulated noise, which is also observed in RCS~\cite{gao2024limitations}. 
}

To showcase the capability of HRCS, we begin with only 10 qubits of IBMQ Torino: $N_A=5$ qubits for the system and $N_B=5$ qubits for the bath. 
\BZ{
In HRCS, due to the circuit structure, the ideal XEB $\chi_{\rm HRCS} = 2^n Z_{\rm HRCS}(t)-1 = 2\exp(t/d_A)-1$ also grows with time. For fair comparison across different time, we introduce the normalized XEB as $\tilde{\chi}_{\rm HRCS}(t) \coloneqq \chi_{\rm HRCS}(t)/(2^n Z_{\rm HRCS}(t)-1)$, which reduces to the vanilla one at $t=1$ for conventional RCS.
The results are presented in Fig.~\ref{fig:experiment}a. At $t=1$ ($n=10$), the ideal normalized XEB $\tilde{\chi} \simeq 1.04$ matches the expectation of conventional RCS though with slight deviation due to shot noise and finite circuit depth. Within the experimental accessible time up to $t = 16$ ($n=85$), the ideal normalized XEB is bounded by $\tilde{\chi} \le 1.14$ (pink circles), indicating that the implemented circuits in experiments can approximate random unitaries in theory. Our experimental data (blue circles) exhibits the two-scale exponential decay predicted by the asymptotic noisy theory of Eq.~\eqref{eq:noisy_xeb} (blue dashed line) with a single fitting parameter $q = 0.93$. Remarkably, utilizing only $10$ physical qubits, our HRCS achieves the normalized XEB of $0.045$ and $0.021$ for generating $n = 55$ and $85$ classical bits. By utilizing the average error rates of IBMQ Torino in Table~\ref{tab:ibm_error} and number of gates (see Appendix~\ref{app:experiment}), we obtain an estimation of single qubit fidelity of $q \simeq 0.944$, which is slightly larger than our fitting result due to unmodeled error. Therefore, the XEB score of HRCS can be considered to benchmark the joint quality of circuit gates and mid-circuit measurements.
}

\BZ{We next implement HRCS in a larger scale with $N_A = 10$ system and $N_B = 10$ bath qubits in Fig.~\ref{fig:experiment}b. We utilized the widely-adopted patch circuit method \cite{arute2019quantum, morvan2024phase, wu2021strong, zhu2022quantum, gao2025establishing, deCross2025}---implementing two disjoint HRCS circuits simultaneously to maintain high circuit fidelity. The ideal normalized XEB is bounded by $\tilde{\chi}\le 1.24$ up to $t = 19$ (pink circles).
Good agreement between the experimental data (blue circles) and the noisy theory (blue dashed line) persists with a fitting parameter of $q = 0.92$ slightly smaller than the 10-qubit experiment as more qubits are used. In particular, we report $\tilde{\chi}_{\rm HRCS}(t=9) = 0.025$ and $\tilde{\chi}_{\rm HRCS}(t=19) = 0.0056$ for generating $n=100$ and $n=200$ random bits, respectively.
}

\begin{figure}[t]
    \centering
    \includegraphics[width=0.45\textwidth]{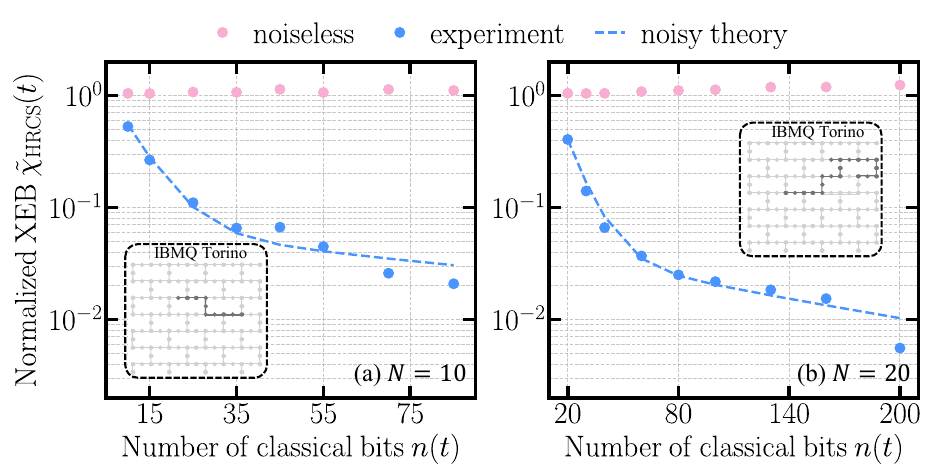}
    \caption{\BZ{{\bf Experimental verification of HRCS}. Normalized cross-entropy benchmark (XEB) of HRCS $\tilde{\chi}_{\rm HRCS}(t) = \chi_{\rm HRCS}(t)/(2^n Z_{\rm HRCS}(t)-1)$ versus number of classical bits $n(t)$ in a system of (a) $N_A = N_B = 5$ qubits and (b) $N_A = N_B = 10$ qubits in two patches. Orange and blue filled circles show the noiseless simulation and experimental results on IBMQ Torino. Blue dashed line represents the noisy theory prediction of Eq.~\eqref{eq:noisy_xeb} with (a) $q=0.93$ and (b) $q = 0.92$, respectively. We perform $10^6$ shots for sampling each circuit instance and each circle is an average of $10$ instances. Dark grey circles in the insets indicate the used qubits in experiments for (a) $t\le 10$ and at (b) $t=19$.}}
    \label{fig:experiment}
\end{figure}


\ 

\section{Discussion}
\BZ{We have proposed the HRCS for utilizing circuit depth to `holographically' expand the sample size while maintaining the AC property of the distribution. For fixed number of physical qubits, it leads to an exponential amplification in the sampling size, which shares the same scaling of information in the monitored dynamics~\cite{zhang2025scaling}. On the other hand, HRCS demonstrates that a protocol for generating $n$ classical bits using only $\calO(\ln(n))$ qubits and depth of $\calO(n\ln(n))$ in 1D circuit. Meanwhile, we do not exclude any possibility to further improve this upper bound for reducing the circuit depth in each step, and a lower bound also remains open.
The logarithmic scaling of qubit number is similar to the qubit-reuse compilation method for circuit compression with applications to quantum algorithms~\cite{decross2023qubit}.}

\BZ{
We discuss the simulation complexity of HRCS here. As the circuit of HRCS at each step shares the same architecture, the complexity of computing output probabilities can be fully inferred from the unitary circuit in a single step when the circuit depth in each step is sufficiently deep. For generating $n$ classical bits with Haar-level anticoncentration, as we require $\calO(\ln n)$ physical qubits, the maximum simulation cost is just $2^{\calO(\ln n)} \sim \calO({\rm poly}(n))$, which is polynomial in the number of bits and holds for classical simulation. Note that in the one-dimensional system, Lemma~\ref{lemma_HRCS_CP} indeed suggests that linear depth with maximum entanglement is sufficient for satisfying AC condition, while the necessary circuit depth for each unitary in high-dimensional system remains open. Meanwhile, a lower bound on the minimum number of physical qubits is an open question. We expect that at least $\calO(\ln n)$ qubits number of qubits is needed, because to generate a $2$-design $n$-qubit state, we need a circuit of $\calO(\ln n)$ depth thus generating $\calO(\ln n)$ amount of entanglement. On the other hand, for given $N$ physical qubits, HRCS with linear-depth unitaries could still be challenging for classical simulation when $N$ is sufficiently large, similar to conventional RCS.
}

\BZ{We have demonstrated a verifiable metric like XEB in HRCS can be used as a rough benchmarking for the combined quality of circuit gates and mid-circuit measurements through experimental validation on IBM Quantum devices, which is essential for dynamical circuits and quantum error correction. An important direction for future work is to develop a more complete understanding of the capabilities and limitations of HRCS for benchmarking and noise learning in quantum devices with mid-circuit measurements.}

\section{Methods}
\subsection{Power sums of HRCS}
In the main text, we have mainly focused on the CP of the sampling distribution in HRCS; \BZ{Here, we extend the analyses to the power sums}.

\BZ{The $K$-th order power sum for a sampling distribution $p(x)$ is defined as}
\be
    \BZ{Z^{(K)} \coloneqq \sum_{x} p(x)^K},
\ee
where $K \ge 1$ is an integer. Note that the $K=2$ case reduces to the CP definition in Eq.~\eqref{eq:CP_def}. Therefore, the following theorem is a direct generalization of Theorem~\ref{HRCS_CP_spt} to higher-order statistical moments of the joint sampling distribution (see Appendix \ref{app:sampling_HRCS} for a proof). 
\begin{theorem}
\label{HRCS_PS_spt}
For holographic random circuit sampling with each unitary $U_t$ in $K$-design, the ensemble-averaged $K$-th power sum at step $t\ge 1$ is
    \begin{align}
       & Z_{\rm HRCS}^{(K)}(t) =\nonumber\\
       & K! d_A d_B^t \left(\frac{(d_A+K-1)!}{(d_A-1)!}\right)^{t-1}\left(\frac{(d_A d_B-1)!}{(d_A d_B + K-1)!}\right)^t,
        \label{eq:HRCS_PS_spt}
    \end{align}
    where $d_A=2^{N_A}, d_B=2^{N_B}$ are Hilbert space dimensions of the system and bath. 
    In the large-system limit $d_A \gg 1$,
    \begin{align}
        &Z_{\rm HRCS}^{(K)}(t) = Z_{\rm H}^{(K)}(\BZ{n})\nonumber\\
        &\qquad \times\exp\left[K(K-1)\frac{t\left(1-d_B^{-1}\right)+d_B^{-t}-1}{2d_A} + \calO\left(\frac{1}{d_A^2}\right)\right].
        \label{eq:HRCS_PS_spt_asymp}
    \end{align}
    Here the $Z_{\rm H}^{(K)}(\BZ{n})= K!\left(2^n\right)!/\left(2^n+K-1\right)!$ is the $K$-th power sum for \BZ{$n$}-qubit Haar random states.
\end{theorem}
It is expected that Eq.~\eqref{eq:HRCS_PS_spt} with $K=2$ reduces to the exact CP result of Eq.~\eqref{eq:HRCS_CP_spt} in Theorem~\ref{HRCS_CP_spt}, while for $t=1$, Eq.~\eqref{eq:HRCS_PS_spt} becomes the $K$-th \BZ{power sum} of Haar random states \BZ{$Z_{\rm H}^{(K)}(n)$ with $n$ qubits} (see proof in Appendix \ref{app:preliminary}).
Following the asymptotic form in Eq.~\eqref{eq:HRCS_PS_spt_asymp},
with the increase of temporal steps $t$, the $K$-th \BZ{power sum} of HRCS $Z_{\rm HRCS}^{(K)}(t)$ also deviates from the corresponding Haar value $Z_{\rm H}^{(K)}(n)$; however, the deviation is suppressed by the $\sim K^2/d_A$ factor similar to Eq.~\eqref{eq:HRCS_CP_spt_asymp}. Therefore, the asymptotic scaling for CP in Fig.~\ref{fig:concept}c also holds similarly for the $K$-th PS. In Fig.~\ref{fig:HRCS_addition}a, we verify our theory with numerical simulations (circles), and we also find the $K$-th \BZ{power sum} of HRCS closely aligns with Haar results with the corresponding dimension (dashed) \BZ{within the accessible time}.
More precisely, for $\epsilon$-approximation Haar of $K$-th \BZ{power sum} defined by $Z_{\rm HRCS}^{(K)}(t) \le (1+\epsilon)Z_{\rm H}^{(K)}(n)$, the maximum temporal step is
\be
    t \le \tau^{(K)} \coloneqq \frac{2d_A d_B \log(1+\epsilon)}{(d_B-1)K(K-1)} \BZ{\simeq 2 d_A \log(1+\epsilon)/K^2}.
    \label{eq:max_time_steps_K}
\ee
\BZ{We verify the threshold time of Eq.~\eqref{eq:max_time_steps_K} (lines) with numerical solutions (crosses) in Fig.~\ref{fig:HRCS_addition}b to confirm the exponential scaling with $N_A$. With increasing statistical order $K$, the threshold time is suppressed (light to dark).}

\BZ{
We can also extend the spacetime tradeoff analysis to the $K$-th power sums, and the question can be formulated as what are the space and time physical resources for generating $n$ classical bits satisfying the $\epsilon$-approximate $K$-th power sum, $Z^{(K)}_{\rm HRCS}(t) \le (1+\epsilon) Z^{(K)}_{\rm H}(t)$. When $K \lesssim 2^{n/2}$, we have the following approximate lower bound for number of physical qubits as (see derivations in Appendix~\ref{app:tradeoff})
\be
    N^{(K)} \ge \frac{n}{t} + \frac{t-1}{t}\log_2\left(\frac{K(K-1)t}{2\ln(1+\epsilon)}\right) + \calO\left(\frac{\ln n}{n}\right),
    \label{eq:N_t_tradeoff_K}
\ee
which is identical to the condition for CP in Eq.~\eqref{eq:N_t_tradeoff} but with an additional term of $\frac{t}{t-1}\log_2(K(K-1)/2)$. The minimum number of necessary physical qubits is
\be
    N_{\rm min}^{(K)} = N_{\rm min} + \log_2(K(K-1)/2),
    \label{eq:N_min_K}
\ee
where $N_{\rm min}$ is solved in Eq.~\eqref{eq:N_min}. The corresponding time is $t^{(K)}{}^\star = t^\star - \ln(K(K-1)/2)$. Note that the minimum bath size $N_{B,\rm min}^{(K)} = 1/\ln(2)$ is identical to the result for CP. On the other hand, when $K\gtrsim 2^{n/2}$, the number of qubits scales linearly with bit length $n$ as the corresponding time can only be a constant.
}

\BZ{
In Fig.~\ref{fig:HRCS_addition}c, the necessary number of physical qubits for achieving $\epsilon$-approximate $K$-th power sum exhibits similar behavior as the one for CP in Fig.~\ref{fig:spacetime}a, and our theory of Eq.~\eqref{eq:N_t_tradeoff_K} again characterizes the bottom envelop. For various choices of $K$, we verify the corresponding minimum bath size is always a constant independent of $K$ in Fig.~\ref{fig:spacetime}d while the minimum number of physical qubits grows logarithmically with $n$ and agrees with the theory of Eq.~\eqref{eq:N_min_K} (lines in Fig.~\ref{fig:spacetime}e). With increasing order of power sums, we require a larger number of physical qubits as expected. 
The corresponding temporal steps $t^{(K)}{}^\star$ thus scales linearly with $n$ shown in Fig.~\ref{fig:spacetime}f.
}

\begin{figure}
    \centering
    \includegraphics[width=0.45\textwidth]{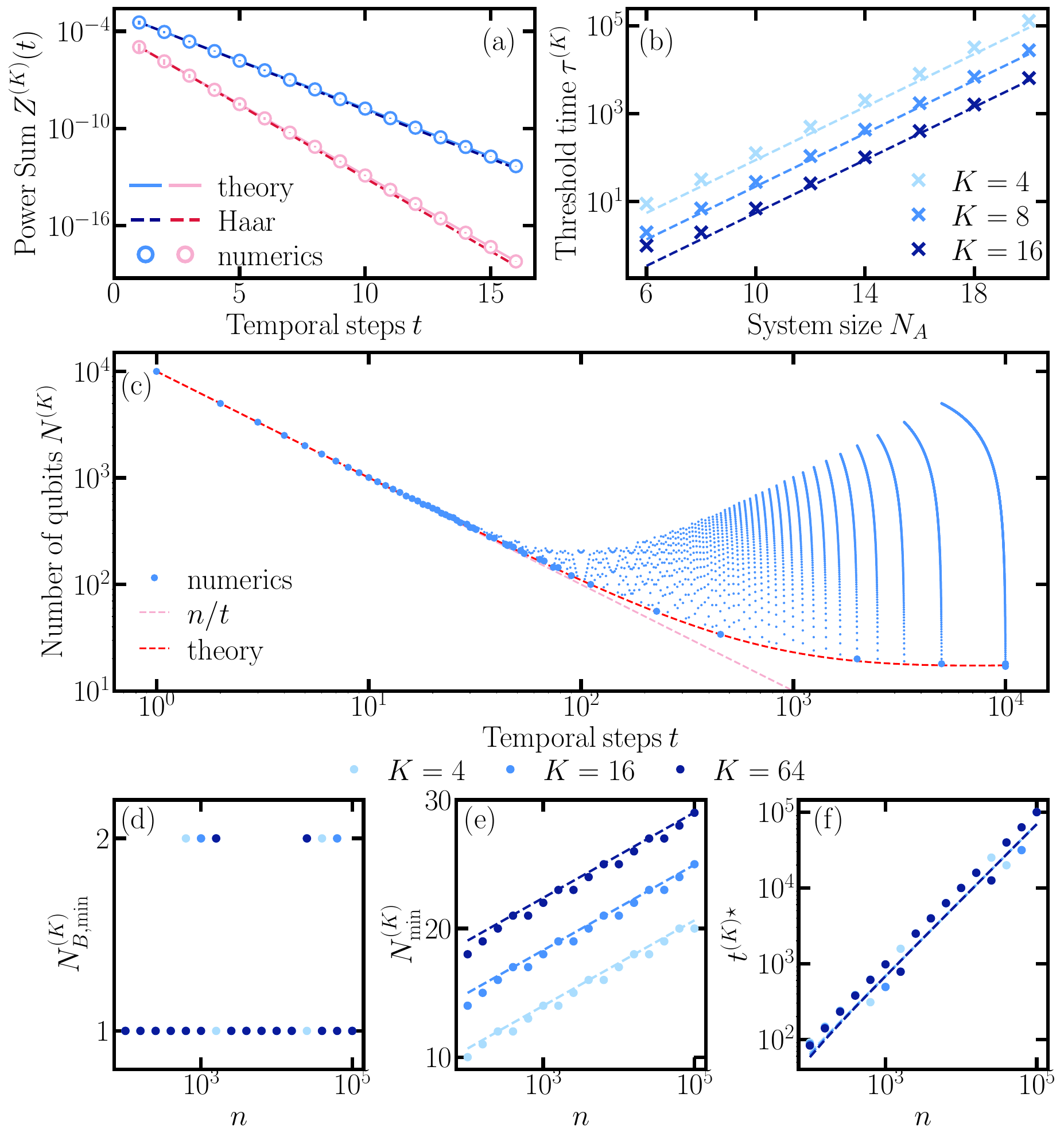}
    \caption{\BZ{{\bf Power sum for HRCS}. (a) Ensemble-averaged power sums (PS) $Z^{(K)}(t)$ for HRCS versus temporal steps in a system of $N_A = 6, N_B=1$ qubits with $K=3, 4$ (blue and pink circles). Colored solid lines represent theoretical results of Eq.~\eqref{eq:HRCS_PS_spt} in Theorem~\ref{HRCS_PS_spt}. Dark-colored dashed lines are PS for Haar random states $Z_{\rm H}^{(K)}(n)$ with $n=N_A + t N_B$. Each circle represents an average over $50$ random circuit instances. (b) Threshold time $\tau^{(K)}$ versus system size $N_A$ for $\epsilon$-approximate $K$-th Haar power sum. Crosses represent the numerical solution of $Z_{\rm HRCS}^{(K)}(t) = (1+\epsilon)Z_{\rm H}^{(K)}(n)$ with $K=4, 8, 16$ (light to dark). Dashed lines show the corresponding theoretical prediction of Eq.~\eqref{eq:max_time_steps_K}. Here we set $N_B = 4$ and $\epsilon = 1$. (c) Spacetime tradeoff for generating $n=10^4$ bits holding $\epsilon$-approximate $K=4$-th power sum. Blue dots show numerical solutions and red dashed line is the analytical result of Eq.~\eqref{eq:N_t_tradeoff_K}. Pink dashed line shows the asymptotic scaling of $n/t$. (d)-(f) Minimum bath size, number of physical qubits and corresponding temporal steps versus number of bits $n$ for different order $K$. Dashed lines in (e) and (f) represents theoretical results of Eq.~\eqref{eq:N_min_K} and $t^{(K)}{}^\star=\ln(2)n-\log_2(K(K-1)/2)$. 
    }}
    \label{fig:HRCS_addition}
\end{figure}

\ 

\subsection{Details of experiments}
We validate the XEB benchmark fidelity through experiments on the IBM quantum processor, Torino~\cite{contributors2023qiskit}. Below are additional details on the experiment.

In our experiment on real devices of $B$ circuit instances with $M$ shots for each instance, the \BZ{vanilla XEB} is estimated by
\be
    \BZ{\chi(t) = \frac{2^{n}}{BM}\sum_{j=1}^B\sum_{i=1}^M P_{C_j}(x_i) -1,}
    \label{F_XEB_exp}
\ee
where $P_{C_j}(x_i)$ is the ideal probability of obtaining bit string $x_i$ from HRCS circuit instance $C_j$. Here we simulate the ideal HRCS circuit via \texttt{TensorCircuit}~\cite{zhang2023tensorcircuit} on an Nvidia A100 GPU.

In each step of HRCS experiment, to implement a random circuit mimicking a Haar random unitary, we adopt the one-dimensional hardware-efficient ansatz (HEA). An $N$-qubit, $L$-layer HEA circuit can be represented by
\be
    U_{\rm HEA}= \prod_{\ell=1}^{L} W \otimes_{i=1}^N e^{-i\phi_{\ell,i} \hat{\sigma}^z_i/2} e^{-i\theta_{\ell,i} \hat{\sigma}^x_i/2},
\ee
where $W = \otimes_{i=1}^{\lfloor N/2\rfloor} {\rm CX}_{2i+1,2i+2} \otimes_{i=1}^{\lfloor N/2\rfloor} {\rm CX}_{2i,2i+1}$ is a fixed layer of CNOT gates on nearest neighbors in a brickwork style to build up entanglement among different qubits. $\hat{\sigma}^x_i, \hat{\sigma}^z_i$ are the Pauli-X operator and Pauli-Z operator on $i$-th qubit, respectively, and the rotation angles $\{\theta_{\ell,i}\}_{\ell,i}, \{\phi_{\ell,i}\}_{\ell,i}$ are uniformly distributed within $[0, 4\pi)$. Note that all of our gates are local to facilitate the experimental implementation. For the IBM Quantum Torino device we used, the naive gates are square root of X (SX), rotation-Z (RZ) and controlled-Z (CZ) gates, and circuit transpilation is implemented with highest level of optimization.

For the experiment in Fig.~\ref{fig:experiment}a, the $L=8$-layer HEA circuit on $N_A=5, N_B=5$ qubits involves $160$ free parameters which is much less compared to the degree of freedom for the $2^{10}$-dimensional unitary, and $880$ gates in total before transpilation. The number of gates in our transpiled circuits is presented in Appendix \ref{app:experiment}. 
For the maximum of $t=16$ in this experiment, we use $5049$ SX, $3684$ RZ and $1152$ CZ gates.
Due to the automatic computing resource allocation from IBM Quantum service, for experiments with temporal steps $t\le 10$, we utilize qubits in IBMQ Torino shown in dark grey circles in inset of Fig.~\ref{fig:experiment}a, and for $t=13, 16$, we utilize slightly different qubits (see Supplementary Fig.~\ref{fig:circuit_5_5}). In Table~\ref{tab:ibm_error}, we list the average error rates of all gates in our experiment and the readout \BZ{(measurement)} errors for $t\le 10$. The other case can be found in Appendix \ref{app:experiment}.

\begin{table}[t]
    \centering
    \begin{tabular}{||c||c||}
    \hline
    Error type &  Average error rate\\
    \hline
    \hline
    SX gate &  $(3.07 \pm 1.27)\times 10^{-4}$ \\
    CZ gate &  $(2.64 \pm 0.672)\times \BZ{10^{-3}}$ \\
    Readout & $0.0214 \pm 0.0148$ \\
    \hline
    \end{tabular}
    \caption{Average error rates of the operations involved in our experiment on IBM Quantum Torino. As IBM Quantum implements RZ gate virtually in hardware via frame changes, it is error free. Calibration data are provided by IBM Quantum.}
    \label{tab:ibm_error}
\end{table}

\BZ{In the larger-scale experiment of Fig.~\ref{fig:experiment}b with 20 physical qubits, we consider the patch technique widely adopted in RCS experiments~\cite{arute2019quantum} to mitigate the device noise}. Specifically, we adopt two patches of quantum circuits, each one can be classically simulated. In each patch, \BZ{we use $N_A^\prime = 5$ and $N_B^\prime = 5$ qubits for the system and bath}; In each step, the unitary is implemented via an $8$-layer one-dimensional hardware-efficient ansatz, which is identical to the setup in \BZ{the prior experiment of Fig.~\ref{fig:experiment}a}. Therefore, we have twice the number of gates compared to Fig.~\ref{fig:experiment}a. For the maximum of $t= 19$ in this experiment, we use $8748$ SX, $11988$ RZ and $2736$ CZ gates with the number of gates for all circuits presented in Appendix \ref{app:experiment}. Due to the allocation of qubit resources from Qiskit, the used qubits in each circuit instance of this experiments are slightly different from each other, however, we still expect a similar average error rate as shown in Table~\ref{tab:ibm_error}. As we use the patch circuit method, both the ideal and noisy theory of XEB follows the single patch result, as detailed in Appendix \ref{app:noisy_theory}. 

\

{\noindent\bfseries \Large Acknowledgments.} \\
\BZ{B.Z. thanks Xun Gao and Zhi-Yuan Wei for helpful discussions.}
Q.Z. and B.Z. acknowledge support from NSF (CCF-2240641, OMA-2326746, 2350153), ONR N00014-23-1-2296, AFOSR MURI FA9550-24-1-0349 and DARPA (HR0011-24-9-0362, HR00112490453, D24AC00153-02). The experiment was conducted using IBM Quantum Systems provided
through USC’s IBM Quantum Innovation Center.

%

\

\clearpage

\appendix

\begin{widetext}

\section{Quantum circuit sampling and anticoncentration}

In this section, we introduce the circuit setup of holographic random circuit sampling (HRCS) and review the anticoncentration in quantum circuit sampling.

\subsection{Setup of holographic random circuit sampling}

We begin with the conventional random circuit sampling (RCS) that has been adopted to demonstrate quantum advantage~\cite{arute2019quantum, wu2021strong} (see Fig.~\ref{fig:concept}a in the main text).
Given a unitary quantum circuit $U \in \calU(2^N)$, and a trivial $N$-qubit initial state $\ket{\bm 0} = \ket{0}^{\otimes N}$, we can generally write out its output state from the circuit as
\be
    \ket{\psi} = U \ket{\bm 0}.
\ee
One then performs a measurement on the output state $\ket{\psi}$ in a fixed set of basis (i.e. computational basis), leading to the probability distribution
\be
p_U(x) = |\braket{x|U|\bm 0}|^2.
\label{pU:RCS}
\ee
We review the analyses of conventional RCS in Appendix~\ref{app:preliminary}.

In HRCS, a quantum system of $N_A$ qubits (in initial state $\ket{\bm 0}_A$) repetitively interacts with a bath with $N_B$ qubits (in initial state $\ket{\bm 0}_B$) via unitary circuits $U_1, U_2,\cdots, U_t$. After the action of each unitary $U_k$, one performs a computational basis measurement on the bath, leading to the outcome $z_k$. For simplicity of description and theoretical analysis, we take the reset of the bath qubits back to $\ket{\bm 0}_B$ states before the next round of unitary interaction $U_{k+1}$, although our main results hold regardless reset or not (see proof in \ref{app:sampling_HRCS}). After the final unitary $U_t$, one measures the $N_A$ system qubits in the computational basis to obtain the result $x(t)$, besides the bath measurement outcome $z_t$. Collecting all temporal measurement results $\bmz(t)= (z_1,\cdots,z_t)$  and the final `spatial' measurement results $x(t)$, the joint probability of the spatiotemporal sampling
\be
p_{\bm U}\left[\bm z\left(t\right), x\left(t\right)\right]=
|\bra{z_1}_{B_1}\cdots \bra{z_t}_{B_t}\bra{x(t)}_{A} U_t\cdots U_1 \ket{\bm 0}_A \ket{\bm 0}_{B_1}\cdots\ket{\bm 0}_{B_t}|^2,
\label{pU:HRCS}
\ee
where we denote the unitary sequence in HRCS as ${\bm U} = (U_1, \cdots, U_t)$ and the bath system after reset in the $k$-th time step as $B_k$.
Comparing the statistics in Eq.~\eqref{pU:RCS} and Eq.~\eqref{pU:HRCS}, the conventional RCS can be regarded as the $t=1$ special case of HRCS. Indeed, the length of the measurement outcome in HRCS, $\bm z\left(t\right)$ and $x\left(t\right)$, increases with the time steps $t$ as $n(t) \coloneqq tN_B+N_A$ due to the reuse of the bath system. Although the probability of HRCS in Eq.~\eqref{pU:HRCS} seems to be a generalization of the probability for conventional RCS in Eq.~\eqref{pU:RCS}, it is indeed a non-unitary circuit due to mid-circuit measurements, shown in Fig.~\ref{fig:concept}b in the main text. We present our analyses of HRCS in \ref{app:sampling_HRCS}.

As a comparison, one can also consider sampling from the marginal distributions, as we detail in \ref{app:marginal_sampling}. In the case of spatial sampling, one samples $x(t)$ from the marginal distribution $p_{\bm U}\left[ x\left(t\right)\right]$; In the case of temporal sampling, one samples only the mid-circuit measurement outcome $\bmz(t)$ following the marginal distribution $p_{\bm U}[\bmz(t)]$.

\subsection{Anticoncentration}
\label{app:AC}

The random circuit sampling problem asks how hard it is to classically simulate a distribution $p_C(\cdot)$ generated from a quantum device $C$.
For a typical state, when the distribution $p_C(x)$ is widely supported over all computational basis $\ket{x}$, but not yet uniform over all basis, the sampling is expected to be a hard task. On the other hand, if the state is only supported on a limited number of basis, classical simulation could become efficient. Therefore, to characterize the spreading of a given discrete distribution, we consider the {\em anticoncentration} (AC) property, which is quantified by the collision probability (CP)
\be
    Z_C = \sum_x p_C(x)^2.
\ee
As a standard metric in the context of statistics, CP has already been widely utilized in the study of RCS, such as Refs.~\cite{dalzell2022random, dalzell2024random, fefferman2024effect}. When $Z_C$ is exponentially small in system size, the distribution is closer to uniform and therefore anticoncentrated; on the other hand, when $Z_C=1$, the distribution is entirely concentrated on a single outcome. Indeed, CP also posts both lower and upper bounds on the maximum probability in the distribution by $Z_C \le \max_{x} p_C(x) \le \sqrt{Z_C}$. 
Note that CP is equivalent to the variance given a fixed mean, and it can also be connected to the R\'enyi-2 entropy $H_2[p_C(\cdot)]$ of the distribution via $Z_C = e^{-H_2[p_C(\cdot)]}$. In this regard, we also introduce higher-order statistical moments to characterize the distribution---the higher-order power sums (PS)
\be 
Z^{(K)}_{C} = \sum_{x} p_C(x)^K.
\label{eq:PS_def}
\ee
In principle, the full information about the discrete-variable distribution $p_C(\cdot)$ on $D$-dimensional support can be obtained from the collection $\left\{Z_C^{(K)}\right\}_{K=1}^D$ via the Newton's identities up to a permutation on its supports.


In an RCS experiment, there is no preference on the choice of each quantum gate in the circuit. Therefore, we focus on the average-case performance of AC in this work, which is characterized by the ensemble-averaged collision probability as
\be
    Z = \E_{C \in \calC}[Z_C],
\ee
where $\calC$ represents the ensemble of quantum circuit with a fixed architecture. Similarly, we also consider the ensemble-averaged PS $Z^{(K)}=\E_{C \in \calC}[Z^{(K)}_{C}]$. When 
the dimension is large, we expect a small variance in $Z_C$ among different circuit realizations, and one can also bound the CP in a typical circuit by the Markov inequality $\Prob(Z_C \ge \alpha Z) \le 1/\alpha$.

For any discrete distribution $p(\cdot)$ supported on a $D$-dimensional space, the CP is lower bounded by the value of the uniform distribution, $Z \ge Z_{\rm uni} = 1/D$, which can also be regarded as the most anti-concentrated case, despite easy to sample classically. The distribution is regarded as ``anticoncentrated'' when $Z/Z_{\rm uni} \le c$ with $c$ to be a constant independent of the dimension $D$. In contrast, if $Z/Z_{\rm uni} = f(D)$ is an unbounded function depending on the dimension, it becomes ``non-anticoncentrated'' or lack of anticoncentration. The above criteria is also adopted in the study of RCS~\cite{dalzell2022random, dalzell2024random}.
Furthermore, the constraint on the CP can also be converted into the pointwise constraint by the Paley-Zygmund inequality, where AC with respect to $Z$ also implies the AC of probability for each measurement outcome $\Prob[p(\bm x) \ge \alpha/D] \ge (1-\alpha)^2/(DZ)$ with $\alpha \in [0, 1]$, indicating a sufficient strong condition~\cite{fefferman2024effect}.





\section{Preliminary}
\label{app:preliminary}

In this section, we introduce the notations in Haar random unitary integration used in derivations of this work and take Haar random states in conventional RCS as a warm-up exercise.

\subsection{Identities for Haar integral}

We first introduce some important identities utilized in the following derivations of CP and PS in both Haar random states and HRCS.
Let us begin with the doubled Hilbert space representation of operators. For an arbitrary operator $\hat{A}$, we define the doubled Hilbert space representation by 
\be
    \hat{A} = \sum_{i,j} \ketbra{i}{j} \braket{i|\hat{A}|j} \to \kett{\hat{A}} = \sum_{i,j} \braket{i|\hat{A}|j} \ket{i} \ket{j}.
\ee
Therefore, when operator $\hat{A}$ undergoing twirling by unitary $U$, we can write it in the doubled Hilbert space representation as
\be
    U \hat{A} U^\dagger \rightarrow \left[U \otimes U^*\right] \kett{\hat{A}},
    \label{eq:twirling_dual}
\ee
where $U^*$ denotes the complex conjugate of $U$. This approach has also be utilized in prior works~\cite{zhang2025holographic, zhang2025scaling}. With this notion, we can write the Hilbert-Schmidt inner product of two operators $\hat{A}$ and $\hat{B}$ simply as
\be
    \tr(\hat{A}^\dagger \hat{B}) = \bbrakett{\hat{A}}{\hat{B}}.
\ee
With the doubled Hilbert space representation, we can now introduce the identity for Haar unitary twirling as
\be
    \E_{\rm Haar}\left[\left(U \otimes U^*\right)^{\otimes K}\right] = \sum_{\sigma, \pi \in S_K} {\rm Wg}_{d}(\sigma^{-1}\pi) \kett{\hat{\sigma}}\bbra{\hat{\pi}}, \label{eq:haar_twirling}
\ee
where $\sigma, \pi$ are permutations and $S_K$ is the permutation group of $K$ entries. Here $ {\rm Wg}_d(\sigma^{-1}\pi)$ is the Weingarten function, defined as ${\rm Wg}_d(\sigma^{-1}\pi) = \left(d^{|\sigma^{-1}\pi|}\right)^{-1}$, where $|\sigma^{-1}\pi|$ denotes the number of cycles of the permutation $\sigma^{-1}\pi$. Specifically, in the case of $K=2$, we have $S_2 = \{e, \tau\}$ consisting of identity $e$ and swap $\tau$. The above twirling identity can be explicitly written out as
\be
    \E_{\rm Haar}\left[\left(U \otimes U^*\right)^{\otimes 2}\right] = \frac{1}{d^2-1}\left(\kett{\hat{e}}\bbra{\hat{e}}+\kett{\hat{\tau}}\bbra{\hat{\tau}}\right) - \frac{1}{d(d^2-1)} \left(\kett{\hat{e}}\bbra{\hat{\tau}}+\kett{\hat{\tau}}\bbra{\hat{e}}\right).
    \label{eq:haar_twirling_K2}
\ee

In the end, we introduce another identity to be utilized in the evaluation of PS.
\be
    \sum_{\sigma \in S_K} {\rm Wg}_d(\sigma) = \frac{(d-1)!}{(K+d-1)!}.
    \label{eq:W_sum}
\ee
We also append a proof of Eq.~\eqref{eq:W_sum} here. 

\begin{proof}
    We first utilize the Haar twirling identity in Eq.~\eqref{eq:haar_twirling} and sandwitch it by $\bbra{\bm 0}$ and $\kett{\bm 0}$ to obtain
    \begin{align}
        \bbra{\bm 0}\E_{\rm Haar}\left[\left(U \otimes U^*\right)^{\otimes K}\right]\kett{\bm 0} &= \sum_{\sigma, \pi \in S_K} {\rm Wg}_{d}(\sigma^{-1}\pi) \bbrakett{\bm 0}{\hat{\sigma}}\bbrakett{\hat{\pi}}{\bm 0} \\
        &= \sum_{\sigma, \pi \in S_K} {\rm Wg}_{d}(\sigma^{-1}\pi)\\
        &= K! \sum_{\sigma \in S_K} {\rm Wg}_{d}(\sigma),
        \label{eq:W_sum_1}
    \end{align}
    where in the last line we utilize the property that $S_K$ is closed and thus rename the dummy variable. Recall that we can also write $\E_{\rm Haar}\left[\left(U \otimes U^*\right)^{\otimes K}\right]\kett{\bm 0}$ as the $K$-th moment of Haar random states ensemble,
    \begin{align}
        \bbra{\bm 0}\E_{\rm Haar}\left[\left(U \otimes U^*\right)^{\otimes K}\right]\kett{\bm 0} &= \tr\left(\ketbra{\bm 0}{\bm  0}^{\otimes K} \int_{\rm Haar} \diff U \left(U\ketbra{\bm 0}{\bm 0}U^\dagger\right)^{\otimes K}\right)\\
        &= \frac{(d-1)!}{(K+d-1)!}\sum_{\pi \in S_K} \tr\left(\ketbra{\bm 0}{\bm  0}^{\otimes K} \hat{\pi}\right)\\
        &=  \frac{K!(d-1)!}{(K+d-1)!},
        \label{eq:W_sum_2}
    \end{align}
    where the second line follows the definition of $K$-th moment of the Haar ensemble (see Refs.~\cite{cotler2023emergent, ho2022exact, zhang2025holographic}). Comparing Eqs.~\eqref{eq:W_sum_1} and~\eqref{eq:W_sum_2}, we can easily find
    \be
        \sum_{\sigma \in S_K} {\rm Wg}_d(\sigma) = \frac{(d-1)!}{(K+d-1)!}.
    \ee
\end{proof}

\subsection{Sampling from Haar random states}

In this section, we review the CP of Haar random states, studied in Ref.~\cite{boixo2018characterizing, dalzell2022random}, and extend to the $K$-th PS.
\begin{lemma}
\label{lemma:Haar_CP}
    The ensemble-averaged collision probability for sampling in $N$-qubit Haar random states is
    \be
        Z_{\rm H}(N) = \frac{2}{2^N + 1}.
        \label{eq:CP_haar}
    \ee
\end{lemma}
\begin{proof}
    Here we provide two different approaches to prove Lemma~\ref{lemma:Haar_CP} utilizing the Haar twirling identity and probability of probability (PoP) separately.
    
    The ensemble averaged probability of measurement outcome $x$ in Haar random states can be evaluated as
    \begin{align}
        &\E_{U\in \rm Haar}[p_U(x)^2] = \E_{U\in \rm Haar}[|\braket{x|U|\bm 0}|^4]\nonumber\\
        &= 
        \E_{U\in \rm Haar}\left[\tr(U^{\otimes 2}\ketbra{\bm 0}{\bm 0}^{\otimes 2} U^\dagger{}^{\otimes 2}\ketbra{x}{x}^{\otimes 2})\right]\\
        &= \frac{1}{d^2-1}\left(\bbrakett{x}{\hat{e}}\bbrakett{\hat{e}}{\bm 0}+\bbrakett{x}{\hat{\tau}}\bbrakett{\hat{\tau}}{\bm 0}\right) - \frac{1}{d(d^2-1)} \left(\bbrakett{x}{\hat{e}}\bbrakett{\hat{\tau}}{\bm 0}+\bbrakett{x}{\hat{\tau}}\bbrakett{\hat{e}}{\bm 0}\right)\\
        &= \frac{2}{d^2-1}- \frac{2}{d(d^2-1)}
        = \frac{2}{d(d+1)},
    \end{align}
    where we utilize the twirling identity of Eq.~\eqref{eq:haar_twirling_K2} in the third line and $d = 2^N$ is the Hilbert space dimension. Here $\kett{\bm 0}, \bbra{x}$ are the doubled Hilbert space representation of $\ketbra{\bm 0}{\bm 0}^{\otimes 2}$ and $\ketbra{x}{x}^{\otimes 2}$, and $\bbrakett{x}{\hat{\sigma}} = 1$ for arbitrary $x$ and $\sigma \in S_2$. Then following the definition, the ensemble-averaged collision probability is
    \be
        Z_{\rm H}(N) = \sum_x \E_{U\in \rm Haar}[p_U(x)^2] = \frac{2}{d+1},
    \ee
    as there are in total $d=2^N$ different measurement outcomes. 
    
    Another method for proof follows from the known fact that the PoP for Haar random states follows Porter-Thomas (PT) distribution~\cite{boixo2018characterizing}
    \be
        P\left(p_{\rm H}(x) = p\right) = (d-1)(1-p)^{d-2},
        \label{eq:PT_dist}
    \ee
    then the CP is simply the second moment of the PT distribution up to a constant multiplier as
    \be
        Z_{\rm H}(N) = d\E_{U\in \rm Haar}[p_U(\bm x)^2] = d\int_0^1 \diff p P\left(p_{\rm H}(x) = p\right)p^2 = \frac{2}{d+1}.
    \ee
\end{proof}

As $Z_{\rm H}(N) = 2/2^N + \calO(1/4^N)$ in the large limit of $N \gg 1$, the measurement distribution is clearly anticoncentrated. We also would like to point out that it indeed only requires unitary of $2$-design to achieve it. 

Next, we extend to the $K$-th PS of full sampling in Haar random states.
\begin{lemma}
\label{lemma:Haar_Kth}
    The ensemble-averaged $K$-th power sum in the sampling of $N$-qubit Haar random states is
    \be
        Z_{\rm H}^{(K)}(N) = \frac{K!d!}{(d+K-1)!},
        \label{eq:haar_PS}
    \ee
    where $d=2^N$ is the Hilbert space dimension.
\end{lemma}
The proof is similar to the one for CP.

\begin{proof}
    We can utilize the Haar twirling identity in Eq.~\eqref{eq:haar_twirling} to evaluate the $K$-th power of measurement probability
    \begin{align}
        &\E_{U\in \rm Haar}[p_U(x)^K] = \E_{U\in \rm Haar}\left[|\braket{x|U|\bm 0}|^{2K}\right] \nonumber \\
        &= 
        \E_{U\in \rm Haar}\left[\tr(U^{\otimes K}\ketbra{\bm 0}{\bm 0}^{\otimes K} U^\dagger{}^{\otimes K}\ketbra{x}{x}^{\otimes K})\right]\\
        &= \sum_{\sigma, \pi \in S_K} {\rm Wg}_{d}(\sigma^{-1}\pi) \bbrakett{x}{\hat{\sigma}}\bbrakett{\hat{\pi}}{\bm 0} \\
        &= \sum_{\sigma, \pi \in S_K} {\rm Wg}_d(\sigma^{-1}\pi)\\
        &= \frac{K!(d-1)!}{(d+K-1)!},
    \end{align}
    where we utilize the identity of sum of Weingarten functions in Eq.~\eqref{eq:W_sum}.
    The $K$-th PS then becomes
    \be
        Z_{\rm H}^{(K)}(N) = \sum_x \E_{U\in \rm Haar}[p_U(x)^K] = d \frac{K!(d-1)!}{(d+K-1)!} = \frac{K!d!}{(d+K-1)!}.
    \ee
    
    We can also prove it utilizing the PoP directly.
    \begin{align}
        Z_{\rm H}^{(K)}(N) = d\E_{U\in \rm Haar}[p_U(\bm x)^K] = d\int_0^1 \diff p P\left(p_{\rm H}(x) = p\right)p^K = \frac{K!d!}{(d+K-1)!}.
    \end{align}

\end{proof}

\subsection{Subsystem sampling in Haar random states}

Finally, we show the sampling on a subsystem of Haar random states.

\begin{lemma}
\label{lemma:sub_dist_haar}
    For a Haar random quantum state $\ket{\psi}_{AB}$ in a bipartite system $\calH = \calH_A \otimes \calH_B$ with $N = N_A + N_B$ qubits, the marginal sampling distribution $p_{\rm Haar}(x_B) = |{}_B\braket{x_B|\psi}_{AB}|^2$ follows the Beta distribution 
    \be
        \Prob[p_{\rm H}(x_B) = p] = {\rm Beta}(d_A, (d_B-1)d_A),
        \label{eq:sub_dist_haar}
    \ee
    where $d_A = 2^{N_A}, d_B = 2^{N_B}$ are the dimension of corresponding subsystems.
\end{lemma}
Note that in the limit of full sampling $N_A\to 0$ thus $d_A \to 1$, the above Beta distribution converges to the Porter-Thomas distribution of $d = d_B$.

\begin{proof}
    A Haar random state $\ket{\psi}_{AB}$ can be written as $\ket{\psi}_{AB} \propto \sum_{i_A, i_B} X_{i_A, i_B} \ket{i_A}_A \ket{i_B}_B$, where the unnormalized amplitude $X_{i_A, i_B} \sim \calC\calN(0, \bI)$ follows a standard complex normal distribution, and thus the $d_A\times d_B$ random matrix $X$ follows complex Ginibre ensemble. After tracing out the subsystem $A$, we have the reduced state as
    \begin{align}
        \rho_A &= \sum_{k_A} {}_A\braket{k_A|\psi}\braket{\psi|k_A}_A \nonumber\\
        &\propto \sum_{i_B,j_B} \sum_{k_A} X_{k_A, i_B} X^*_{k_A,j_B} \ketbra{i_B}{j_B}_B = \frac{X^T X^*}{\tr(X^T X^*)},
    \end{align}
    which is a fixed traced Wishardt random matrix ensemble. Its diagonal elements form a normalized vector $\bm s= \left(s_1, \dots, s_{d_B}\right)^T$ where $s_i = \sum_{i_A} |X_{i_A,i}|^2 \sim {\rm Gamma}(d_A, 1)$ follows Gamma distribution. Then the normalized vector ${\bm s}/(\sum_{i} s_i)$ follows the Dirichlet distribution as
    \be
        \left(\frac{s_1}{\sum_i s_i}, \dots, \frac{s_{d_B}}{\sum_i s_i}\right)^T \sim {\rm Dir}(d_A, \cdots, d_A).
    \ee
    Note that the normalized vector corresponds to the diagonal entries of the reduced state $\rho_A$, and thus the probability of each sampling outcome. Therefore, the distribution of marginal sampling distribution is
    \be
        \Prob[p_{\rm H}(x_B) = p] = {\rm Beta}(d_A, (d_B-1)d_A).
    \ee
\end{proof}

\begin{corollary}
\label{corollary:haar_marginal}
    The ensemble-averaged collision probability of the marginal sampling distribution $p_{\rm H}(x_B)$ of Haar random states is
    \be
        Z_{{\rm H},B}(N_A, N_B) = \E_{U\in \rm Haar}\left[\sum_{x_B} p_{U}(x_B)^2\right] = \frac{d_A+1}{d_A d_B + 1}.
        \label{eq:CP_haar_sub}
    \ee
\end{corollary}
\begin{proof}
We can prove the CP by Haar twirling identity in Eq.~\eqref{eq:haar_twirling_K2} as
\begin{align}
    Z_{{\rm H},B}(N_A, N_B) &= d_B \E_{U\in \rm Haar}\left[p_U(x_A)^2\right]\\
    &= d_B\E_{U\in \rm Haar}\left[\tr(U^{\otimes 2}\ketbra{\bm 0}{\bm 0}^{\otimes 2} U^\dagger{}^{\otimes 2}\bI_A^{\otimes 2} \otimes \ketbra{x_B}{x_B}_B^{\otimes 2} )\right]\\
    &= \frac{d_B}{d^2-1}\left(\bbrakett{x_B}{\hat{e}}_B \bbrakett{\hat{e}}{\hat{e}}_A \bbrakett{\hat{e}}{\bm 0}+\bbrakett{x_B}{\hat{\tau}}_B\bbrakett{\hat{e}}{\hat{\tau}}_A \bbrakett{\hat{\tau}}{\bm 0}\right) \nonumber\\
    &\quad - \frac{d_B}{d(d^2-1)} \left(\bbrakett{x_B}{\hat{e}}_B \bbrakett{\hat{e}}{\hat{e}}_A \bbrakett{\hat{\tau}}{\bm 0}+\bbrakett{x_B}{\hat{\tau}}_B \bbrakett{\hat{e}}{\hat{\tau}}_A \bbrakett{\hat{e}}{\bm 0}\right)\\
    &= \frac{d_B}{d^2-1}\left(d_A^2 + d_A\right) - \frac{d_B}{d(d^2-1)} \left( d_A^2 + d_A\right)
    = \frac{d_A+1}{d_A d_B + 1}.
\end{align}

The CP can also be proved utilizing the PoP from Lemma~\ref{lemma:sub_dist_haar} as
\be
    Z_{{\rm H},B}(N_A, N_B) = d_B\int_{0}^1 \diff p P(p_{\rm H}(x_B)=p) p^2 = \frac{d_A+1}{d_A d_B + 1}.
\ee
\end{proof}

One can easily show that the CP in Eq.~\eqref{eq:CP_haar_sub} is strictly upper bounded by $Z_{\rm H}$ of dimension $d_B$ as $Z_{{\rm H}, B}(N_A, N_B) < Z_{\rm H}(N_B)$ given a nontrivial subsystem $A$ with $d_A>1$. In large system size limit of $d_A \gg 1$, the CP for marginal sampling $Z_{{\rm H},B}(N_A, N_B) = Z_{\rm uni}(N_B)(1+1/d_A)$ converges to the uniform case of dimension $2^{N_B}$, up to finite-size corrections. The convergence of the marginal sampling of a subsystem in Haar random states toward the uniform distribution also hints a necessary condition for the emergent state design in random states generated via projective measurements on another subsystem, known as {\em deep thermalization}~\cite{ho2022exact, cotler2023emergent}.


\section{Spatiotemporal sampling in HRCS}
\label{app:sampling_HRCS}
In this section, we provide the proof to the theorems in the main text. Additionally, we present some numerical simulation results to verify the theorems. Before we start the proof, we first restate the theorems of the main text here for reader's convenience. The proof relies on an equivalent expansion of bath system at different temporal steps into a joint enlarged bath and postpone all mid-circuit measurements to the end, illustrated in Supplementary Fig.~\ref{fig:proof}a. For simplicity, we begin with the case where qubit reset is performed following every mid-circuit measurement; in the end, we show that the results remain the same regardless of reset of not reset.

\subsection{Proof of Theorem~\ref{HRCS_CP_spt}}

\begin{theorem}
\label{HRCS_CP_spt_app}(Theorem~\ref{HRCS_CP_spt} in main text)
    For holographic random circuit sampling with each unitary $U_t$ in $2$-design, the ensemble-averaged collision probability at step $t\ge 1$ is
    \be
        Z_{\rm HRCS}(t) = \frac{2(d_A+1)^{t-1}}{ (1+d_A d_B)^t},
        \label{eq:HRCS_CP_spt_app}
    \ee
    where $d_A = 2^{N_A}, d_B = 2^{N_B}$ are Hilbert space dimensions of the system and bath. In the large-system limit $d_A \gg 1$,
    \begin{align}
        &Z_{\rm HRCS}(t) 
        = Z_{\rm H}(n)\exp\left[\frac{t\left(1-d_B^{-1}\right) + d_B^{-t}-1}{d_A} + \calO\left(\frac{1}{d_A^2}\right)\right],
        \label{eq:HRCS_CP_spt_asymp_app}
    \end{align}
    where $n = N_A + t N_B$ is the total number of sampling bits.
\end{theorem}

\begin{proof}
    From the schematic in Supplementary Fig.~\ref{fig:proof}a, we first write out the joint probability of sampling $\bmz(t), x(t)$ in the output state as
    \be
        P_U[\bmz(t), x(t)] = \tr\left(U\ketbra{\bm 0}{\bm 0}_{A\bm B} U^\dagger \ketbra{x(t), \bmz(t)}{x(t), \bmz(t)}_{A\bm B}\right),
        \label{eq:pr_spt}
    \ee
    where $\bm B = (B_1, \dots, B_t)$ is the equivalent expanded bath system and $\ket{\bm 0}_{A \bm B} = \ket{\bm 0_A, \bm 0_{B_1},\dots, \bm 0_{B_t}}$ is the initial state of the whole state. Here we denote $U = \prod_{i=1}^t U_i$ as the global unitary of the equivalent circuit. The ensemble average of the squared probability becomes
    \begin{align}
        \E_{U\in \rm Haar}\left[P_U[\bmz(t), x(t)]^2\right] &= \E_{U\in \rm Haar} \tr\left(U^{\otimes 2}\ketbra{\bm 0}{\bm 0}_{A\bm B}^{\otimes 2} U^\dagger{}^{\otimes 2} \ketbra{x(t), \bmz(t)}{x(t), \bmz(t)}_{A\bm B}^{\otimes 2}\right)\\
        &= \bbra{x(t),\bmz(t)}_{A \bm B} \E_{U\in \rm Haar}\left[U^{\otimes 2} \otimes U^*{}^{\otimes 2}\right] \kett{\bm 0}_{A\bm B},
        \label{eq:pr_spt_square}
    \end{align}
    where we utilize the doubled Hilbert space representation in the second line and omit the $(\cdot)^{\otimes 2}$ in the doubled bra and ket.
    Now we evaluate Eq.~\eqref{eq:pr_spt_square} step by step from $U_1$ to $U_t$.

    For the first step involving $U_1$, we have
    \begin{align}
        &\bbra{z_1(t)}_{B_1} \E_{U_1\in \rm Haar}\left[U_1^{\otimes 2} \otimes U_1^*{}^{\otimes 2}\right] \kett{\bm 0}_{AB_1}\nonumber\\
        &= \frac{1}{d^2-1}\left(\bbrakett{z_1(t)}{\hat{e}}_{B_1}\bbrakett{\hat{e}}{\bm 0}_{AB_1}\kett{\hat{e}}_A +\bbrakett{z_1(t)}{\hat{\tau}}_{B_1}\bbrakett{\hat{\tau}}{\bm 0}_{AB_1}\kett{\hat{\tau}}_{A}\right) \nonumber\\
        &\quad - \frac{1}{d(d^2-1)} \left(\bbrakett{z_1(t)}{\hat{e}}_{B_1}\bbrakett{\hat{\tau}}{\bm 0}_{AB_1}\kett{\hat{e}}_A +\bbrakett{z_1(t)}{\hat{\tau}}_{B_1}\bbrakett{\hat{e}}{\bm 0}_{AB_1} \kett{\hat{\tau}}_A \right)\\
        &= \frac{1}{d^2-1}\left(\kett{\hat{e}}_A +\kett{\hat{\tau}}_{A}\right) - \frac{1}{d(d^2-1)} \left(\kett{\hat{e}}_A +\kett{\hat{\tau}}_{A}\right)\\
        &= \frac{1}{d(d+1)} \left(\kett{\hat{e}}_A +\kett{\hat{\tau}}_{A}\right).
        \label{eq:pr_spt_square_eq0}
    \end{align}
    As both permutations $e, \tau$ from $S_2$ appears, we next focus on the evolution of each of them separately under Haar random unitary.
\begin{figure}[t]
    \centering
    \includegraphics[width=\textwidth]{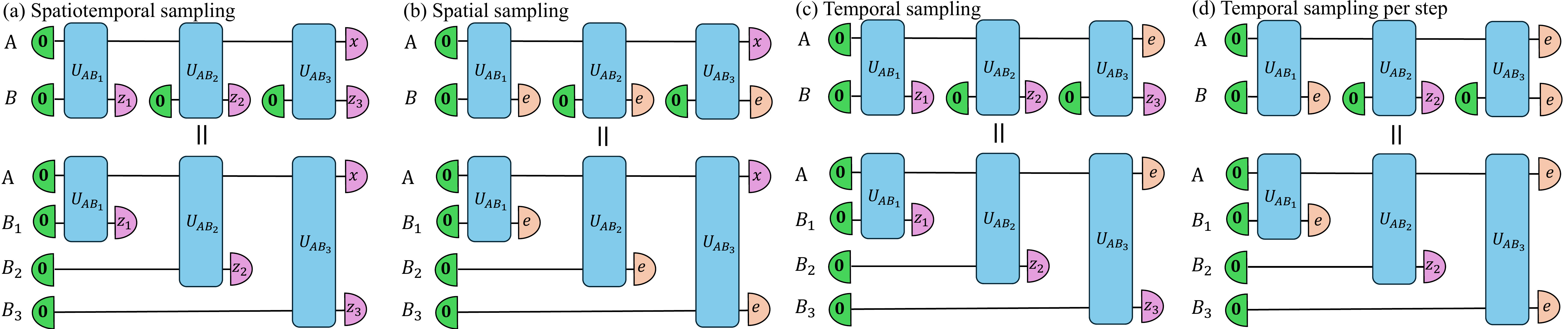}
    \caption{Schematic of the equivalent expansion of bath systems in HRCS for different sampling. Here we show an example of $t=3$. In (d), we show the schematic for temporal sampling at the second step, with distribution $p[z_2]$.}
    \label{fig:proof}
\end{figure}
    \begin{align}
        &\bbra{z_2(t)}_{B_2} \E_{U_2\in \rm Haar}\left[U_2^{\otimes 2} \otimes U_2^*{}^{\otimes 2}\right] \kett{\hat{e}}_A \kett{\bm 0}_{B_2} \nonumber \\
        &= \frac{1}{d^2-1}\left(\bbrakett{z_2(t)}{\hat{e}}_{B_2} \bbrakett{\hat{e}}{\hat{e}}_A \bbrakett{\hat{e}}{\bm 0}_{B_2}\kett{\hat{e}}_A + \bbrakett{z_2(t)}{\hat{\tau}}_{B_2} \bbrakett{\hat{\tau}}{\hat{e}}_{A}\bbrakett{\hat{\tau}}{\bm 0}_{B_2}\kett{\hat{\tau}}_{A}\right) \nonumber\\
        &\quad - \frac{1}{d(d^2-1)} \left(\bbrakett{z_2(t)}{\hat{e}}_{B_2}\bbrakett{\hat{\tau}}{\hat{e}}_{A}\bbrakett{\hat{\tau}}{\bm 0}_{B_2}\kett{\hat{e}}_A +\bbrakett{z_2(t)}{\hat{\tau}}_{B_2}\bbrakett{\hat{e}}{\hat{e}}_{A}\bbrakett{\hat{e}}{\bm 0}_{B_2} \kett{\hat{\tau}}_A \right)\\
        &= \frac{1}{d^2-1}\left( d_A^2 \kett{\hat{e}}_A + d_A\kett{\hat{\tau}}_{A}\right) - \frac{1}{d(d^2-1)} \left(d_A\kett{\hat{e}}_A +d_A^2\kett{\hat{\tau}}_A \right)\\
        &= \frac{d_A^2 d_B - 1}{d_B\left(d_A^2 d_B^2-1\right)} \kett{\hat{e}}_A + \frac{d_A (d_B-1)}{d_B\left(d_A^2 d_B^2-1\right)} \kett{\hat{\tau}}_A.
        \label{eq:pr_spt_square_eq1}
    \end{align}
    Similarly, for $\tau$, we have
    \begin{align}
        &\bbra{z_2(t)}_{B_2} \E_{U_2\in \rm Haar}\left[U_2^{\otimes 2} \otimes U_2^*{}^{\otimes 2}\right] \kett{\hat{\tau}}_A \kett{\bm 0}_{B_2} \nonumber \\
        &= \frac{1}{d^2-1}\left(\bbrakett{z_2(t)}{\hat{e}}_{B_2} \bbrakett{\hat{e}}{\hat{\tau}}_A \bbrakett{\hat{e}}{\bm 0}_{B_2}\kett{\hat{e}}_A + \bbrakett{z_2(t)}{\hat{\tau}}_{B_2} \bbrakett{\hat{\tau}}{\hat{\tau}}_{A}\bbrakett{\hat{\tau}}{\bm 0}_{B_2}\kett{\hat{\tau}}_{A}\right) \nonumber\\
        &\quad - \frac{1}{d(d^2-1)} \left(\bbrakett{z_2(t)}{\hat{e}}_{B_2}\bbrakett{\hat{\tau}}{\hat{\tau}}_{A}\bbrakett{\hat{\tau}}{\bm 0}_{B_2}\kett{\hat{e}}_A +\bbrakett{z_2(t)}{\hat{\tau}}_{B_2}\bbrakett{\hat{e}}{\hat{\tau}}_{A}\bbrakett{\hat{e}}{\bm 0}_{B_2} \kett{\hat{\tau}}_A \right)\\
        &= \frac{1}{d^2-1}\left( d_A \kett{\hat{e}}_A + d_A^2 \kett{\hat{\tau}}_{A}\right) - \frac{1}{d(d^2-1)} \left(d_A^2\kett{\hat{e}}_A +d_A\kett{\hat{\tau}}_A \right)\\
        &= \frac{d_A(d_B-1)}{d_B\left(d_A^2 d_B^2-1\right)} \kett{\hat{e}}_A + \frac{d_A^2 d_B - 1}{d_B\left(d_A^2 d_B^2-1\right)} \kett{\hat{\tau}}_A.
        \label{eq:pr_spt_square_eq2}
    \end{align}
    Until now, we have obtain all the mapping for the evolution of $e, \tau$ under Haar ensemble of unitaries. For notation convenience, we use a $2$-dimensional linear system to represent $\kett{\hat{e}}_A, \kett{\hat{\tau}}_A$ and denote them as $\ket{\underline{0}}\coloneqq \kett{\hat{e}}_A, \ket{\underline{1}}\coloneqq \kett{\hat{\tau}}_A$ though they do not form orthogonal bases. 
    For initial condition of $\kett{\bm 0}_A$, Eq.~\eqref{eq:pr_spt_square_eq0} leads to $\kett{\bm 0}_A \to (\ket{\underline{0}} + \ket{\underline{1}})/d(d+1)$, which is taken to be the initial condition for evolution of this $2$-dimensional linear system.
    The mapping in Eqs.~\eqref{eq:pr_spt_square_eq1} and~\eqref{eq:pr_spt_square_eq2} can be represented by a $2\times 2$ matrix as
    \be
        M = \frac{1}{d_B\left(d_A^2 d_B^2 - 1\right)}\begin{pmatrix}
            d_A^2 d_B-1 & d_A(d_B-1)\\
            d_A(d_B-1)  & d_A^2 d_B-1
        \end{pmatrix},
        \label{eq:M_def}
    \ee
    and the ensemble-averaged squared probability in Eq.~\eqref{eq:pr_spt_square} becomes
    \begin{align}
        \E_{U\in \rm Haar}\left[P_U[\bmz(t), x(t)]^2\right] &= \frac{1}{d(d+1)}\left[\left(\bra{\underline{0}}M^{t-1}\left(\ket{\underline{0}} + \ket{\underline{1}}\right)\right) \bbrakett{x(t)}{\hat{e}}_A + \left(\bra{\underline{1}}M^{t-1}\left(\ket{\underline{0}} + \ket{\underline{1}}\right)\right) \bbrakett{x(t)}{\hat{\tau}}_A\right]\\
        &= \frac{2}{d(d+1)} \left(\frac{1 + d_A}{d_A d_B^2 + d_B}\right)^{t-1}.
    \end{align}
    Therefore, the ensemble-averaged CP is
    \begin{align}
        Z_{\rm HRCS}(t) &= \sum_{x(t),\bmz(t)} \E_{U\in \rm Haar}\left[P_U[\bmz(t), x(t)]^2\right] \\
        &= \frac{2 d_A d_B^t}{d(d+1)} \left(\frac{1 + d_A}{d_A d_B^2 + d_B}\right)^{t-1}\\
        &= \frac{2(1+d_A)^{t-1}}{(d_A d_B+1)^t},
    \end{align}
    which is Eq.~\eqref{eq:HRCS_CP_spt_app} in Theorem~\ref{HRCS_CP_spt_app}.
    
    In the following, we derive the asymptotic expression of CP in HRCS in the large system limit of $d_A \gg 1$.
    Recall that the CP for Haar random states is $Z_{\rm H}(n) = 2/(d_A d_B^t+1)$ (see Eq.~\eqref{eq:CP_haar} with $N_{\rm eff} = N_A + tN_B$). The CP ratio between Eq.~\eqref{eq:HRCS_CP_spt_app} of HRCS and Haar one becomes
    \begin{align}
        \frac{Z_{\rm HRCS}(t)}{Z_{\rm H}(n)} &= \frac{2(d_A+1)^{t-1}}{(1+d_A d_B)^t} \frac{d_A d_B^t + 1}{2} \\
        &= \frac{d_A^{t-1}d_A d_B^t}{(d_A d_B)^t} \left(1 + \frac{1}{d_A}\right)^{t-1} \left(1+\frac{1}{d_A d_B^t}\right) \left(1+\frac{1}{d_A d_B}\right)^{-t}\\
        &=\exp\left[(t-1)\ln\left(1 + \frac{1}{d_A}\right)+ \ln\left(1+\frac{1}{d_A d_B^t}\right) -t\ln\left(1+\frac{1}{d_A d_B}\right)\right]\\
        &=\exp\left[\frac{t\left(1-d_B^{-1}\right) + d_B^{-t}-1}{d_A} + \calO\left(\frac{1}{d_A^2}\right)\right].
    \end{align}
    Therefore, we have the asymptotic form of $Z_{\rm HRCS}(t)$ as
    \be
        Z_{\rm HRCS}(t) = Z_{\rm H}(n) \exp\left[\frac{t\left(1-d_B^{-1}\right) + d_B^{-t}-1}{d_A} + \calO\left(\frac{1}{d_A^2}\right)\right].
        \label{eq:CP_asymp_app}
    \ee
\end{proof}

For $\epsilon$-approximate Haar AC in terms of CP, we can directly solve $Z_{\rm HRCS}(t) \le (1+\epsilon)Z_{\rm H}(n)$ to find the critical number of temporal steps. However, to obtain a intuitive and informative solution, we first upper bound the asymptotic form of Eq.~\eqref{eq:HRCS_CP_spt_asymp_app} by
\begin{align}
    Z_{\rm HRCS}(t) &=Z_{\rm H}(n)\exp\left[\frac{t\left(1-d_B^{-1}\right) + d_B^{-t}-1}{d_A} + \calO\left(\frac{1}{d_A^2}\right)\right] \nonumber\\
    &\le Z_{\rm H}(n) \exp\left[\frac{t\left(1-d_B^{-1}\right)}{d_A} + \calO\left(\frac{1}{d_A^2}\right)\right],
\end{align}
where we utilize $d_B^{-t} \le 1$, and the corresponding critical temporal steps can be simply solved as
\be
    t \le \tau \coloneqq \frac{d_A d_B}{d_B-1}\ln(1+\epsilon),
    \label{eq:HRCS_CP_tau_app}
\ee
which is Eq.~\eqref{eq:max_steps_overview} in the main text.

\subsection{Proof of Theorem~\ref{HRCS_PS_spt}}
\label{proof:theorem2}

\begin{theorem}
\label{HRCS_PS_spt_app}
(Theorem~\ref{HRCS_PS_spt} in the main text)
For holographic random circuit sampling with each unitary $U_t$ in $K$-design, the ensemble-averaged $K$-th power sum at step $t\ge 1$ is
    \begin{align}
       & Z_{\rm HRCS}^{(K)}(t) =
       K! d_A d_B^t \left(\frac{(d_A+K-1)!}{(d_A-1)!}\right)^{t-1}\left(\frac{(d_A d_B-1)!}{(d_A d_B + K-1)!}\right)^t,
        \label{eq:HRCS_PS_spt_app}
    \end{align}
    where $d_A=2^{N_A}, d_B=2^{N_B}$ are Hilbert space dimensions of the system and bath. 
    In the large-system limit $d_A \gg 1$,
    \begin{align}
        &Z_{\rm HRCS}^{(K)}(t) = Z_{\rm H}^{(K)}(n)\exp\left[K(K-1)\frac{t\left(1-d_B^{-1}\right)+d_B^{-t}-1}{2d_A} + \calO\left(\frac{1}{d_A^2}\right)\right].
        \label{eq:HRCS_PS_spt_asymp_app}
    \end{align}
    Here the $Z_{\rm H}^{(K)}(n)= K!\left(2^n\right)!/\left(2^n+K-1\right)!$ is the $K$-th power sum for $n$-qubit Haar random states.
\end{theorem}

\begin{proof}
    The joint probability of sampling $\bmz(t)$ and $x(t)$ still follows Eq.~\eqref{eq:pr_spt}, and according to definition of $K$-th PS (Eq.~\eqref{eq:PS_def}), the ensemble-averaged $K$-th power of $p[\bmz(t), x(t)]$ becomes
    \begin{align}
        \E_{U\in \rm Haar}\left[P_U[\bmz(t), x(t)]^K\right] &= \E_{U\in \rm Haar} \tr\left(U^{\otimes K}\ketbra{\bm 0}{\bm 0}_{A\bm B}^{\otimes K} U^\dagger{}^{\otimes K} \ketbra{x(t), \bmz(t)}{x(t), \bmz(t)}_{A\bm B}^{\otimes K}\right)\\
        &= \bbra{x(t),\bmz(t)}_{A \bm B} \E_{U\in \rm Haar}\left[U^{\otimes K} \otimes U^*{}^{\otimes K}\right] \kett{\bm 0}_{A\bm B},
        \label{eq:pr_spt_Kpow}
    \end{align}
    where we omit the $(\cdot)^{\otimes K}$ in the doubled bra and ket. We then evaluate Eq.~\eqref{eq:pr_spt_Kpow} through all unitaries from $U_1$ to $U_t$.
    For the first step involving $U_1$, we have
    \begin{align}
        &\bbra{z_1(t)}_{B_1} \E_{U_1\in \rm Haar}\left[U_1^{\otimes K} \otimes U_1^*{}^{\otimes K}\right] \kett{\bm 0}_{AB_1}\nonumber\\
        &= \bbra{z_1(t)}_{B_1} \left[\sum_{\sigma, \pi \in S_K} {\rm Wg}_d\left(\sigma^{-1}\pi\right)\kett{\sigma}\bbra{\pi}\right] \kett{\bm 0}_{AB_1} \\
        &= \sum_{\sigma, \pi \in S_K} {\rm Wg}_d\left(\sigma^{-1}\pi\right) \kett{\sigma}_A \\
        &=  \sum_{\sigma \in S_K}  \kett{\sigma}_A \left(\sum_{\pi \in S_K} {\rm Wg}_d\left(\sigma^{-1}\pi\right)\right) = \frac{(d-1)!}{(K+d-1)!} \sum_{\sigma \in S_K}  \kett{\sigma}_A ,
        \label{eq:pr_spt_Kpow_eq0}
    \end{align}
    where in the second line we utilize Haar unitary twirling of Eq.~\eqref{eq:haar_twirling}, and in the last line, we utilize the identity in Eq.~\eqref{eq:W_sum}. Now we can regard the sum of all permutations $\sum_{\sigma \in S_K} \kett{\sigma}_A$ as a single entry and evaluate its evolution under the twirling of Haar random unitary $U_2$
    \begin{align}
        &\bbra{z_2(t)}_{B_2} \E_{U_2\in \rm Haar}\left[U_2^{\otimes K} \otimes U_2^*{}^{\otimes K}\right] \left(\sum_{\sigma \in S_K} \kett{\sigma}_A\right) \kett{\bm 0}_{B_2} \nonumber\\
        &= \bbra{z_2(t)}_{B_2} \left[\sum_{\pi, \pi' \in S_K} {\rm Wg}_d\left(\pi^{-1}\pi'\right)\kett{\pi}\bbra{\pi'}\right] \left(\sum_{\sigma \in S_K} \kett{\sigma}_A\right) \kett{\bm 0}_{B_2} \\
        &= \sum_{\pi, \pi' \in S_K} {\rm Wg}_d\left(\pi^{-1}\pi'\right) \kett{\pi}_{A} \bbrakett{z_2(t)}{\pi}_{B_2} \sum_{\sigma \in S_K} \bbrakett{\pi'}{\sigma}_A \bbrakett{\pi'}{\bm 0}_{B_2} \\
        &= \sum_{\pi, \pi'\in S_K}{\rm Wg}_d\left(\pi^{-1}\pi'\right) \kett{\pi}_{A} \left(\sum_{\sigma \in S_K} d_A^{|\pi'^{-1} \sigma|}\right) \\
        &= \sum_{\pi, \pi'\in S_K}{\rm Wg}_d\left(\pi^{-1}\pi'\right) \kett{\pi}_{A} \frac{(d_A+K-1)!}{(d_A-1)!} \\
        &= \frac{(d_A+K-1)! (d-1)!}{(d_A-1)! (d+K-1)!} \sum_{\sigma \in S_K} \kett{\sigma}_A,
        \label{eq:pr_spt_Kpow_eq1}
    \end{align}
    where the second to last line we utilize the definition of Stirling numbers of the first kind which leads to the falling factorial. From Eq.~\eqref{eq:pr_spt_Kpow_eq1}, we find that the input entry $\sum_{\sigma \in S_K} \kett{\sigma}_A$ is exactly reproduced, which is expected as the sum of all permutation elements remains invariant under twirling of Haar unitary. Therefore, we can evaluate Eq.~\eqref{eq:pr_spt_Kpow} utilizing Eq.~\eqref{eq:pr_spt_Kpow_eq0} and \eqref{eq:pr_spt_Kpow_eq1} as
    \begin{align}
        \E_{U\in \rm Haar}\left[P_U[\bmz(t), x(t)]^K\right] &= \left(\frac{(d_A+K-1)! (d-1)!}{(d_A-1)! (d+K-1)!}\right)^{t-1} \frac{(d-1)!}{(d+K-1)!} \bbra{x(t)}\sum_{\sigma \in S_K} \kett{\sigma}_A \\
        &= K! \left(\frac{(d_A+K-1)!}{(d_A-1)!}\right)^{t-1} \left(\frac{(d_A d_B-1)!}{(d_A d_B +K-1)!}\right)^t,
    \end{align}
    and the $K$-th PS becomes
    \be
        Z^{(K)}_{\rm HRCS}(t) = \sum_{x(t), \bmz(t)} \E_{U\in \rm Haar}\left[P_U[\bmz(t), x(t)]^K\right] = K! d_A d_B^t \left(\frac{(d_A+K-1)!}{(d_A-1)!}\right)^{t-1} \left(\frac{(d_A d_B-1)!}{(d_A d_B +K-1)!}\right)^t,
    \ee
    which completes the proof of Eq.~\eqref{eq:HRCS_PS_spt_app} in Theorem~\ref{HRCS_PS_spt_app}.

    In the large system limit of $d_A \gg 1$, recall that the $K$-th PS for Haar random states is $Z_{\rm H}^{(K)}(n) = K!(d_A d_B^t)!/(d_A d_B^t + K-1)!$ (see Eq.~\eqref{eq:haar_PS} with $N=n = N_A + tN_B$). The $K$-th PS ratio between Eq.~\eqref{eq:HRCS_PS_spt_app} of HRCS and Haar one becomes
    \begin{align}
        \frac{Z_{\rm HRCS}^{(K)}(t)}{Z_{\rm H}^{(K)}(n)} &= \frac{K!d_A d_B^t \left[(d_A+K-1)!\right]^{t-1} \left[(d_A d_B-1)!\right]^t \left(d_A d_B^t + K-1\right)!}{\left[(d_A-1)!\right]^{t-1} \left[(d_A d_B + K-1)!\right]^t K!(d_A d_B^t)!}  \\
        &= \frac{\left[(d_A+K-1)!\right]^{t-1} \left[(d_A d_B-1)!\right]^t \left(d_A d_B^t + K-1\right)!}{\left[(d_A-1)!\right]^{t-1} \left[(d_A d_B + K-1)!\right]^t (d_A d_B^t-1)!}\\
        &= \left[\prod_{j=0}^{K-1} (d_A+j)\right]^{t-1} \left[\prod_{j=0}^{K-1} (d_A d_B +j)\right]^{-t} \left[\prod_{j=0}^{K-1} (d_A d_B^t+j)\right] \\
        &= \exp\left[(t-1)\sum_{j=0}^{K-1}\ln(d_A+j) - t\sum_{j=0}^{K-1} \ln(d_A d_B + j) + \sum_{j=0}^{K-1} \ln\left(d_A d_B^t + j\right)\right] \\
        &= \exp\left[(t-1)\sum_{j=0}^{K-1}\left(\ln(d_A) + \ln\left(1+\frac{j}{d_A}\right)\right) - t\sum_{j=0}^{K-1} \left(\ln(d_A d_B) + \ln\left(1+\frac{j}{d_A d_B}\right)\right) \nonumber\right.\\
        &\left. \qquad \quad + \sum_{j=0}^{K-1} \left(\ln\left(d_A d_B^t\right)\ln\left(1 + \frac{j}{d_A d_B^t}\right)\right)\right]\\
        &= \exp\left[(t-1)\ln(d_A) - t\ln(d_A d_B) + \ln\left(d_A d_B^t\right)\right]\nonumber\\
        &\quad \times  \exp\left[\sum_{j=0}^{K-1} \left((t-1)\ln\left(1+\frac{j}{d_A}\right) - t\ln\left(1+\frac{j}{d_A d_B}\right) + \ln\left(1 + \frac{j}{d_A d_B^t}\right)\right)\right]\\
        &= \exp\left[\sum_{j=0}^{K-1} \left((t-1)\frac{j}{d_A} - t\frac{j}{d_A d_B} + \frac{j}{d_A d_B^t}\right) + \calO\left(\frac{1}{d_A^2}\right)\right] \\
        &= \exp\left[\frac{K(K-1)}{2d_A} \left((t-1) - \frac{t}{d_B} + \frac{1}{d_B^t}\right) + \calO\left(\frac{1}{d_A^2}\right)\right],
    \end{align}
    which the the asymptotic form in Eq.~\eqref{eq:HRCS_PS_spt_asymp_app}.
\end{proof}

For $\epsilon$-approximate $K$-th PS in HRCS defined by $Z_{\rm HRCS}^{(K)}(t) \le (1+\epsilon)Z_{\rm H}^{(K)}(n)$, similarly to the approach in CP, we again first introduce the upper bound to asymptotic form (Eq.~\eqref{eq:HRCS_PS_spt_asymp_app}) as
\begin{align}
    Z_{\rm HRCS}^{(K)}(t) &= Z_{\rm H}^{(K)}(n)\exp\left[K(K-1)\frac{t\left(1-d_B^{-1}\right)+d_B^{-t}-1}{2d_A} + \calO\left(\frac{1}{d_A^2}\right)\right] \le Z_{\rm H}^{(K)}(n)\exp\left[K(K-1)\frac{t(1-d_B^{-1})}{2d_A} + \calO\left(\frac{1}{d_A^2}\right)\right].
\end{align}
We instead require the upper bound above to satisfy the $\epsilon$-approximate PS and leads to
\be
    t \le \tau^{(K)} \coloneqq \frac{d_B}{d_B-1}\frac{2d_A \ln(1+\epsilon)}{K(K-1)},
\ee
which reduces to Eq.~\eqref{eq:HRCS_CP_tau_app} with $K=2$.

\subsection{Total variation distance}

In this section, we utilize total variation distance (TVD) to quantify the deviation of the statistics of HRCS sampling distribution from the Haar statistics. Unlike the noisy RCS study~\cite{dalzell2024random, fefferman2024effect} with the ensemble average over the random unitary circuit of RCS only, we instead consider the TVD averaged over both HRCS circuit instances $C$ over the unitaries $U_1,\cdots,U_t$ and the Haar random unitary $U$ in comparison as
\small
\be
{\rm TVD} \coloneqq \frac{1}{2} \E_{\{U_k\}\in \rm Haar} \E_{U\in \rm Haar} \|P_C[\bmz(t), x(t)]-P_U[\bmz(t), x(t)]\|_1,
\label{TVD_def}
\ee
\normalsize
where 
$P_C[\bmz(t), x(t)]$ is the joint distribution of HRCS circuit instance $C$ and $P_U[\bmz(t), x(t)]$ is the corresponding distribution of conventional Haar-random RCS with $n = N_A + t N_B$ qubits.
We have following inequality
\be 
{\rm TVD}\le \frac{1}{2}\sqrt{2^{n}Z_{\rm HRCS}(t)}\le \frac{1}{\sqrt{2}}\exp\left[\frac{t-1}{2d_A} + \calO\left(\frac{1}{d_A^2}\right)\right],
\label{TVD_UB}
\ee 
where the second inequality holds in the large-system limit similar to Eq.~\eqref{eq:HRCS_CP_spt_asymp_app}.
This upper bound shows the statistical distance is suppressed exponentially by the dimension of the system and supports the time scale of $\epsilon$-approximate Haar AC in Ineq.~\eqref{eq:max_steps_overview}. 
The proof is as follows.

\begin{proof}
    From the definition in Eq.~\eqref{TVD_def},
we have the upper bound by
\begin{align}
    {\rm TVD}&=\E_{C}\E_U\left[\frac{1}{2} \|P_C[\bm y]-P_U(\bm y)\|_1\right] \nonumber\\
    &\le 2^{n/2}\E_{C}\E_U\left[\frac{1}{2} \|P_C(\bm y)-P_{U}(\bm y)\|_2\right] \\
    &\le \frac{1}{2} \sqrt{2^{n}\E_C\E_U \|P_C(\bm y)-P_{U}(\bm y)\|_2^2},
    \label{TVD_middle_step}
\end{align}
where we denote $\bm y=(\bmz, x)$ for simplicity. $P_C[\bm y]$ is the sampling distribution in HRCS with circuit instance $C$ and $P_U[\bm y]$ is the sampling distribution from Haar random states with Haar unitary $U$.
The expectation value
\begin{align}
    &\E_C \E_U \|P_C(\bm y)-P_U(\bm y)\|_2^2\nonumber\\
    &= \E_C \E_U\sum_{\bm y}P_C(\bm y)^2 + P_{U}(\bm y)^2 -2P_C(\bm y) P_{U}(\bm y)\\
    &= Z_{\rm HRCS}(t)+Z_{\rm H}(n) - 2\sum_{\bm y}\E_{C}\left[P_C(\bm y)\right]\E_U[P_U(\bm y)]\\
    &= Z_{\rm HRCS}(t)+Z_{\rm H}(n) - \frac{2^{n+1}}{4^{n}}\\
    &= Z_{\rm HRCS}(t) - Z_{\rm H}(n)/2^{n}.
\end{align}
Inputting the above into Eq.~\eqref{TVD_middle_step}, we arrive at the upper bound for TVD
\begin{align}
    {\rm TVD} &\le \frac{1}{2}\sqrt{2^{n}Z_{\rm HRCS}(t) - Z_{\rm H}\left(n\right)}\\
    &\le \frac{1}{2}\sqrt{2^{n}Z_{\rm HRCS}(t)}\\
    &\le  \frac{1}{2}\sqrt{2^{n}Z_{\rm H}(n)\exp\left[\frac{t-1}{d_A} + \calO\left(\frac{1}{d_A^2}\right)\right]}\\
    &\le \frac{1}{\sqrt{2}}\exp\left[\frac{t-1}{2d_A} + \calO\left(\frac{1}{d_A^2}\right)\right],
\end{align}
where we utilize Eq.~\eqref{eq:CP_asymp_app} to obtain leading order result.
\end{proof}

\subsection{Reset in HRCS}

In this section, we show that whether adopting the reset strategy does not affect the CP and PS. The intuition is simple. Suppose we perform measurement on bath without following reset, we just need to add a few Pauli-X gate on corresponding bath qubits to flip their states from $\ket{1}$ to $\ket{0}$. Since the circuit unitary is assumed to be Haar random and both CP and PS are defined to be ensemble-averaged over circuit unitaries, the result will not change. We present the detailed proof as follows.

The joint probability is
\begin{align}
    P_U'[\bmz(t), x(t)] &= \tr\left(U \left(\ketbra{\bm 0}{\bm 0}_A \otimes \ketbra{\bm 0}{\bm 0}_{B_0}\otimes_{j=1}^{t-1}\ketbra{z_j}{z_j}_{B_j}\right) U^\dagger \ketbra{x(t)}{x(t)}_A \otimes \ketbra{\bmz(t)}{\bmz(t)}_{\bm B}\right)\\
    &= \tr\left(U \left(\ketbra{\bm 0}{\bm 0}_A \otimes \ketbra{\bm 0}{\bm 0}_{B_0}\otimes_{j=1}^{t-1}X^{z_j}\ketbra{\bm 0}{\bm 0}_{B_j}X^{z_j}\right) U^\dagger \ketbra{x(t)}{x(t)}_A \otimes \ketbra{\bmz(t)}{\bmz(t)}_{\bm B}\right)\\
    &= \tr\left(U_{\bmz} \left(\ketbra{\bm 0}{\bm 0}_A \otimes \ketbra{\bm 0}{\bm 0}_{B_0}\otimes_{j=1}^{t-1}\ketbra{\bm 0}{\bm 0}_{B_j}\right) U_{\bmz}^\dagger \ketbra{x(t)}{x(t)}_A \otimes \ketbra{\bmz(t)}{\bmz(t)}_{\bm B}\right),
\end{align}
where in last line we define $U_{\bmz} = U_t X^{z_{t-1}} U_{t-1} \cdots X^{z_1}U_1$. The ensemble-averaged CP then becomes
\begin{align}
    Z' &= \E_{U \in \rm Haar}\sum_{\bmz(t), x(t)} P_U[\bmz(t), x(t)]^2\\
    &= \sum_{\bmz(t), x(t)} \E_{U \in \rm Haar}\tr\left(U_{\bmz} \left(\ketbra{\bm 0}{\bm 0}_A \otimes \ketbra{\bm 0}{\bm 0}_{B_0}\otimes_{j=1}^{t-1}\ketbra{\bm 0}{\bm 0}_{B_j}\right) U_{\bmz}^\dagger \ketbra{x(t)}{x(t)}_A \otimes \ketbra{\bmz(t)}{\bmz(t)}_{\bm B}\right)\\
    &= \sum_{\bmz(t), x(t)} \E_{V \in \rm Haar}\tr\left(V \left(\ketbra{\bm 0}{\bm 0}_A \otimes \ketbra{\bm 0}{\bm 0}_{B_0}\otimes_{j=1}^{t-1}\ketbra{\bm 0}{\bm 0}_{B_j}\right) V^\dagger \ketbra{x(t)}{x(t)}_A \otimes \ketbra{\bmz(t)}{\bmz(t)}_{\bm B}\right)\\
    &= Z,
\end{align}
where in the second to last line we utilize the left and right invariance property of Haar random unitaries. The proof for PS can be completed in the same way by replacing the exponent from $2$ to $K$, and we don't repeat it here.


\section{Spacetime tradeoff}
\label{app:tradeoff}

\subsection{Tradeoff for collision probability}

In this section, we show the spacetime tradeoff in HRCS for holding $\epsilon$-approximate AC of Haar random states.  Specifically, we consider the tradeoff between number of temporal steps $t$ and number of physical qubits $N = N_A + N_B$ for generating $n$ bits. For analysis convenience, we alleviate the integer constraints of $N, N_A, N_B, t \in \mathbb{N}^+$ and focus on the asymptotic scaling.
Recall the asymptotic form of CP of HRCS in Eq.~\eqref{eq:CP_asymp_app}, we first derive a simpler form of time constraints for $\epsilon$-approximate AC as
\begin{align}
    Z_{\rm HRCS}(t) &= Z_{\rm H}(n)\exp\left[\frac{t(1-d_B^{-1})+d_B^{-t}-1}{d_A} + \calO\left(\frac{1}{d_A^2}\right)\right]\le Z_{\rm H}(n)\exp\left[\frac{t}{d_A} + \calO\left(\frac{1}{d_A^2}\right)\right] \le (1+\epsilon) Z_{\rm H}(n)\\
    t &\le 2^{N_A} \ln(1+\epsilon). 
\end{align}
By using the condition $N_A = n - N_B t$, we have
\begin{align}
    t &\le \frac{W_0\left(2^n N_B \ln(2) \ln(1+\epsilon)\right)}{N_B \ln(2)} \\
    &= \frac{1}{N_B}\left[\log_2\left(2^n N_B \ln(2) \ln(1+\epsilon)\right) - \log_2\left(\ln\left(2^n N_B \ln(2) \ln(1+\epsilon)\right)\right)\right] + \calO\left(\frac{\ln n}{n}\right)\\
    &= \frac{1}{N_B} \left[n+\log_2(N_B)+\log_2\left(\ln(2)\ln(1+\epsilon)\right)-\log_2\left(n\ln(2) +\log_2(N_B)+\log_2\left(\ln(2)\ln(1+\epsilon)\right)\right)\right] + \calO\left(\frac{\ln n}{n}\right)\\
    &= \frac{1}{N_B} \left[n-\log_2(n)+\log_2(N_B)+\log_2\left(\ln(1+\epsilon)\right)\right] + \calO\left(\frac{\ln n}{n}\right),
\end{align}
where $W_0(z)$ is the principal branch of the solution $We^W = z$, and can be approximated by $W_0(z) = \ln(z) - \ln(\ln(z)) + \calO(\ln \ln z/\ln z)$ in the asymptotic limit of $z \to +\infty$.
We can substitute bath size with $N_B = (n-N)/(t-1)$ and reduce the above inequality to
\begin{align}
    \frac{t(n-N)}{t-1}-\log_2\left(\frac{n-N}{t-1}\right) -n + \log_2(n)- \log_2\left(\ln(1+\epsilon)\right) \le 0,
\end{align}
which can be solved as
\begin{align}
    N &\ge n+\frac{t-1}{t\ln(2)}W_{-1}\left(-\frac{2^{-n} n t \ln(2)}{\ln (1+\epsilon)}\right),\\
    &=n+\frac{t-1}{t\ln(2)}\left[\ln\left(\frac{2^{-n}nt \ln(2)}{\ln(1+\epsilon)}\right) - \ln\left(-\ln\left(\frac{2^{-n}nt \ln(2)}{\ln(1+\epsilon)}\right)\right) + \calO\left(\frac{\ln n}{n}\right)\right]\\
    &= n+\frac{t-1}{t}\left[-n+\log_2(n)+\log_2(t) + \log_2(\ln(2)) -\log_2(\ln(1+\epsilon)) \right.\nonumber\\
    &\qquad \qquad \qquad \left.  - \log_2\left(n \ln(2)-\ln(n)-\ln(t) - \ln(\ln(2)) +\ln(\ln(1+\epsilon))\right) + \calO\left(\frac{\ln n}{n}\right) \right]\\
    &= \frac{n}{t} + \frac{t-1}{t}\log_2\left(\frac{t}{\ln(1+\epsilon)}\right) + \calO\left(\frac{\ln n}{n}\right),
    \label{eq:N_t_app}
\end{align}
where $W_{-1}(\cdot)$ is the lower branch of the lambert W function, and the second line is the leading order expansion of $W_{-1}(z)$ with $z \to 0^{-}$. When $t/n \ll 1$, Eq.~\eqref{eq:N_t_app} exhibits the scaling of $N \sim n/t$. We can further solve $\diff_t N = 0$ to identify the theoretical optimal time to minimize $N$ as
\begin{align}
    t^\star &= W_0\left(e \ln(1+\epsilon) 2^n\right)\\
    &= n \ln(2) + \ln(\ln(1+\epsilon))+1 - \ln\left[n \ln(2) + \ln(\ln(1+\epsilon))+1\right] + \calO\left(\frac{\ln n}{n}\right)\\
    &= \ln(2) n - \ln(n) + \ln(\log_2(1+\epsilon)) + 1 + \calO\left(\frac{\ln n}{n}\right) \label{eq:t_opt_supp},
\end{align}
which scales as $\ln(2) n$ in the leading order.
With the optimal time $t^\star$ in Eq.~\eqref{eq:t_opt_supp}, we can solve the minimum number of qubits as
\begin{align}
    N_{\rm min} &= \frac{n}{t^\star} + \frac{t^\star-1}{t^\star}\log_2\left(\frac{t^\star}{\ln(1+\epsilon)}\right)\\
    &= \frac{n}{\ln(2) n - \ln(n) + a} + \frac{\ln(2) n - \ln(n) + a-1}{\ln(2) n - \ln(n) + a}\log_2\left(\frac{\ln(2) n - \ln(n) + a}{\ln(1+\epsilon)}\right) \\
    &= \log_2\left(\frac{n}{\log_2(1+\epsilon)}\right) + \frac{1}{\ln(2)} +\calO\left(\frac{\ln n}{n}\right),
\end{align}
where in the second line we denote $a \coloneqq \ln(\log_2(1+\epsilon))+1$ for simplicity.
The minimum bath and system size are thus
\begin{align}
    N_{B, \rm min} &= \frac{1}{\ln(2)} + \calO\left(\frac{\ln n}{n}\right),\\
    N_{A,\rm min} &=  \log_2\left(\frac{n}{\log_2(1+\epsilon)}\right) +\calO\left(\frac{\ln n}{n}\right).
\end{align}

\subsection{Tradeoff for $K$-th power sum}

Similarly, we can extend the above analysis to $K$-th order power sums. We begin with the asymptotic form of $Z^{(K)}_{\rm HRCS}(t)$ in Eq.~\eqref{eq:HRCS_PS_spt_asymp} for $\epsilon$-approximate $K$-th PS as
\begin{align}
    Z_{\rm HRCS}^{(K)}(t) &\le Z_{\rm H}^{(K)}(n) \exp\left[K(K-1)\frac{t}{2d_A}\right] \le (1+\epsilon) Z_{\rm H}^{(K)}(n)\\
    t &\le \frac{2^{N_A-1} \ln(1+\epsilon)}{K(K+1)}.
\end{align}
Using the condition $N_A = n - N_B t$, we have
\begin{align}
    t & \le \frac{1}{N_B \ln(2)}W_0\left(\frac{2^{n+1}N_B \ln(2)\ln(1+\epsilon)}{K(K-1)}\right)\\
    &= \frac{1}{N_B}\left[n+1 + \log_2(N_B) + \log_2(\ln(1+\epsilon))-\log_2(K(K-1)) -\log_2(n)\right] + \calO\left(\frac{\ln n}{n}\right),
\end{align}
where the second line holds when $K\ll 2^{n/2}$. On the other hand, for $K \gtrsim 2^{n/2}$, then we can only have $t$ to be a constant. In particular, for matching $K=2^n$, HRCS reduces to conventional RCS. 

As a consequence, we focus on the approximation to relative lower orders in the following. By further substituting the bath with $N_B = (n-N)/(t-1)$, we can solve the necessary number of physical qubits as
\begin{align}
    N^{(K)} & \ge n + \frac{t-1}{t\ln(2)} W_{-1}\left( -\frac{2^{-n-1}n t K(K-1)\ln(2) }{\ln (1+\epsilon)}\right)  \\
    &= \frac{n}{t} + \frac{t-1}{t}\log_2\left(\frac{K(K-1)t}{2\ln(1+\epsilon)}\right) + \calO\left(\frac{\ln n}{n}\right).
\end{align}
The theoretical optimal time for minimizing $N$ is thus
\begin{align}
    t^{(K)}{}^\star &= W_0\left(\frac{e 2^n 2\ln(1+\epsilon)}{K(K-1)}\right) = \ln(2) n - \ln(n) + \ln\left(\frac{2\log_2(1+\epsilon)}{K(K-1)}\right) + 1 + \calO\left(\frac{\ln n}{n}\right) \label{eq:tK_opt_supp}.
\end{align}
We can also solve the minimum number of qubits as
\begin{align}
    N_{\rm min}^{(K)} &= \frac{n}{t^{(K)}{}^\star} + \frac{t^{(K)}{}^\star-1}{t^{(K)}{}^\star}\log_2\left(\frac{K(K-1)t^{(K)}{}^\star}{2\ln(1+\epsilon)}\right)= \log_2\left(\frac{K(K-1)n}{2\log_2(1+\epsilon)}\right) + \frac{1}{\ln 2} + \calO\left(\frac{\ln n}{n}\right),
\end{align}
where in the second line we denote $a \coloneqq \ln(2\log_2(1+\epsilon)/K(K-1))$ for simplicity. Compare to the result of collision probability, the minimum number of qubits also scales as logarithmically with number of bits but with an addition of order of power sum as $\sim \log_2 n + 2\log_2 K$.
The minimum bath and system size are thus
\begin{align}
    N_{B, \rm min}^{(K)} &= \frac{1}{\ln(2)} + \calO\left(\frac{\ln n}{n}\right),\\
    N_{A,\rm min}^{(K)} &=  \log_2\left(\frac{K(K-1)n}{2\log_2(1+\epsilon)}\right) +\calO\left(\frac{\ln n}{n}\right).
\end{align}

\section{Collision probability for HRCS with local brickwork circuits}
\label{app:brickwork}

In this section, we derive the collision probability of HRCS with unitary brickwork circuit each step consisting of local 2-qubit gate on nearest neighbor. For simplicity, we consider open boundary conditions and even number of physical qubits and layer of gates in each step, and our results also hold for other cases.
We restate the lemma here for convenience.
\begin{lemma}
\label{lemma:HRCS_CP_app}
(Lemma~\ref{lemma_HRCS_CP} in the main text)
    For holographic random circuit sampling with each unitary $U_t$ to be $L$ layers brickwork circuit where every $2$-qubit unitary satisfies $2$-design, the ensemble-averaged collision probability at step $t \ge 1$ is upper bounded by
    \be
        Z(t) \le \left[(2t-1)\exp\left(\frac{N}{2} \left(\frac{4}{5}\right)^L\right) - (2t-2)\right]Z_{\rm HRCS}(t)
        \label{eq:HRCS_CP_brickwork_app}.
    \ee
\end{lemma}

\begin{proof}
Our proof follows the same technique as conventional RCS results~\cite{barak2020spoofing, dalzell2022random}.

The ensemble average of collision probability of brickwork HRCS is
\be
    Z(t) = \E\sum_{\bmz, x}P[\bmz, x]^2 = 2^n \E P[\bm 0, 0]^2 = 2^n \E \tr(U^{\otimes 2} \ketbra{0}{0}_{A\bm B}^{\otimes 2} U^\dagger{}^{\otimes 2} \ketbra{0}{0}_{A\bm B}^{\otimes 2}),
\ee
where $\E$ represents the ensemble average over every 2-qubit gate. According to the twirling of 2-copy Haar unitaries in Eq.~\eqref{eq:haar_twirling_K2}, we can effectively rewrite the above average collision probability as the partition function of a hexagon spin lattice, where each effective spin can be either identity or swap operator.

\be
Z(t) = 2^n \sum_{\scalebox{0.5}{\begin{tikzpicture}
	\begin{pgfonlayer}{nodelayer}
		\node [style=permred] (0) at (0, 0) {};
	\end{pgfonlayer}
\end{tikzpicture}
, \begin{tikzpicture}
	\begin{pgfonlayer}{nodelayer}
		\node [style=permblue] (0) at (0, 0) {};
	\end{pgfonlayer}
\end{tikzpicture}
} \in S_2}\scalebox{0.6}{\begin{tikzpicture}
	\begin{pgfonlayer}{nodelayer}
		\node [style=permred] (r1) at (-9.25, 3.75) {};
		\node [style=permred] (r2) at (-9.25, 1.25) {};
		\node [style=permred] (r3) at (-9.25, -1.25) {};
        \node [style=permred] (r4) at (-9.25, -3.75) {};
        
        \node [style=permblue] (b1) at (-8, 3.75) {};
		\node [style=permblue] (b2) at (-8, 1.25) {};
		\node [style=permblue] (b3) at (-8, -1.25) {};
        \node [style=permblue] (b4) at (-8, -3.75) {};
        
        \node [style=none] (e1) at (-7, 5) {};
        \node [style=permred] (r5) at (-7, 2.5) {};
		\node [style=permred] (r6) at (-7, 0) {};
        \node [style=permred] (r7) at (-7, -2.5) {};
        \node [style=none] (e2) at (-7, -5) {};

        \node [style=none] (e3) at (-5.75, 5) {};
		\node [style=permblue] (b5) at (-5.75, 2.5) {};
		\node [style=permblue] (b6) at (-5.75, 0) {};
        \node [style=permblue] (b7) at (-5.75, -2.5) {};
        \node [style=none] (e4) at (-5.75, -5) {};
        
		\node [style=permred] (r8) at (-4.75, 3.75) {};
		\node [style=permred] (r9) at (-4.75, 1.25) {};
		\node [style=permred] (r10) at (-4.75, -1.25) {};
        \node [style=permred] (r11) at (-4.75, -3.75) {};
        
        \node [style=permblue] (b8) at (-3.5, 3.75) {};
        \node [style=permblue] (b9) at (-3.5, 1.25) {};
        \node [style=permblue] (b10) at (-3.5, -1.25) {};
        \node [style=permblue] (b11) at (-3.5, -3.75) {};
        
        \node [style=none] (e5) at (-2.5, 5) {};
        \node [style=permred] (r12) at (-2.5, 2.5) {};
		\node [style=permred] (r13) at (-2.5, 0) {};
        \node [style=permred] (r14) at (-2.5, -2.5) {};

        \node [style=none] (e6) at (-1.25, 5) {};
		\node [style=permblue] (b12) at (-1.25, 2.5) {};
		\node [style=permblue] (b13) at (-1.25, 0) {};
        \node [style=permblue] (b14) at (-1.25, -2.5) {};

		\node [style=permred] (r21) at (-0.25, 3.75) {};
		\node [style=permred] (r22) at (-0.25, 1.25) {};
		\node [style=permred] (r23) at (-0.25, -1.25) {};
        \node [style=permred] (r24) at (-0.25, -3.75) {};
        
        \node [style=permblue] (b21) at (1, 3.75) {};
		\node [style=permblue] (b22) at (1, 1.25) {};
		\node [style=permblue] (b23) at (1, -1.25) {};
        \node [style=permblue] (b24) at (1, -3.75) {};
        
        \node [style=none] (e21) at (2, 5) {};
        \node [style=permred] (r25) at (2, 2.5) {};
		\node [style=permred] (r26) at (2, 0) {};
        \node [style=permred] (r27) at (2, -2.5) {};
        \node [style=none] (e22) at (2, -5) {};

        \node [style=none] (e23) at (3.25, 5) {};
		\node [style=permblue] (b25) at (3.25, 2.5) {};
		\node [style=permblue] (b26) at (3.25, 0) {};
        \node [style=permblue] (b27) at (3.25, -2.5) {};
        \node [style=none] (e24) at (3.25, -5) {};

		\node [style=permred] (r28) at (4.25, 3.75) {};
		\node [style=permred] (r29) at (4.25, 1.25) {};
		\node [style=permred] (r210) at (4.25, -1.25) {};
        
        \node [style=permred] (r211) at (4.25, -3.75) {};
        
        \node [style=permblue] (b28) at (5.5, 3.75) {};
        \node [style=permblue] (b29) at (5.5, 1.25) {};
        \node [style=permblue] (b210) at (5.5, -1.25) {};
        \node [style=permblue] (b211) at (5.5, -3.75) {};
        
        \node [style=none] (e25) at (6.5, 5) {};
        \node [style=permred] (r212) at (6.5, 2.5) {};
		\node [style=permred] (r213) at (6.5, 0) {};
        \node [style=permred] (r214) at (6.5, -2.5) {};

        \node [style=none] (e26) at (7.75, 5) {};
		\node [style=permblue] (b212) at (7.75, 2.5) {};
		\node [style=permblue] (b213) at (7.75, 0) {};
        \node [style=permblue] (b214) at (7.75, -2.5) {};
	\end{pgfonlayer}
    
	\begin{pgfonlayer}{edgelayer}
        \draw (r1.center) to (b1.center);
        \draw (r2.center) to (b2.center);
        \draw (r3.center) to (b3.center);
        \draw (r4.center) to (b4.center);
        \draw (b1.center) to (e1.center);
        \draw (b1.center) to (r5.center);
        \draw (b2.center) to (r5.center);
        \draw (b2.center) to (r6.center);
        \draw (b3.center) to (r6.center);
        \draw (b3.center) to (r7.center);
        \draw (b4.center) to (r7.center);
        \draw (b4.center) to (e2.center);
        \draw (e1.center) to (e3.center);
        \draw (r5.center) to (b5.center);
        \draw (r6.center) to (b6.center);
        \draw (r7.center) to (b7.center);
        \draw (r8.center) to (b8.center);
        \draw (e2.center) to (e4.center);
        \draw (e3.center) to (r8.center);
        \draw (b5.center) to (r8.center);
        \draw (b5.center) to (r9.center);
        \draw (b6.center) to (r9.center);
        \draw (b6.center) to (r10.center);
        \draw (b7.center) to (r10.center);
        \draw (b7.center) to (r11.center);
        \draw (e4.center) to (r11.center);
        \draw (r8.center) to (b8.center);
        \draw (r9.center) to (b9.center);
        \draw (r10.center) to (b10.center);
        \draw (r11.center) to (b11.center);
        \draw (b8.center) to (e5.center);
        \draw (b8.center) to (r12.center);
        \draw (b9.center) to (r12.center);
        \draw (b9.center) to (r13.center);
        \draw (b10.center) to (r13.center);
        \draw (b10.center) to (r14.center);
        \draw (b11.center) to (r14.center);
        \draw (e5.center) to (e6.center);
        \draw (r12.center) to (b12.center);
        \draw (r13.center) to (b13.center);
        \draw (r14.center) to (b14.center);

        \draw (e6.center) to (r21.center);
        \draw (b12.center) to (r21.center);
        \draw (b12.center) to (r22.center);
        \draw (b13.center) to (r22.center);
        \draw (b13.center) to (r23.center);
        \draw (b14.center) to (r23.center);
        
        \draw (r21.center) to (b21.center);
        \draw (r22.center) to (b22.center);
        \draw (r23.center) to (b23.center);
        \draw (r24.center) to (b24.center);
        \draw (b21.center) to (e21.center);
        \draw (b21.center) to (r25.center);
        \draw (b22.center) to (r25.center);
        \draw (b22.center) to (r26.center);
        \draw (b23.center) to (r26.center);
        \draw (b23.center) to (r27.center);
        \draw (b24.center) to (r27.center);
        \draw (b24.center) to (e22.center);
        \draw (e21.center) to (e23.center);
        \draw (r25.center) to (b25.center);
        \draw (r26.center) to (b26.center);
        \draw (r27.center) to (b27.center);
        \draw (e22.center) to (e24.center);
        \draw (e23.center) to (r28.center);
        \draw (b25.center) to (r28.center);
        \draw (b25.center) to (r29.center);
        \draw (b26.center) to (r29.center);
        \draw (b26.center) to (r210.center);
        \draw (b27.center) to (r210.center);
        \draw (b27.center) to (r211.center);
        \draw (e24.center) to (r211.center);
        \draw (r28.center) to (b28.center);
        \draw (r29.center) to (b29.center);
        \draw (r210.center) to (b210.center);
        \draw (r211.center) to (b211.center);
        \draw (b28.center) to (r212.center);
        \draw (b29.center) to (r212.center);
        \draw (b29.center) to (r213.center);
        \draw (b210.center) to (r213.center);
        \draw (b210.center) to (r214.center);
        \draw (b211.center) to (r214.center);
        \draw (r212.center) to (b212.center);
        \draw (r213.center) to (b213.center);
        \draw (r214.center) to (b214.center);
	\end{pgfonlayer}
\end{tikzpicture}
},
\ee
where $S_2 \coloneqq \{e, \tau\}$ is the permutation group of two elements. Note that we omit the inner product between permutation operators and computational basis which is always $1$ due to permutation invariance of $\ketbra{0}{0}_{A\bm B}^{\otimes 2}$.
In the above lattice diagram we take an example of $t=2$ steps and $L=4$ layers in each step in a system of $N_A = 6$ and $N_B = 2$.

Next, we implement the reduction by summing over all blue spins, developed in Ref.~\cite{hunter2019unitary}. We first quote the results here, which can be also be evaluated from the twirling result in Eq.~\eqref{eq:haar_twirling_K2}.
\begin{align}
    J_{\sigma_1\sigma_2\sigma_3} &= \sum_{\scalebox{0.5}{\begin{tikzpicture}
	\begin{pgfonlayer}{nodelayer}
		\node [style=permblue] (0) at (0, 0) {};
	\end{pgfonlayer}
\end{tikzpicture}
}\in S_2} \scalebox{0.8}{\begin{tikzpicture}
	\begin{pgfonlayer}{nodelayer}
		\node [style=permred] (14) at (1, -1.25) {$\sigma_3$};
		\node [style=permred] (15) at (1, 1.25) {$\sigma_2$};
		\node [style=permred] (16) at (-1.25, 0) {$\sigma_1$};
		\node [style=permblue] (17) at (0, 0) {$\pi$};
	\end{pgfonlayer}
	\begin{pgfonlayer}{edgelayer}
		\draw (14.center) to (17.center);
		\draw (15.center) to (17.center);
		\draw (16.center) to (17.center);
	\end{pgfonlayer}
\end{tikzpicture}} = \scalebox{0.8}{\begin{tikzpicture}
	\begin{pgfonlayer}{nodelayer}
		\node [style=plauette] (3) at (0, 0) {};
		\node [style=permred] (4) at (0.75, -1.25) {$\sigma_3$};
		\node [style=permred] (5) at (0.75, 1.25) {$\sigma_2$};
		\node [style=permred] (6) at (-1.5, 0) {$\sigma_1$};
	\end{pgfonlayer}
\end{tikzpicture}
} = \sum_{\pi \in S_2} {\rm Wg}(\sigma_1^{-1}\pi) 2^{|\pi^{-1}\sigma_2| + |\pi^{-1}\sigma_3|} \\
    &= \begin{cases}
        1 & \sigma_1=\sigma_2=\sigma_3, \\
        2/5 & \sigma_1=\sigma_2, \sigma_1 \neq \sigma_3,\\
        2/5 & \sigma_1=\sigma_3, \sigma_1 \neq \sigma_2,\\
        0 & {\rm otherwise}.
    \end{cases}.
    \label{eq:plauette}
\end{align}
Similarly, we also have a reduction over the boundary involving two red spins and a central connected blue spin as
\be
    G_{\sigma_1\sigma_2} = \sum_{\scalebox{0.5}{\begin{tikzpicture}
	\begin{pgfonlayer}{nodelayer}
		\node [style=permblue] (0) at (0, 0) {};
	\end{pgfonlayer}
\end{tikzpicture}
}\in S_2} \scalebox{0.8}{\begin{tikzpicture}
	\begin{pgfonlayer}{nodelayer}
		\node [style=permred] (1) at (0, 0) {$\sigma_1$};
		\node [style=permblue] (2) at (1.25, 0) {$\pi$};
		\node [style=permred] (3) at (2.25, 1.25) {$\sigma_2$};
	\end{pgfonlayer}
    \begin{pgfonlayer}{edgelayer}
        \draw (1.center) to (2.center);
        \draw (2.center) to (3.center);
    \end{pgfonlayer}
\end{tikzpicture}} = \scalebox{0.8}{\begin{tikzpicture}
	\begin{pgfonlayer}{nodelayer}
		\node [style=permred] (1) at (0, 0) {$\sigma_1$};
		\node [style=permred] (3) at (2.25, 1.25) {$\sigma_2$};
	\end{pgfonlayer}
    \begin{pgfonlayer}{edgelayer}
        \draw[purple] (1.center) to (3.center);
    \end{pgfonlayer}
\end{tikzpicture}} = \sum_{\pi \in S_2} {\rm Wg}(\sigma_1^{-1}\pi) 2^{|\pi^{-1}\sigma_2|} \\
    = \begin{cases}
        7/30 & \sigma_1 = \sigma_2,\\
        1/15 & \sigma_1 \neq \sigma_2.
    \end{cases}.
    \label{eq:CP_brickwork_hex}
\ee
Utilizing the above two identities, we can reduce the hexagon lattice to a triangular spin lattice as
\be
    Z(t) = 2^n \left(\frac{1}{20}\right)^{N/2-1} \sum_{\scalebox{0.5}{\begin{tikzpicture}
	\begin{pgfonlayer}{nodelayer}
		\node [style=permred] (0) at (0, 0) {};
	\end{pgfonlayer}
\end{tikzpicture}
}\in S_2} \scalebox{0.6}{\begin{tikzpicture}
	\begin{pgfonlayer}{nodelayer}
		\node [style=plauette] (pl11) at (-9.25, 4) {};
        \node [style=plauette] (pl12) at (-9.25, 1.4) {};
        \node [style=plauette] (pl13) at (-9.25, -1.2) {};
        \node [style=plauette] (pl14) at (-9.25, -3.8) {};
        \node [style=plauette] (pl15) at (-7, 2.7) {};
        \node [style=plauette] (pl16) at (-7, 0.1) {};
        \node [style=plauette] (pl17) at (-7, -2.5) {};
        \node [style=plauette] (pl18) at (-4.75, 4) {};
        \node [style=plauette] (pl19) at (-4.75, 1.4) {};
        \node [style=plauette] (pl110) at (-4.75, -1.2) {};
        \node [style=plauette] (pl111) at (-2.5, 2.7) {};
        \node [style=plauette] (pl112) at (-2.5, 0.1) {};
        \node [style=plauetterotate] (pr11) at (-7.75, 4) {};
        \node [style=plauetterotate] (pr12) at (-7.75, -3.8) {};
        \node [style=plauetterotate] (pr13) at (-3.25, 4) {};

        \node [style=permred] (n11) at (-10.6, 4) {};
        \node [style=permred] (n12) at (-10.6, 1.4) {};
        \node [style=permred] (n13) at (-10.6, -1.2) {};
        \node [style=permred] (n14) at (-10.6, -3.8) {};
        \node [style=permred] (n15) at (-8.5, 2.7) {};
        \node [style=permred] (n16) at (-8.5, 0.1) {};
        \node [style=permred] (n17) at (-8.5, -2.5) {};
        \node [style=permred] (n18) at (-6.2, 4) {};
        \node [style=permred] (n19) at (-6.2, 1.4) {};
        \node [style=permred] (n110) at (-6.2, -1.2) {};
        \node [style=permred] (n111) at (-6.2, -3.8) {};
        \node [style=permred] (n112) at (-4, 2.7) {};
        \node [style=permred] (n113) at (-4, 0.1) {};
        \node [style=permred] (n114) at (-4, -2.5) {};

        \node [style=plauette] (pl21) at (-0.25, 4) {};
        \node [style=plauette] (pl22) at (-0.25, 1.4) {};
        \node [style=plauette] (pl23) at (-0.25, -1.2) {};
        \node [style=plauette] (pl24) at (-0.25, -3.8) {};
        \node [style=plauette] (pl25) at (2, 2.7) {};
        \node [style=plauette] (pl26) at (2, 0.1) {};
        \node [style=plauette] (pl27) at (2, -2.5) {};
        \node [style=plauette] (pl28) at (4.25, 1.4) {};
        \node [style=plauette] (pl29) at (4.25, -1.2) {};
        \node [style=plauetterotate] (pr21) at (1.25, 4) {};
        \node [style=plauetterotate] (pr22) at (1.25, -3.8) {};

        \node [style=permred] (n21) at (-1.7, 4) {};
        \node [style=permred] (n22) at (-1.7, 1.4) {};
        \node [style=permred] (n23) at (-1.7, -1.2) {};
        \node [style=permred] (n24) at (-1.7, -3.8) {};
        \node [style=permred] (n25) at (0.5, 2.7) {};
        \node [style=permred] (n26) at (0.5, 0.1) {};
        \node [style=permred] (n27) at (0.5, -2.5) {};
        \node [style=permred] (n28) at (2.7, 4) {};
        \node [style=permred] (n29) at (2.7, 1.4) {};
        \node [style=permred] (n210) at (2.7, -1.2) {};
        \node [style=permred] (n211) at (2.7, -3.8) {};
        \node [style=permred] (n212) at (5, 2.7) {};
        \node [style=permred] (n213) at (5, 0.1) {};
        \node [style=permred] (n214) at (5, -2.5) {};
	\end{pgfonlayer}
	\begin{pgfonlayer}{edgelayer}
		\draw[purple] (n111.center) to (n114.center);
		\draw[purple] (n114.center) to (n23.center);
		\draw[purple] (n28.center) to (n212.center);
        \draw[purple] (n211.center) to (n214.center);
	\end{pgfonlayer}
\end{tikzpicture}
}
    \label{eq:CP_brickwork_tri},
\ee
where the constant factor $1/20$ follows from summation of the last column of blue spins in Eq.~\eqref{eq:CP_brickwork_hex}.
Since the domain wall pattern has one-to-one correspondence to the spin configurations, we can rewrite the collision probability as
\be
    Z(t) = 2^n \left(\frac{1}{20}\right)^{N/2-1} \sum_{{\rm dw} \in \scalebox{0.1}{}} w({\rm dw}),
\ee
where $w(\rm dw)$ is the weight of domain wall. According to the Eq.~\eqref{eq:plauette}, there are three types of domain wall that are allowed in the triangular lattice in Eq.~\eqref{eq:CP_brickwork_tri}, listed as below.
\begin{itemize}
    \item Type I: A path starting from the left boundary of a unitary and ending at the right boundary of a unitary.
    \item Type II: A path starting from and ending at the right boundary of the same unitary.
    \item Type III: A path starting from the right boundary of a unitary and ending at the right boundary of another following unitary.
\end{itemize}
Note that Type III domain wall only appears in HRCS while others also appear in conventional RCS. Taking the three types of domain wall into account, we can upper bound the summation of domain wall weights as
\begin{align}
     \sum_{{\rm dw} \in \scalebox{0.1}{}} w({\rm dw}) \le \left(1 + \sum_{{\rm Type\:I\: dw}} w({\rm dw}) + \sum_{{\rm Type\:III}} w({\rm dw})\right)\left(1 + \sum_{{\rm Type\:II\:dw}} w({\rm dw})\right).
\end{align}

For Type I domain wall, as the weight scales as $(2/5)^\ell$ where $\ell$ is the total number of gate layers that the domain wall passes, the maximum weight for each domain wall is upper bounded by $(2/5)^L$ where $L$ is the number of gate layers in each step. The total weight of Type I domain wall is
\begin{align}
    \sum_{\text{Type I}} w(\rm dw) &\le \sum_{\text{Type I:}\ell({\rm dw})=L} w(\rm dw)\\
    &= t\sum_{m=1}^{N/2} \binom{N/2}{m} \left(\frac{2}{5}\right)^{L m} 2^{Lm}\\
    &= t \sum_{m=1}^{N/2} \binom{N/2}{m} \left(\frac{4}{5}\right)^{L m} \\
    &= t \left[\left(1 + \left(\frac{4}{5}\right)^L\right)^{N/2} - 1\right]\\
    &\le t\left[\exp\left(\frac{N}{2} \left(\frac{4}{5}\right)^L\right) - 1\right].
    \label{eq:weight_1}
\end{align}
For Type III domain wall, the path length is strictly larger than Type I domain wall when passing through the same number of steps, and thus we can also bound it by
\be
    \sum_{\text{Type III}} w({\rm dw}) \le (t-1)\left[\exp\left(\frac{N}{2} \left(\frac{4}{5}\right)^L\right) - 1\right].
    \label{eq:weight_3}
\ee
For Type II domain wall, we instead consider the infinite layer limit of $L\to \infty$, and only Type II domain walls can survive and contribute to the collision probability. We can thus bound it as
\be
    1 + \sum_{\text{Type II dw}} w({\rm dw}) \le \frac{20^{N/2-1}}{2^n} Z_{\rm HRCS}(t),
    \label{eq:weight_2}
\ee
where $Z_{\rm HRCS}(t)$ is the collision probability of HRCS with Haar random unitaries in Eq.~\eqref{eq:HRCS_CP_spt_app}. Finally, we have the bound for the brickwork HRCS collision probability as
\begin{align}
    Z(t) &\le 2^n \left(\frac{1}{20}\right)^{N/2-1} \left[1+t\exp\left(\frac{N}{2} \left(\frac{4}{5}\right)^L\right) - t + (t-1)\exp\left(\frac{N}{2} \left(\frac{4}{5}\right)^L\right) - (t-1)\right]\frac{20^{N/2-1}}{2^n} Z_{\rm HRCS}(t)\\
    &= \left[(2t-1)\exp\left(\frac{N}{2} \left(\frac{4}{5}\right)^L\right) - (2t-2)\right]Z_{\rm HRCS}(t),
\end{align}
which is the result in Lemma~\ref{lemma:HRCS_CP_app}.
A simple sanity check of $t=1$ recovers the result of conventional RCS in Ref.~\cite{barak2020spoofing}.
\end{proof}

At the end of this section, we introduce another estimation for the circuit depth utilizing the shallow random unitaries. We begin with the analysis on CP. Suppose each unitary $U$ in HRCS is independently sampled from ensemble $\calE$ to be $\varepsilon$-approximate $2$-design, we can formally define it by the relative error
\be
    (1-\varepsilon)\Phi_{\rm H}(\cdot) \preceq \Phi_\calE (\cdot) \preceq (1+\varepsilon)\Phi_{\rm H}(\cdot),
\ee
where $\Phi(\cdot):= \E_{U\sim \calE} U^{\otimes 2} (\cdot) U^\dagger{}^{\otimes 2}$ is the two-fold twirling channel, and $\Phi_{\rm H}(\cdot)$ is the $2$-fold twirling with Haar random unitaries. Here $\Phi^\prime \preceq \Phi$ denotes $\Phi-\Phi'$ is a completely-positive map. Then we can upper bound the ensemble-averaged CP as 
\begin{align}
    Z(t) &= 2^n \E \tr(U_t \cdots U_1 \ketbra{0}{0}_{A\bm B} U_1^\dagger \cdots U_t^\dagger \ketbra{0}{0}_{A \bm B})^2\\
    &= 2^n \tr(\Phi_t \circ \cdots \Phi_1(|\ketbra{0}{0}_{A\bm B}^{\otimes 2})\ketbra{0}{0}_{A \bm B}^{\otimes 2})\\
    &\le (1+\varepsilon)^t Z_{\rm HRCS}(t)\\
    &\le \exp\left[t\left(\frac{1}{d_A} + \varepsilon\right)\right] Z_{\rm H}(n).
\end{align}
In order to have $\epsilon$-approximate Haar value CP, we require the relative error of unitary $2$-design in each step to be
\be
    \varepsilon \le \ln(1+\epsilon)/t - 1/d_A.
\ee
We next utilize the construction of shallow random unitary in Ref.~\cite{schuster2025random}, where for $\varepsilon$-approximate $K$-design on $N$ qubits, it requires the unitary depth to be $L\sim K{\rm poly}(\ln K)\ln(N/\varepsilon)$. For minimum number of physical qubits, we require the relative error to be $\varepsilon \sim 1/n$ and thus $L^\star \sim \ln(n)$. Similarly, for $K$-th order power sums, suppose every unitary is $\varepsilon$-approximate $K$-design, we have $Z^{(K)}(t) \le (1+\varepsilon)^t Z_{\rm HRCS}^{(K)}(t) \le \exp\left[t\varepsilon + tK(K-1)/(2d_A)\right]Z_{\rm H}^{(K)}(n)$. The relative error is thus required to be $\varepsilon \le \ln(1+\epsilon)/t - K(K-1)/(2d_A)$, which is similar to the case of CP in spite of a reduction due to higher order $K$. Therefore, the unitary depth can still scale logarithmically with number of classical bits $\sim \ln(n)$ but with polynomial amplification factor of ${\rm poly}(\ln K)$.

\section{Marginal sampling in HRCS}
\label{app:marginal_sampling}

In this section, we study the spatial and temporal marginal sampling in HRCS, referring to the distribution of $P[x(t)]$ and $P[\bmz(t)]$, and we show that both of them are nearly uniform in terms of CP.

The main results are the following, with detailed proofs in later sections.
\begin{theorem}
\label{HRCS_CP_sp_te}
For holographic random circuit sampling with each unitary $U_t$ satisfying $2$-design, the ensemble-averaged collision probability for marginal spatial sampling with $P[x(t)]$ 
    \begin{align}
        Z_{\rm Sp}(t) 
        &= \frac{d_A d_B+1}{d_A^2 d_B+1} + \frac{(d_A-1) (d_A d_B-1)}{(d_A+1) \left(d_A^2 d_B+1\right)} \left(\frac{\left(d_A^2-1\right) d_B}{d_A^2 d_B^2-1}\right)^t  \label{eq:HRCS_CP_sp} \\
        &= Z_{\rm uni}\left(N_A\right)\left(1 + \frac{1}{d_A d_B} + \frac{1}{d_B^t}-\frac{2+d_B^{-1}}{d_Ad_B^t}\right) +  \calO\left(\frac{1}{d_A^3}\right), \label{eq:HRCS_CP_sp_asymp}
    \end{align}
    where $d_A, d_B = 2^{N_A}, 2^{N_B}$ are the Hilbert space dimensions of system and bath, and the second line holds for large system limit $d_A \gg 1$.
    For marginal temporal sampling with $P[\bmz(t)]$, we have
    \begin{align}
        Z_{\rm Te}(t) &= \left(\frac{d_A+1}{d_A d_B + 1}\right)^t      \label{eq:HRCS_CP_te} \\
        &= Z_{\rm uni}\left(tN_B\right)\exp\left[\frac{t(d_B-1)}{d_A d_B} + \calO\left(\frac{1}{d_A^2}\right)\right].
        \label{eq:HRCS_CP_te_asymp}
    \end{align}
\end{theorem}
Firstly, a simple sanity check shows that at $t=1$, $Z_{\rm Sp}(t=1) = (1+d_B)/(1+d_A d_B)$ and $Z_{\rm Te}(t=1)=(1+d_A)/(1+d_A d_B)$, which equals the CP of sampling only on the system or bath of Haar random states (see proof in~\ref{app:preliminary}). 
With increasing number of temporal steps $t$, the leading order approximation in Eq.~\eqref{eq:HRCS_CP_sp_asymp} indicates an exponential convergence of $2^{-tN_B}$ toward the uniform distribution with a finite-size correction of $1/(d_A d_B)$. The scaling of convergence time $\sim 1/N_B$ is identical to the convergence time for first-order state design in holographic deep thermalization~\cite{zhang2025holographic} and quantum information lifetime in unmonitored dynamics~\cite{zhang2025scaling}. 
On the other hand, the temporal sampling in HRCS is also close to uniform distribution on $tN_B$ qubits but with a dynamical deviation of $\exp(t/d_A)$ shown by the asymptotic result of Eq.~\eqref{eq:HRCS_CP_te_asymp}. Therefore we expect the approximate to uniform distribution can hold till exponential temporal steps. The CP of temporal sampling $Z_{\rm Te}(t)$ may also serve as a benchmark for mid-circuit measurements. 

We can also formally introduce $\epsilon$-approximate uniform distribution for both spatial and temporal sampling.
For $\epsilon$-approximate uniform CP in spatial sampling as $Z_{\rm Sp}(t)\le (1+\epsilon)Z_{\rm uni}(N_A)$, we require the temporal steps to be at least 
\be
    t \ge \tau_{\rm Sp} \coloneqq \log_2(1/\epsilon)/N_B + \calO\left(2^{-N_A}\right).
    \label{eq:tau_sp}
\ee
For temporal sampling to satisfy $\epsilon$-approximate uniform CP, the maximum temporal steps is
\be
    t \le \tau_{\rm Te} \coloneqq \frac{d_A d_B}{d_B -1}\ln(1+\epsilon) \simeq 2^{N_A}\ln(1+\epsilon).
    \label{eq:tau_te}
\ee
The approximation to uniform CP within an exponential number of temporal steps provides evidence to support the generation of random states with nearly uniform probability in holographic deep thermalization~\cite{zhang2025holographic}. Meanwhile, beyond the threshold time step $\tau_{\rm Te}$, the CP of temporal sampling quickly becomes lack of AC.

\subsection{Proof of Theorem~\ref{HRCS_CP_sp_te}}

We will divide Theorem~\ref{HRCS_CP_sp_te} into two parts, including the spatial part (Theorem~\ref{HRCS_CP_sp_app}) and and temporal part (Theorem~\ref{HRCS_CP_te_app}).

The proofs still rely on the equivalent expansion though with different boundary conditions (see Supplementary Fig.~\ref{fig:proof}(b) and (c)).
\begin{theorem}
\label{HRCS_CP_sp_app}
(Theorem~\ref{HRCS_CP_sp_te}, spatial sampling part)
For holographic random circuit sampling with each unitary $U_t$ satisfying $2$-design, the ensemble-averaged collision probability for marginal spatial sampling at step $t\ge 1$  is
    \begin{align}
        Z_{\rm Sp}(t) &= \frac{d_A d_B+1}{d_A^2 d_B+1} + \frac{(d_A-1) (d_A d_B-1)}{(d_A+1) \left(d_A^2 d_B+1\right)} \left(\frac{\left(d_A^2-1\right) d_B}{d_A^2 d_B^2-1}\right)^t  \label{eq:HRCS_CP_sp_app} \\
        &= Z_{\rm uni}\left(N_A\right)\left(1 + \frac{1}{d_A d_B} + \frac{1}{d_B^t}-\frac{2+d_B^{-1}}{d_Ad_B^t}\right) +  \calO\left(\frac{1}{d_A^3}\right), \label{eq:HRCS_CP_sp_asymp_app}
    \end{align}
    where $d_A, d_B = 2^{N_A}, 2^{N_B}$ are the Hilbert space dimension of system and bath, and the second line holds for large system limit $d_A \gg 1$.
\end{theorem}

\begin{proof}
    Following the equivalent expansion in Supplementary Fig.~\ref{fig:proof}b, we first write out the marginal probability of sampling $x(t)$ in the output state as
    \be
        P_U[x(t)] = \tr\left(U\ketbra{\bm 0}{\bm 0}_{A\bm B} U^\dagger \ketbra{x(t)}{x(t)}_{A} \otimes \bI_{\bm B}\right).
        \label{eq:pr_sp}
    \ee
    The ensemble average of the square of the marginal probability becomes
    \begin{align}
        \E_{U\in \rm Haar}\left[p_U[x(t)]^2\right] &= \E_{U\in \rm Haar} \tr\left(U^{\otimes 2}\ketbra{\bm 0}{\bm 0}_{A\bm B}^{\otimes 2} U^\dagger{}^{\otimes 2} \ketbra{x(t)}{x(t)}_{A}^{\otimes 2} \otimes \bI_{\bm B}^{\otimes 2}\right)\\
        &= \bbra{x(t)}_{A}\bbra{\hat{e}}_{\bm B} \E_{U\in \rm Haar}\left[U^{\otimes 2} \otimes U^*{}^{\otimes 2}\right] \kett{\bm 0}_{A\bm B}.
        \label{eq:pr_sp_square}
    \end{align}
    Now we evaluate Eq.~\eqref{eq:pr_sp_square} step by step from $U_1$ to $U_t$. 

    For the first step involving $U_1$, we have
    \begin{align}
        &\bbra{\hat{e}}_{B_1} \E_{U_1\in \rm Haar}\left[U_1^{\otimes 2} \otimes U_1^*{}^{\otimes 2}\right] \kett{\bm 0}_{AB_1}\nonumber\\
        &= \frac{1}{d^2-1}\left(\bbrakett{\hat{e}}{\hat{e}}_{B_1}\bbrakett{\hat{e}}{\bm 0}_{AB_1}\kett{\hat{e}}_A +\bbrakett{\hat{e}}{\hat{\tau}}_{B_1}\bbrakett{\hat{\tau}}{\bm 0}_{AB_1}\kett{\hat{\tau}}_{A}\right) \nonumber\\
        &\quad - \frac{1}{d(d^2-1)} \left(\bbrakett{\hat{e}}{\hat{e}}_{B_1}\bbrakett{\hat{\tau}}{\bm 0}_{AB_1}\kett{\hat{e}}_A +\bbrakett{\hat{e}}{\hat{\tau}}_{B_1}\bbrakett{\hat{e}}{\bm 0}_{AB_1} \kett{\hat{\tau}}_A \right)\\
        &= \frac{1}{d^2-1}\left(d_B^2 \kett{\hat{e}}_A +d_B\kett{\hat{\tau}}_{A}\right) - \frac{1}{d(d^2-1)} \left(d_B^2\kett{\hat{e}}_A +d_B\kett{\hat{\tau}}_{A}\right)\\
        &= \frac{d_B}{d_A(d_A d_B +1)} \kett{\hat{e}}_A + \frac{1}{d_A(d_A d_B +1)}\kett{\hat{\tau}}_{A}.
        \label{eq:pr_sp_square_eq0}
    \end{align}
    We next focus on the evolution of each of them separately under Haar random unitaries.
    \begin{align}
        &\bbra{\hat{e}}_{B_2} \E_{U_2\in \rm Haar}\left[U_2^{\otimes 2} \otimes U_2^*{}^{\otimes 2}\right] \kett{\hat{e}}_A \kett{\bm 0}_{B_2} \nonumber \\
        &= \frac{1}{d^2-1}\left(\bbrakett{\hat{e}}{\hat{e}}_{B_2} \bbrakett{\hat{e}}{\hat{e}}_A \bbrakett{\hat{e}}{\bm 0}_{B_2}\kett{\hat{e}}_A + \bbrakett{\hat{e}}{\hat{\tau}}_{B_2} \bbrakett{\hat{\tau}}{\hat{e}}_{A}\bbrakett{\hat{\tau}}{\bm 0}_{B_2}\kett{\hat{\tau}}_{A}\right) \nonumber\\
        &\quad - \frac{1}{d(d^2-1)} \left(\bbrakett{\hat{e}}{\hat{e}}_{B_2}\bbrakett{\hat{\tau}}{\hat{e}}_{A}\bbrakett{\hat{\tau}}{\bm 0}_{B_2}\kett{\hat{e}}_A +\bbrakett{\hat{e}}{\hat{\tau}}_{B_2}\bbrakett{\hat{e}}{\hat{e}}_{A}\bbrakett{\hat{e}}{\bm 0}_{B_2} \kett{\hat{\tau}}_A \right)\\
        &= \frac{1}{d^2-1}\left( d_A^2  d_B^2 \kett{\hat{e}}_A + d_A d_B\kett{\hat{\tau}}_{A}\right) - \frac{1}{d(d^2-1)} \left(d_A d_B^2 \kett{\hat{e}}_A +d_A^2 d_B\kett{\hat{\tau}}_A \right)\\
        &= \frac{d_B\left(d_A^2 d_B - 1\right)}{d_A^2 d_B^2-1} \kett{\hat{e}}_A + \frac{d_A (d_B-1)}{d_A^2 d_B^2-1} \kett{\hat{\tau}}_A.
        \label{eq:pr_sp_square_eq1}
    \end{align}
    Similarly, for $\tau$, we have
    \begin{align}
        &\bbra{\hat{e}}_{B_2} \E_{U_2\in \rm Haar}\left[U_2^{\otimes 2} \otimes U_2^*{}^{\otimes 2}\right] \kett{\hat{\tau}}_A \kett{\bm 0}_{B_2} \nonumber \\
        &= \frac{1}{d^2-1}\left(\bbrakett{\hat{e}}{\hat{e}}_{B_2} \bbrakett{\hat{e}}{\hat{\tau}}_A \bbrakett{\hat{e}}{\bm 0}_{B_2}\kett{\hat{e}}_A + \bbrakett{\hat{e}}{\hat{\tau}}_{B_2} \bbrakett{\hat{\tau}}{\hat{\tau}}_{A}\bbrakett{\hat{\tau}}{\bm 0}_{B_2}\kett{\hat{\tau}}_{A}\right) \nonumber\\
        &\quad - \frac{1}{d(d^2-1)} \left(\bbrakett{\hat{e}}{\hat{e}}_{B_2}\bbrakett{\hat{\tau}}{\hat{\tau}}_{A}\bbrakett{\hat{\tau}}{\bm 0}_{B_2}\kett{\hat{e}}_A +\bbrakett{\hat{e}}{\hat{\tau}}_{B_2}\bbrakett{\hat{e}}{\hat{\tau}}_{A}\bbrakett{\hat{e}}{\bm 0}_{B_2} \kett{\hat{\tau}}_A \right)\\
        &= \frac{1}{d^2-1}\left( d_A d_B^2 \kett{\hat{e}}_A + d_A^2 d_B\kett{\hat{\tau}}_{A}\right) - \frac{1}{d(d^2-1)} \left(d_A^2 d_B^2 \kett{\hat{e}}_A +d_A d_B\kett{\hat{\tau}}_A \right)\\
        &= \frac{d_Ad_B(d_B-1)}{d_A^2 d_B^2-1} \kett{\hat{e}}_A + \frac{d_A^2 d_B - 1}{d_A^2 d_B^2-1} \kett{\hat{\tau}}_A.
        \label{eq:pr_sp_square_eq2}
    \end{align}

    We still utilize the $2$-dimensional linear system representation as $\ket{\underline{0}}\coloneqq \kett{\hat{e}}_A, \ket{\underline{1}}\coloneqq \kett{\hat{\tau}}_A$. 
    The mapping in Eq.~\eqref{eq:pr_sp_square_eq1} and Eq.~\eqref{eq:pr_sp_square_eq2} can be represented by a $2\times 2$ matrix as
    \be
        M_2 = \frac{1}{d_A^2 d_B^2 - 1}\begin{pmatrix}
            d_B\left(d_A^2 d_B-1\right) & d_A d_B(d_B-1)\\
            d_A(d_B-1)  & d_A^2 d_B-1
        \end{pmatrix},
        \label{eq:M2_def}
    \ee
    and the ensemble-averaged squared probability in Eq.~\eqref{eq:pr_sp_square} becomes
    \begin{align}
        \E_{U\in \rm Haar}\left[P_U[x(t)]^2\right] &= \frac{d_B}{d_A(d_A d_B+1)}\left(\bra{\underline{0}}+ \bra{\underline{1}}\right) M_2^{t-1}\ket{\underline{0}} + \frac{1}{d_A(d_A d_B+1)}\left(\bra{\underline{0}}+ \bra{\underline{1}}\right) M_2^{t-1}\ket{\underline{1}} \\
        &= \frac{d_A d_B+1}{d_A\left(d_A^2 d_B+1\right)} + \frac{(d_A-1)(d_A d_B-1)}{d_A(d_A + 1)\left(d_A^2 d_B+1\right)}\left(\frac{\left(d_A^2-1\right) d_B}{d_A^2 d_B^2 - 1}\right)^t.
    \end{align}
    Therefore, the ensemble-averaged CP is
    \begin{align}
        Z_{\rm Sp}(t) &= \sum_{x(t)} \E_{U\in \rm Haar}\left[P_U[ x(t)]^2\right] \\
        &= \frac{d_A d_B+1}{d_A^2 d_B+1} + \frac{(d_A-1)(d_A d_B-1)}{(d_A + 1)\left(d_A^2 d_B+1\right)}\left(\frac{\left(d_A^2-1\right) d_B}{d_A^2 d_B^2 - 1}\right)^t,
    \end{align}
    which proves Eq.~\eqref{eq:HRCS_CP_sp_app} in Theorem~\ref{HRCS_CP_sp_app}. In the large system limit of $d_A \gg 1$, we can directly have the asymptotic form in Eq.~\eqref{eq:HRCS_CP_sp_asymp_app} with $Z_{\rm uni}(N_A) = 2^{-N_A}$ to be the CP for uniform distribution.
\end{proof}

For the $\epsilon$-approximate uniform CP of spatial marginal sampling in HRCS, defined by $Z_{\rm Sp}(t) \le (1+\epsilon)Z_{\rm uni}(N_A)$. As $Z_{\rm Sp}(t)$ monotonically decays with $t$, we introduce an upper bound for simplification
\begin{align}
    Z_{\rm Sp}(t) &= Z_{\rm uni}\left(N_A\right)\left(1 + \frac{1}{d_A d_B} + \frac{1}{d_B^t}-\frac{2+d_B^{-1}}{d_Ad_B^t}\right) +  \calO\left(\frac{1}{d_A^3}\right) \nonumber\\
    &\le Z_{\rm uni}\left(N_A\right)\left(1 + \frac{1}{d_A d_B} + \frac{1}{d_B^t}\right) +  \calO\left(\frac{1}{d_A^3}\right),
\end{align}
which results in the critical number of temporal steps as
\begin{align}
    t \ge \tau_{\rm Sp} \coloneqq \frac{\ln\left(\frac{d_A d_B}{d_A d_B \epsilon -1}\right)}{\ln(d_B)} = \frac{\ln(1/\epsilon) +\calO(1/d_A)}{\ln(d_B)}.
    \label{eq:tau_sp_app}
\end{align}

\begin{theorem}
\label{HRCS_CP_te_app}
(Theorem~\ref{HRCS_CP_sp_te}, temporal sampling part)
For holographic random circuit sampling with each unitary $U_t$ satisfying $2$-design, the ensemble-averaged collision probability for marginal temporal sampling at step $t\ge 1$  is
\be
   Z_{\rm Te}(t) = \left(\frac{d_A+1}{d_A d_B + 1}\right)^t,      \label{eq:HRCS_CP_te_app} 
\ee
where $d_A=2^{N_A}, d_B=2^{N_B}$ are Hilbert space dimensions of system and bath. In the large system limit of $d_A \gg 1$, we have
\begin{align}
    Z_{\rm Te}(t) = Z_{\rm uni}\left(tN_B\right)\exp\left[\frac{t(d_B-1)}{d_A d_B} + \calO\left(\frac{1}{d_A^2}\right)\right].
    \label{eq:HRCS_CP_te_asymp_app}
\end{align}
\end{theorem}


\begin{proof}
As illustrated by the bath expansion in Fig.~\ref{fig:proof}c, the marginal probability of sampling $\bmz(t)$ in the output state
    \be
        P_U[\bmz(t)] = \tr\left(U\ketbra{\bm 0}{\bm 0}_{A\bm B} U^\dagger \bI_{A}\otimes \ketbra{\bmz(t)}{\bmz(t)}_{\bm B} \right),
        \label{eq:pr_te}
    \ee
    where $\bmz(t) = (z_1, z_2, \dots, z_t)$.
    The ensemble average of the square of the marginal probability becomes
    \begin{align}
        \E_{U\in \rm Haar}\left[P_U[\bmz(t)]^2\right] &= \E_{U\in \rm Haar} \tr\left(U^{\otimes 2}\ketbra{\bm 0}{\bm 0}_{A\bm B}^{\otimes 2} U^\dagger{}^{\otimes 2} \bI_{A}^{\otimes 2} \ketbra{\bmz(t)}{\bmz(t)}_{\bm B}^{\otimes 2} \right)\\
        &= \bbra{\hat{e}}_{A}\bbra{\bmz(t)}_{\bm B} \E_{U\in \rm Haar}\left[U^{\otimes 2} \otimes U^*{}^{\otimes 2}\right] \kett{\bm 0}_{A\bm B}.
        \label{eq:pr_te_square}
    \end{align}
    Now we evaluate Eq.~\eqref{eq:pr_te_square} step by step from $U_1$ to $U_t$. 

    For the first step involving $U_1$, we have
    \begin{align}
        \bbra{z_1}_{B_1} \E_{U_1\in \rm Haar}\left[U_1^{\otimes 2} \otimes U_1^*{}^{\otimes 2}\right] \kett{\bm 0}_{AB_1} = \frac{1}{d(d +1)} \left(\kett{\hat{e}}_A + \kett{\hat{\tau}}_{A}\right),
        \label{eq:pr_te_square_eq0}
    \end{align}
    which is already derived in Eq.~\eqref{eq:pr_spt_square_eq0}.
    We next focus on the evolution of each of them separately under Haar random unitary (see derivation in Eqs.~\eqref{eq:pr_spt_square_eq1},~\eqref{eq:pr_spt_square_eq2}).
    \begin{align}
        \left\{ \begin{array}{ll}
        \bbra{z_2}_{B_2} \E_{U_2\in \rm Haar}\left[U_2^{\otimes 2} \otimes U_2^*{}^{\otimes 2}\right] \kett{\hat{e}}_A \kett{\bm 0}_{B_2} &= \frac{d_A^2 d_B - 1}{d_B\left(d_A^2 d_B^2-1\right)} \kett{\hat{e}}_A + \frac{d_A (d_B-1)}{d_B\left(d_A^2 d_B^2-1\right)} \kett{\hat{\tau}}_A,\\
        \bbra{z_2}_{B_2} \E_{U_2\in \rm Haar}\left[U_2^{\otimes 2} \otimes U_2^*{}^{\otimes 2}\right] \kett{\hat{\tau}}_A \kett{\bm 0}_{B_2} &=\frac{d_A (d_B-1)}{d_B\left(d_A^2 d_B^2-1\right)} \kett{\hat{e}}_A + \frac{d_A^2 d_B - 1}{d_B\left(d_A^2 d_B^2-1\right)}\kett{\hat{\tau}}_A
        .\end{array} \right.
        \label{eq:pr_te_square_eq}
    \end{align}

    With the $2$-dimensional linear system representation as $\ket{\underline{0}}\coloneqq \kett{\hat{e}}_A, \ket{\underline{1}}\coloneqq \kett{\hat{\tau}}_A$, the mapping of Eqs.~\eqref{eq:pr_te_square_eq} can be represented by the $2\times 2$ matrix in Eq.~\eqref{eq:M_def}
    and the ensemble-averaged squared probability in Eq.~\eqref{eq:pr_sp_square} becomes
    \begin{align}
        \E_{U\in \rm Haar}\left[P_U[\bmz(t)]^2\right] &= \frac{1}{d(d+1)}\left[\bra{\underline{0}}M^{t-1}\left(\ket{\underline{0}} + \ket{\underline{1}}\right) \bbrakett{\hat{e}}{\hat{e}}_A + \bra{\underline{1}} M^{t-1}\left(\kett{\underline{0}} + \ket{\underline{1}}\right) \bbrakett{\hat{e}}{\hat{\tau}}_A \right]\\
        &= \left(\frac{1 + d_A}{d_B(1+d_A d_B)}\right)^t.
    \end{align}
    Therefore, the ensemble-averaged CP is
    \begin{align}
        Z_{\rm Te}(t) &= \sum_{\bmz (t)} \E_{U\in \rm Haar}\left[P_U[ x(t)]^2\right] = \left(\frac{1 + d_A}{1+d_A d_B}\right)^t,
    \end{align}
    which proves Eq.~\eqref{eq:HRCS_CP_te_app} in Theorem~\ref{HRCS_CP_te_app}. 

    In the large system limit of $d_A \gg 1$, we have
    \begin{align}
        Z_{\rm Te}(t) &= \left(\frac{1 + d_A}{1+d_A d_B}\right)^t \nonumber\\
        &= \exp\left[t\ln\left(\frac{1+d_A}{1+d_A d_B}\right)\right]\\
        &= \exp\left[t\ln\left(\frac{1}{d_B}\right) + t\left(\frac{d_B-1}{d_A d_B}\right) + \calO\left(\frac{1}{d_A^2}\right)\right]\\
        &= Z_{\rm uni}(tN_B) \exp\left[t\left(\frac{d_B-1}{d_A d_B}\right) + \calO\left(\frac{1}{d_A^2}\right)\right],
    \end{align}
    which is the asymptotic result in Eq.~\eqref{eq:HRCS_CP_te_asymp_app}.
\end{proof}

For $\epsilon$-approximate uniform CP defined by $Z_{\rm Te}(t) \le (1+\epsilon)Z_{\rm uni}(tN_B)$, we require
\be
    t \le \tau_{\rm Te} \coloneqq \frac{d_B}{d_B - 1}d_A \ln(1+\epsilon).
\ee

\subsection{Temporal sampling per step in HRCS}

In this section, we focus on the marginal temporal sampling per step in HRCS, i.e., the distribution of $P[z_1], P[z_2], \dots, P[z_t]$, and we have the following theorem for corresponding collision probability.

\begin{theorem}
\label{HRCS_CP_te_per_step}
    {\rm (Temporal sampling per step in HRCS)}
    For holographic random circuit sampling with each unitary $U_t$ satisfying $2$-design, the ensemble-averaged collision probability for temporal sampling only at a single step $t\ge 1$  is
    \begin{align}
        Z_{{\rm Te},t}(t) = 
        \frac{d_A^2 + 1}{d_A^2 d_B + 1} + \frac{d_A (d_B - 1) (d_A d_B - 1)}{(d_A^2 d_B + 1)(d_A+1) d_B}\left(\frac{\left(d_A^2-1\right) d_B}{d_A^2 d_B^2-1}\right)^t,
        \label{eq:HRCS_CP_te_per_step}
    \end{align}
    where $d_A, d_B = 2^{N_A}, 2^{N_B}$ are the Hilbert space dimension of memory and bath system. In the large-system limit with $d_A \gg 1$, we have
    \begin{align}
        Z_{{\rm Te},t}(t) &= Z_{\rm uni}(N_B)\left[1 + \frac{1-d_B^{-1}}{d_A^2} + \frac{1-d_B^{-1}}{d_A d_B^{t-1}}\left(1-\frac{1+d_B^{-1}}{d_A}\right)\right] +\calO\left(\frac{1}{d_A^3}\right).
        \label{eq:HRCS_CP_te_per_step_asymp}
    \end{align}
\end{theorem}
Note that the CP only depends on the step index $t$ but independent of the total number of steps in HRCS as expected. For $\epsilon$-approximate uniform CP, the minimum temporal step is
\begin{align}
    t \ge \tau_{{\rm Te},t} \coloneqq \frac{1}{N_B} \log_2 \left(\frac{(d_B-1) ((d_A-1) d_B-1)}{d_B \left(d_A^2 \epsilon-1\right)+1}\right).
\end{align}
In particular, for $N_B \le \log _2\left(\frac{d_A^2 \epsilon + d_A-1}{d_A-1}\right)$, the above lower bound is trivial and the marginal sampling at arbitrary temporal step satisfies $\epsilon$-approximate uniform CP.

\begin{proof}
    We prove it utilizing the expansion shown in Supplementary Fig.~\ref{fig:proof}d. Without loosing generality, we focus on the marginal temporal sampling only at step $t$ assuming there are in total $T$ steps in HRCS, and the sampling probability thus becomes
    \be
        P_U[z_t] = \tr\left(U\ketbra{\bm 0}{\bm 0}_{A\bm B} U^\dagger \bI_{A}\otimes \bI_{B_1\shortto B_{t-1}}\otimes \ketbra{z_t}{z_t}_{B_t} \otimes \bI_{B_{t+1} \shortto B_T} \right),
        \label{eq:pr_te_per_step}
    \ee
    where $\bI_{B_1\shortto B_{t-1}} = \otimes_{j=1}^{t-1} \bI_{B_j}$ and $\bI_{B_{t+1} \shortto B_T}$ is defined similarly. The ensemble-averaged squared probability becomes
    \begin{align}
        \E_{U\in \rm Haar}\left[P_U[z_t]^2\right] &= \E_{U\in \rm Haar} \tr\left(U^{\otimes 2}\ketbra{\bm 0}{\bm 0}_{A\bm B}^{\otimes 2} U^\dagger{}^{\otimes 2} \bI_{A}^{\otimes 2}\otimes \bI_{B_1\shortto B_{t-1}}^{\otimes 2}\otimes \ketbra{z_t}{z_t}_{B_t}^{\otimes 2} \otimes \bI_{B_{t+1} \shortto B_T}^{\otimes 2} \right)\\
        &= \bbra{\hat{e}}_A \bbra{\hat{e}}_{B_1\shortto B_{t-1}} \bbra{z_t}_{B_t} \bbra{\hat{e}}_{B_{t+1}\shortto B_T} \E_{U \in \rm Haar}\left[U^{\otimes 2} \otimes U^*{}^{\otimes 2}\right] \kett{\bm 0}_{A \bm B},
        \label{eq:pr_te_per_step_square}
    \end{align}
    where $\kett{\hat{e}}_{B_1 \to B_{t-1}} = \otimes_{j=1}^{t-1} \kett{\hat{e}}_{B_j}$.
    
    For the first step unitary twirling with $U_1$, we have
    \begin{align}
        \bbra{\hat{e}}_{B_1} \E_{U_1\in \rm Haar}\left[U_1^{\otimes 2} \otimes U_1^*{}^{\otimes 2}\right] \kett{\bm 0}_{AB_1} = \frac{d_B}{d_A(d_A d_B +1)} \kett{\hat{e}}_A + \frac{1}{d_A(d_A d_B +1)}\kett{\hat{\tau}}_{A},
        \label{eq:pr_te_per_step_square_eq0}
    \end{align}
    which has been derived in Eq.~\eqref{eq:pr_sp_square_eq0}. The unitary twirling for $\kett{\hat{e}}$ and $\kett{\hat{\tau}}$ are
    \begin{align}
        \left\{ \begin{array}{ll}
        \bbra{\hat{e}}_{B_2} \E_{U_2\in \rm Haar}\left[U_2^{\otimes 2} \otimes U_2^*{}^{\otimes 2}\right] \kett{\hat{e}}_A \kett{\bm 0}_{B_2} &= \frac{d_B\left(d_A^2 d_B - 1\right)}{d_A^2 d_B^2-1} \kett{\hat{e}}_A + \frac{d_A (d_B-1)}{d_A^2 d_B^2-1} \kett{\hat{\tau}}_A, \\
        \bbra{\hat{e}}_{B_2} \E_{U_2\in \rm Haar}\left[U_2^{\otimes 2} \otimes U_2^*{}^{\otimes 2}\right] \kett{\hat{\tau}}_A \kett{\bm 0}_{B_2} &= \frac{d_Ad_B(d_B-1)}{d_A^2 d_B^2-1} \kett{\hat{e}}_A + \frac{d_A^2 d_B - 1}{d_A^2 d_B^2-1} \kett{\hat{\tau}}_A
        .\end{array} \right.
        \label{eq:pr_te_per_step_square_eq1}
    \end{align}
    Note that the above twirling result can be represented by $M_2$ in Eq.~\eqref{eq:M2_def}. Through the first $t-1$ steps of evolution, we reach
    \begin{align}
        &\E_{U \in \rm Haar} \tr_{B_1\shortto B_{t-1}}\left(U_{1\shortto t-1}^{\otimes 2}\ketbra{\bm 0}{\bm 0}_{A,B_1 \shortto B_{t-1}}^{\otimes 2} U_{1\shortto t-1}^\dagger{}^{\otimes 2} \bI_{B_1\shortto B_{t-1}}^{\otimes 2}\right) \nonumber\\
        &= \left(\frac{d_B}{d_A(d_A d_B + 1)}\braket{\underline{0}|M_2^{t-2}|\underline{0}} + \frac{1}{d_A(d_A d_B + 1)}\braket{\underline{0}|M_2^{t-2}|\underline{1}}\right)\kett{\hat{e}}_A \nonumber\\
        &\quad + \left(\frac{d_B}{d_A(d_A d_B + 1)}\braket{\underline{1}|M_2^{t-2}|\underline{0}} + \frac{1}{d_A(d_A d_B + 1)}\braket{\underline{1}|M_2^{t-2}|\underline{1}}\right) \kett{\hat{\tau}}_A\\
        &= \left[\frac{d_B}{d_A^2 d_B + 1}-\frac{(d_A d_B-1)^2 (d_A d_B+1) \left(\frac{\left(d_A^2-1\right) d_B}{d_A^2 d_B^2-1}\right)^t}{d_A d_B(d_A-1)(d_A+1)^2 \left(d_A^2 d_B+1\right)}\right] \kett{\hat{e}}_A \nonumber\\ 
        &\quad + \left[\frac{1}{d_A\left(d_A^2 d_B + 1\right)} + \frac{(d_A d_B-1)^2 (d_A d_B+1) \left(\frac{\left(d_A^2-1\right) d_B}{d_A^2 d_B^2-1}\right)^t}{(d_A-1) (d_A+1)^2 d_B\left(d_A^2 d_B+1\right)}\right] \kett{\hat{\tau}}_A.
        \label{eq:pr_te_per_step_square2}
    \end{align}

    Next we do the unitary twirling with $U_t$ for different boundary condition.
    \begin{align}
        \left\{ \begin{array}{ll}
        \bbra{z_t}_{B_t} \E_{U_t\in \rm Haar}\left[U_t^{\otimes t} \otimes U_t^*{}^{\otimes 2}\right] \kett{\hat{e}}_A \kett{\bm 0}_{B_t} &= \frac{d_A^2 d_B - 1}{d_B\left(d_A^2 d_B^2-1\right)} \kett{\hat{e}}_A + \frac{d_A (d_B-1)}{d_B\left(d_A^2 d_B^2-1\right)} \kett{\hat{\tau}}_A,\\
        \bbra{z_t}_{B_t} \E_{U_t\in \rm Haar}\left[U_t^{\otimes 2} \otimes U_t^*{}^{\otimes 2}\right] \kett{\hat{\tau}}_A \kett{\bm 0}_{B_t} &=\frac{d_A (d_B-1)}{d_B\left(d_A^2 d_B^2-1\right)} \kett{\hat{e}}_A + \frac{d_A^2 d_B - 1}{d_B\left(d_A^2 d_B^2-1\right)}\kett{\hat{\tau}}_A,
        \end{array} \right.
        \label{eq:pr_te_per_step_square_eq2}
    \end{align}
    which is the same as Eqs.~\eqref{eq:pr_te_square_eq}. Applying it to Eq.~\eqref{eq:pr_te_per_step_square2}, we have
    \begin{align}
        &\E_{U \in \rm Haar} \tr_{B_1\shortto B_{t}}\left(U_{1\shortto t}^{\otimes 2}\ketbra{\bm 0}{\bm 0}_{A,B_1 \shortto B_{t}}^{\otimes 2} U_{1\shortto t}^\dagger{}^{\otimes 2} \bI_{B_1\shortto B_{t}}^{\otimes 2}\right) \nonumber\\
        &= \frac{(1-d_A d_B) \left(\frac{\left(d_A^2-1\right) d_B}{d_A^2 d_B^2-1}\right)^t+d_A (d_A+1) d_B}{d_A (d_A+1) d_B^2 \left(d_A^2 d_B+1\right)} \kett{\hat{e}}_A + \frac{d_A (d_A d_B-1) \left(\frac{\left(d_A^2-1\right) d_B}{d_A^2 d_B^2-1}\right)^t+d_A+1}{d_A (d_A+1) d_B \left(d_A^2 d_B+1\right)} \kett{\hat{\tau}}_A.
    \end{align}

    The following unitary twirling from $U_{t+1}$ to $U_T$ is still described by Eq.~\eqref{eq:pr_te_per_step_square_eq1}, and we can finally derive the full expression of Eq.~\eqref{eq:pr_te_per_step_square} as
    \begin{align}
        \E_{U\in \rm Haar}\left[P_U[z_t]^2\right] &= \frac{d_A (d_B-1) (d_A d_B-1) }{d_B^2\left(d_A+1\right) \left(d_A^2 d_B+1\right)}\left(\frac{\left(d_A^2-1\right) d_B}{d_A^2 d_B^2-1}\right)^t + \frac{d_A^2+1}{d_B \left(d_A^2 d_B+1\right)},
    \end{align}
    and the CP becomes
    \be
        Z_{{\rm Te}, t}(t) = \sum_{z_t}\E_{U\in \rm Haar}\left[P_U[z_t]^2\right] = \frac{d_A (d_B-1) (d_A d_B-1) }{d_B\left(d_A+1\right) \left(d_A^2 d_B+1\right)}\left(\frac{\left(d_A^2-1\right) d_B}{d_A^2 d_B^2-1}\right)^t + \frac{d_A^2+1}{d_A^2 d_B+1}.
    \ee
\end{proof}

\subsection{Additional numerical results on marginal sampling in HRCS}
\label{sec:sp_t}

We verify the decay of CP in spatial sampling in Supplementary Fig.~\ref{fig:marginal}a. Specifically, the relative deviation of CP with respect to uniform one $Z_{\rm Sp}(t)/Z_{\rm uni}(N_A)-1$ reveals an exponential decay in early stage of $2^{-t N_B}$ (dashed lines) and later convergence to the finite-size correction of $1/d_A^2 d_B$. The threshold temporal steps for $\epsilon$-approximate uniform CP in Eq.~\eqref{eq:tau_sp} aligns with the numerical solutions (crosses) in Supplementary Fig.~\ref{fig:marginal}b. Moreover, we plot the PoP of marginal spatial sampling distribution $P[x(t)]$ in Supplementary Fig.~\ref{fig:marginal}c, which concentrates around the uniform distribution with probability $2^{-N_A}$.
Similar to the spatial sampling, we also see agreements between numerical results and theories for temporal sampling in Supplementary Fig.~\ref{fig:marginal}d-e. The PoP of marginal temporal sampling $P[\bmz(t)]$ also concentrates around the uniform distribution of dimension $t N_B$ though with finite spreading, shown in Supplementary Fig.~\ref{fig:marginal}f.

\begin{figure*}[t]
    \centering
    \includegraphics[width=\textwidth]{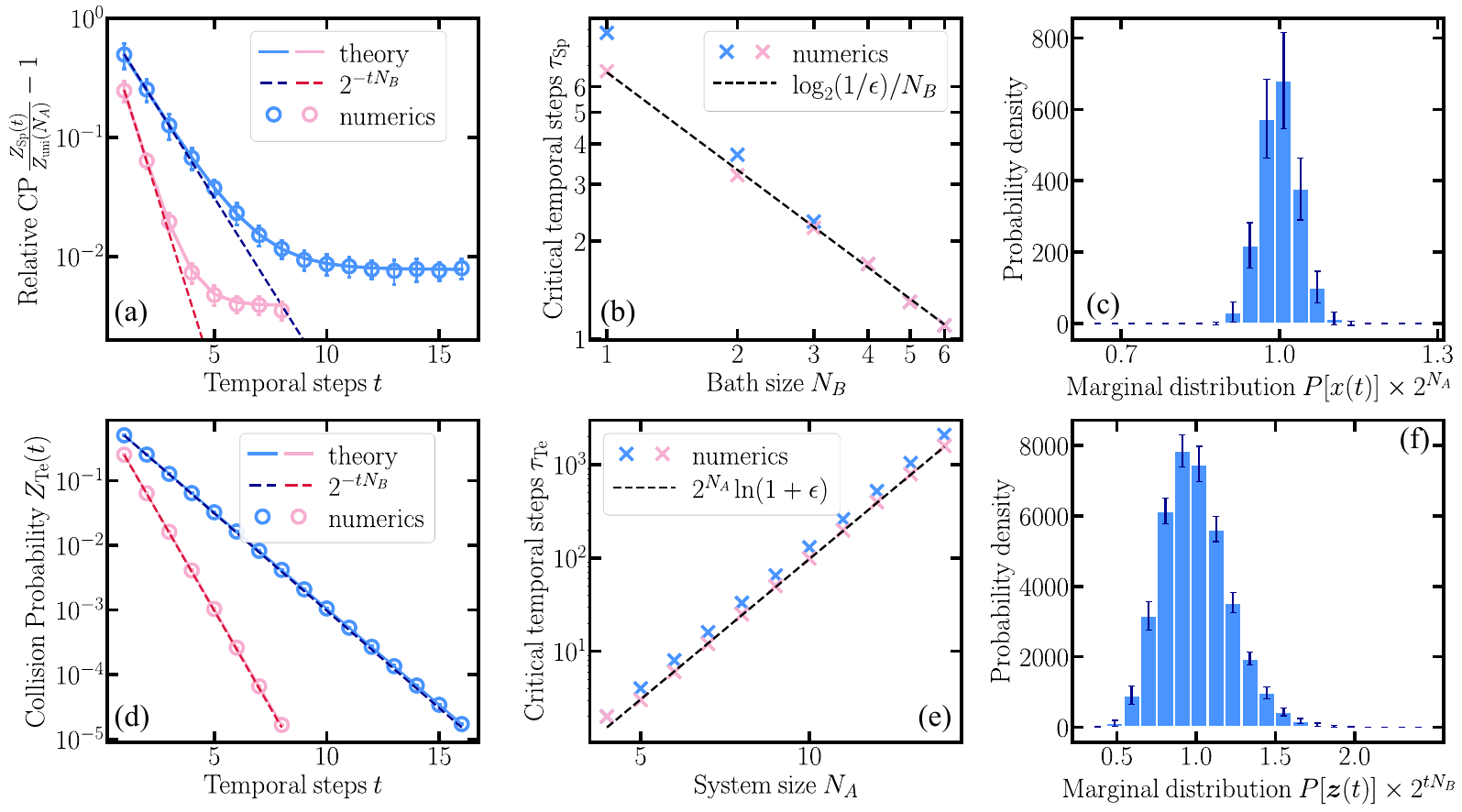}
    \caption{{\bf Marginal sampling in HRCS.} (a) Ensemble-averaged relative deviation of collision probability (CP) $Z_{\rm Sp}(t)/Z_{\rm uni}(N_A)-1$ for spatial sampling in HRCS versus temporal steps in a system of $N_A=6$, $N_B=1, 2$ (blue and orange dots) qubits. Solid lines represent theoretical result from Eq.~\eqref{eq:HRCS_CP_sp} in Theorem~\ref{HRCS_CP_sp_te}. Dark-colored dashed lines are universal convergence scaling of $2^{-tN_B}$.
    (b) Growth of critical temporal steps $\tau_{\rm Sp}$ versus bath size $N_B$. Blue and orange crosses represent numerical solutions for $Z_{\rm Sp}(t)=(1+\epsilon)Z_{\rm uni}(N_A)$ for $N_A=6, 12$ separately. The black line is Eq.~\eqref{eq:tau_sp}.
    (c) Ensemble-averaged probability density function of the spatial sampling distribution $P[x(t)]$ in HRCS of $N_A=6, N_B=4$ qubits at $t=3$. 
    (d) CP of temporal sampling $Z_{\rm Te}(t)$ versus temporal steps $t$. Light-colored solid lines are theory corresponding to Eq.~\eqref{eq:HRCS_CP_te} in Theorem~\ref{HRCS_CP_sp_te} and dark-colored dashed lines are CP for uniform distribution $Z_{\rm uni}(tN_B)$.
    (e) Critical temporal steps for temporal sampling $\tau_{\rm Te}$ versus system size $N_A$. Blue and orange crosses represent numerical solutions for $Z_{\rm Te}(t)=(1+\epsilon)Z_{\rm uni}(tN_B)$ for $N_B=2, 6$ separately. The black line is Eq.~\eqref{eq:tau_te}.
    (f) Ensemble-averaged probability density function of the temporal sampling distribution $P[\bmz(t)]$ in HRCS of $N_A=6, N_B=4$ qubits at $t=3$.
    Error bars in (a),(c),(d),(f) show the standard deviation over $50$ circuit instances. 
    }
    \label{fig:marginal}
\end{figure*}

\section{Noisy model of HRCS}
\label{app:noisy_theory}

\begin{figure}[t]
    \centering
    \includegraphics[width=0.5\textwidth]{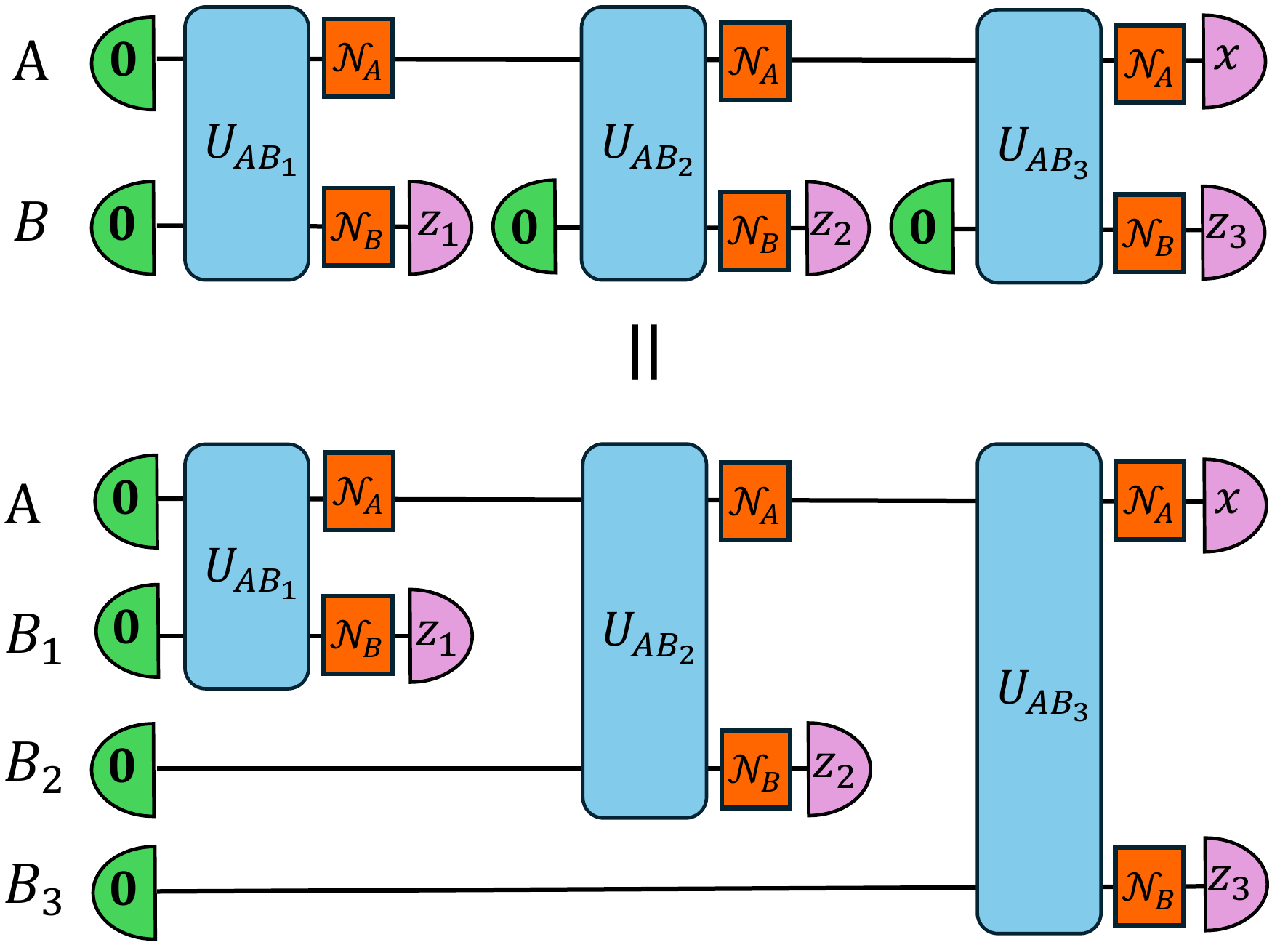}
    \caption{Circuit diagram for the expansion of bath system in noisy HRCS. $\calN_A, \calN_B$ represent effective depolarizing channel on system $A$ and bath $B$ with parameter $\gamma_A, \gamma_B$ separately. Here we show an example of $t=3$ steps.}
    \label{fig:noisy_circuit}
\end{figure}

In this section, we introduce a noisy model for HRCS to understand its performance in experiments.
To model the circuit noise in HRCS, we add a local depolarizing channel $\calN(\cdot)$ to the system $A$ and bath $B$ simultaneously following the circuit unitary $U_t$ in each step, as shown in Supplementary Fig.~\ref{fig:noisy_circuit} top.
Here, we consider the depolarizing channel on a $d$-dimensional system as
\be
    \calN(\rho) = \gamma \rho + (1-\gamma) \bI/d,
    \label{eq:depol_def}
\ee
where $\bI/d$ is the fully mixed state of dimension $d$ and $1-\gamma$ is depolarizing probability. Note that the effective parameter $\gamma$ can be different for system $A$ and bath $B$.

We begin with the analysis on linear cross entropy benchmarking (XEB). Similar to the noiseless derivation in \ref{app:sampling_HRCS}, we still expand the bath system in different steps, $B_1, B_2, \dots, B_t$, to an equivalent joint bath system $\bm B = (B_1, B_2, \dots, B_t)$, shown in Supplementary Fig.~\ref{fig:noisy_circuit} bottom. Therefore, we can write out the noisy sampling probability of $\bmz(t), x(t)$ as
\be
    \tilde{P}\left[\bmz(t), x(t)\right] = \tr\left(\tilde{\calU}\left[\ketbra{\bm 0}{\bm 0}_{A\bm B}\right] \ketbra{x(t), \bmz(t)}{x(t), \bmz(t)}_{A \bm B}\right),
    \label{eq:p_noisy}
\ee
where $\tilde{\calU}$ represents the noisy circuit channel, explicitly defined as
\be
    \tilde{\calU} \coloneqq \left(\calN_{AB_t} \circ \calU_t\right) \circ \cdots \circ \left(\calN_{AB_1} \circ \calU_1\right),
\ee
with $\calU_t(\cdot) \coloneqq U_t (\cdot) U_t^\dagger$ to be noiseless unitary channel. Here, $\calN_{AB_t} = \calN_A \otimes \calN_{B_t}$ is a composite of local channels on $A$ and $B_t$. Recall that the vanilla XEB~\cite{arute2019quantum} in noisy circuit becomes
\be
    \chi(t) = 2^{n} \left(\sum_{\bmz(t), x(t)} P\left[\bmz(t), x(t)\right] \tilde{P}\left[\bmz(t), x(t)\right]\right) -1,
    \label{eq:XEB_def_app}
\ee
where $n = N_A + tN_B$.
We evaluate the probability product as follows.
\begin{align}
    &P\left[\bmz(t), x(t)\right] \tilde{P}\left[\bmz(t), x(t)\right] \nonumber\\
    &=  \tr\left(\calU\left[\ketbra{\bm 0}{\bm 0}_{A\bm B}\right] \ketbra{x(t), \bmz(t)}{x(t), \bmz(t)}_{A \bm B}\right) \tr\left(\tilde{\calU}\left[\ketbra{\bm 0}{\bm 0}_{A\bm B}\right] \ketbra{x(t), \bmz(t)}{x(t), \bmz(t)}_{A \bm B}\right) \\
    &= \tr\left[\left(\calU\otimes\tilde{\calU}\right)\left[\ketbra{\bm 0}{\bm 0}_{A\bm B}^{\otimes 2}\right] \ketbra{x(t), \bmz(t)}{x(t), \bmz(t)}_{A \bm B}^{\otimes 2}\right]\\
    &= \tr\left[\left({\bm \calN}_A\circ\calU_t^{\otimes 2}\right)\circ\left({\bm \calN}_A\circ \calU_{t-1}^{\otimes 2}\right)\circ \cdots \left({\bm \calN}_A\circ \calU_{1}^{\otimes 2}\right)\left[\ketbra{\bm 0}{\bm 0}_{A\bm B}^{\otimes 2}\right]  \ketbra{x(t)}{x(t)}_A^{\otimes 2}\otimes_{j=1}^t {\bm \calN}_{B_j}^\dagger\left(\ketbra{z_j}{z_j}_{B_j}^{\otimes 2}\right)\right],
    \label{eq:noisy_p_sq}
\end{align}
where the two-folded channel $\calU\otimes \tilde{\calU}$ owns two copies of same unitary twirling but only one noisy channel on the second copy. As the local channels on bath systems $\calN_{B_1}, \cdots, \calN_{B_t}$ and last noisy channel on system $\calN_{A}$ are only followed by measurements without any further unitary operations, we adopt the adjoint channel map to isolate these channels from circuit dynamics, resulting in the last line. Here, we define ${\bm \calN}_A \coloneqq \calI \otimes \calN_A$ for simplicity and same definition holds for $\bm \calN_B$. 
As depolarizing channel is self-adjoint, we can immediately have
\be
   {\bm \calN}_B^\dagger\left(\ketbra{z}{z}_B^{\otimes 2}\right) = \ketbra{z}{z}_B \otimes \left(\gamma_B \ketbra{z}{z}_B + (1-\gamma_B)\frac{\bI}{d_B}\right) = \gamma_B \ketbra{z}{z}^{\otimes 2}_B + \frac{1-\gamma_B}{d_B} \ketbra{z}{z}_B \otimes \bI.
   \label{eq:noisy_proj}
\ee
Furthermore, the inner product between $\kett{\hat{e}}_B, \kett{\hat{\tau}}_B$ with this noisy projector is
\begin{align}
    &{\bm \calN}_B^\dagger\left(\bbra{z}\right)\kett{\hat{e}}_B =  \left(\gamma_B \bbra{z}_B + \frac{1-\gamma_B}{d_B} \bbra{z}_B\bbra{\hat{e}}_B\right) \kett{\hat{e}}_B = 1,\\
    &{\bm \calN}_B^\dagger\left(\bbra{z}\right)\kett{\hat{\tau}}_B =  \left(\gamma_B \bbra{z}_B + \frac{1-\gamma_B}{d_B} \bbra{z}_B\bbra{\hat{e}}_B\right) \kett{\hat{\tau}}_B = \gamma_B + \frac{1-\gamma_B}{d_B} \eqqcolon g_B,
\end{align}
where we define $g_B$ in the end for simplicity.

There are another two identities we need to introduce before the formal derivation, which is the noisy channel on identity and swap operators. For the identity operator, we simply have $\calN_A(\kett{\hat{e}}_A) = \kett{\hat{e}}_A$ since the depolarizing channel is unital. For the swap operator, we have
\begin{align}
    {\bm \calN}_A\left(\kett{\hat{\tau}}\right) &= (\calI \otimes \calN_A)\left[\sum_{i,j}\ketbra{i,j}{j,i}\right] = \sum_{i,j}\ketbra{i}{j}\otimes \calN_A(\ketbra{j}{i})\\
    &=\sum_{i,j}\ketbra{i}{j}\otimes \left(\gamma_A \ketbra{j}{i} + \frac{(1-\gamma_A)\delta_{ij}}{d}\bI\right)\\
    &= \gamma_A \kett{\hat{\tau}}_A + \frac{1-\gamma_A}{d_A}\kett{\hat{e}}_A.
\end{align}

We now consider the Haar ensemble average of each circuit unitary to derive analytical results. Starting from $t=1$, the unitary twirling of $U_1$ on $\ketbra{\bm 0}{\bm 0}_{AB_1}^{\otimes 2}$ followed by local noisy channels and bath measurements, we have
\begin{align}
    &{\bm \calN}_{B_1}^\dagger\left(\bbra{z_1}_{B_1}\right) {\bm \calN}_{A}\left[\E_{U_1\in \rm Haar}\calU_1^{\otimes 2} \left(\kett{\bm 0}_{AB_1}\right)\right] \nonumber\\
    &= \frac{1}{d_A d_B(d_A d_B + 1)}{\bm \calN}_{B_1}^\dagger\left(\bbra{z_1}_{B_1}\right) \left[{\bm \calN}_{A}\left(\kett{\hat{e}}_A\right) \kett{\hat{e}}_{B_1} + {\bm \calN}_{A}\left(\kett{\hat{\tau}}_A\right) \kett{\hat{\tau}}_{B_1}\right]\\
    &= \frac{1}{d_A d_B(d_A d_B + 1)} {\bm \calN}_{B_1}^\dagger\left(\bbra{z_1}_{B_1}\right) \left[\kett{\hat{e}}_A\kett{\hat{e}}_{B_1} + \gamma_A \kett{\hat{\tau}}_A\kett{\hat{\tau}}_{B_1} + \frac{1-\gamma_A}{d_A}\kett{\hat{e}}_A \kett{\hat{\tau}}_{B_1} \right]\\
    &= \frac{1}{d_A d_B(d_A d_B + 1)} \left[\left(1 + \frac{1-\gamma_A}{d_A}g_B\right)\kett{\hat{e}}_A + \gamma_A g_B\kett{\hat{\tau}}_A\right].
    \label{eq:noisy_p_sq_0}
\end{align}
As both the identity and swap operators appear, we next turn to the transition matrix under local depolarizing noise.
For twirling on identity, we have
\begin{align}
    &{\bm \calN}_{B_2}^\dagger\left(\bbra{z_2}_{B_2}\right) {\bm \calN}_{A}\left[\E_{U_2\in \rm Haar}\calU_2^{\otimes 2} \left(\kett{\hat{e}}_A\kett{\bm 0}_{B_2}\right)\right] \nonumber\\
    &= {\bm \calN}_{B_2}^\dagger\left(\bbra{z_2}_{B_2}\right) \left[\frac{d_A^2 d_B -1}{d_B\left(d_A^2 d_B^2-1\right)}  {\bm \calN}_{A}\left(\kett{\hat{e}}_A\right)\kett{\hat{e}}_{B_2} + \frac{d_A(d_B-1)}{d_B\left(d_A^2 d_B^2-1\right)} {\bm \calN}_{A}\left(\kett{\hat{\tau}}_A\right)\kett{\hat{\tau}}_{B_2}\right]\\
    &= \frac{d_A^2 d_B -1}{d_B\left(d_A^2 d_B^2-1\right)} \kett{\hat{e}}_A + \frac{d_A(d_B-1)}{d_B\left(d_A^2 d_B^2-1\right)} \left(\gamma_A \kett{\hat{\tau}}_A + \frac{1-\gamma_A}{d_A}\kett{\hat{e}}_A\right) g_B \\
    &= \left[\frac{d_A^2 d_B -1}{d_B\left(d_A^2 d_B^2-1\right)} + \frac{d_A(d_B-1)}{d_B\left(d_A^2 d_B^2-1\right)} \frac{1-\gamma_A}{d_A} g_B\right] \kett{\hat{e}}_A + \frac{d_A(d_B-1)}{d_B\left(d_A^2 d_B^2-1\right)} \gamma_A g_B  \kett{\hat{\tau}}_A\\
    &\coloneqq m_{00} \kett{\hat{e}}_A + m_{10} \kett{\hat{\tau}}_A.
    \label{eq:noisy_p_sq_1}
\end{align}
Similarly, we can derive the twirling on swap as
\begin{align}
    &{\bm \calN}_{B_2}^\dagger\left(\bbra{z_2}_{B_2}\right) {\bm \calN}_{A}\left[\E_{U_2\in \rm Haar}\calU_2^{\otimes 2} \left(\kett{\hat{\tau}}_A\kett{\bm 0}_{B_2}\right)\right] \nonumber\\
    &= {\bm \calN}_{B_2}^\dagger\left(\bbra{z_2}_{B_2}\right) \left[\frac{d_A(d_B-1)}{d_B\left(d_A^2 d_B^2-1\right)}  {\bm \calN}_{A}\left(\kett{\hat{e}}_A\right)\kett{\hat{e}}_{B_2} +  \frac{d_A^2 d_B -1}{d_B\left(d_A^2 d_B^2-1\right)} {\bm \calN}_{A}\left(\kett{\hat{\tau}}_A\right)\kett{\hat{\tau}}_{B_2}\right]\\
    &= \frac{d_A(d_B-1)}{d_B\left(d_A^2 d_B^2-1\right)} \kett{\hat{e}}_A + \frac{d_A^2 d_B -1}{d_B\left(d_A^2 d_B^2-1\right)} \left(\gamma_A \kett{\hat{\tau}}_A + \frac{1-\gamma_A}{d_A}\kett{\hat{e}}_A\right) g_B \\
    &= \left[\frac{d_A(d_B-1)}{d_B\left(d_A^2 d_B^2-1\right)} +\frac{d_A^2 d_B -1}{d_B\left(d_A^2 d_B^2-1\right)} \frac{1-\gamma_A}{d_A} g_B\right] \kett{\hat{e}}_A + \frac{d_A^2 d_B -1}{d_B\left(d_A^2 d_B^2-1\right)} \gamma_A g_B  \kett{\hat{\tau}}_A\\
    &\coloneqq m_{01} \kett{\hat{e}}_A + m_{11} \kett{\hat{\tau}}_A.
    \label{eq:noisy_p_sq_2}
\end{align}
Combining Eqs.~\eqref{eq:noisy_p_sq_1} and \eqref{eq:noisy_p_sq_2}, we can summarize the noisy transition matrix as
\begin{align}
    \tilde{M} = \begin{pmatrix}
        m_{00} & m_{01} \\
        m_{10} & m_{11}
    \end{pmatrix},
    \label{eq:noisy_M_def}
\end{align}
and when $\gamma_A = \gamma_B = 1$ for noiseless circuit, the transition matrix $\tilde{M}$ reduces to the transition matrix $M$ in Eq.~\eqref{eq:M_def} as expected.

The ensemble average of the sum in Eq.~\eqref{eq:noisy_p_sq} then becomes
\begin{align}
    &\E_{U\in \rm Haar}P\left[\bmz(t), x(t)\right] \tilde{P}\left[\bmz(t), x(t)\right] \nonumber\\
    &= \frac{1}{d_A d_B(d_A d_B+1)}\bra{\underline{0}}\tilde{M}^{t-1}\left[\left(1 + \frac{1-\gamma_A}{d_A}g_B\right)\ket{\underline{0}} + \gamma_A g_B\ket{\underline{1}}\right] \nonumber\\
    &\quad + \frac{1}{d_A d_B(d_A d_B+1)}\bra{\underline{1}}\tilde{M}^{t-1}\left[\left(1 + \frac{1-\gamma_A}{d_A}g_B\right)\ket{\underline{0}} + \gamma_A g_B\ket{\underline{1}}\right]\\
    &= \frac{1}{d_A d_B(d_A d_B+1)}\left(\bra{\underline{0}} + \bra{\underline{1}}\right) \tilde{M}^{t-1}\left[\left(1 + \frac{1-\gamma_A}{d_A}g_B\right)\ket{\underline{0}} + \gamma_A g_B\ket{\underline{1}}\right],
    \label{eq:noisy_p_sq_result}
\end{align}
where we still use $\ket{\underline{0}}, \ket{\underline{1}}$ to represent $\kett{\hat{e}}, \kett{\hat{\tau}}$ as in \ref{app:sampling_HRCS}. The closed form of Eq.~\eqref{eq:noisy_p_sq_result} is too lengthy to present here, but in principle, one can evaluate Eq.~\eqref{eq:noisy_p_sq_result} numerically as it is only a $2$-dimensional linear system. The vanilla linear XEB then becomes
\be
    \chi_{\rm HRCS}(t) = 2^n \frac{1}{d_A d_B(d_A d_B+1)}\left(\bra{\underline{0}} + \bra{\underline{1}}\right) \tilde{M}^{t-1}\left[\left(1 + \frac{1-\gamma_A}{d_A}g_B\right)\ket{\underline{0}} + \gamma_A g_B\ket{\underline{1}}\right] - 1.
    \label{eq:full_noisy_xeb_app}
\ee

Here we provide an approximate approach to simplify the calculation in the asymptotic limit of $d_A, d_B \gg 1$. Before evaluation of the matrix power $\tilde{M}^{t-1}$, we first keep the leading order of the matrix $\tilde{M}$ as
\begin{align}
    \tilde{M} \simeq \begin{pmatrix}
        1/d_B^2 & (1+\gamma_B - \gamma_A \gamma_B)/(d_A d_B^2) \\
        \gamma_A \gamma_B/(d_A d_B^2) & \gamma_A \gamma_B / d_B^2
    \end{pmatrix},
\end{align}
and
\be
    \frac{1}{d_A d_B(d_A d_B + 1)}\left[\left(1 + \frac{1-\gamma_A}{d_A}g_B\right)\kett{\hat{e}} + \gamma_A g_B\kett{\hat{\tau}}\right] \simeq \frac{1}{d_A^2 d_B^2} \kett{\hat{e}} + \frac{\gamma_A \gamma_B}{d_A^2 d_B^2} \kett{\hat{\tau}}.
\ee
We now can evaluate the linear XEB between ideal and noisy distribution as
\begin{align}
    &2^n \E_{U\in \rm Haar}P\left[\bmz(t), x(t)\right] \tilde{P}\left[\bmz(t), x(t)\right] = \left(\bra{\underline{0}} + \bra{\underline{1}}\right) \tilde{M}^{t-1}\left[\frac{1}{d_A^2 d_B^2}\ket{\underline{0}} + \frac{\gamma_A \gamma_B}{d_A^2 d_B^2}\ket{\underline{1}}\right] \\
    &\simeq \frac{d_B^{-t}}{d_A^3(\gamma_A \gamma_B-1)}\left[\gamma_A^t \gamma_B^t \left(d_A^2 (\gamma_A \gamma_B-1)-\gamma_A \gamma_B^2\right) + 1+\gamma_B+\gamma_A \gamma_B d_A ((\gamma_A-1) \gamma_B-2) + d_A^2(\gamma_A \gamma_B - 1)\right.\nonumber\\
    &\qquad \qquad \qquad \qquad \left. + \gamma_A \gamma_B^2 \left(\gamma_A \left(\gamma_A^t \gamma_B^t-1\right)+1\right)\right] \\
    &\simeq \frac{\gamma_A^t \gamma_B^t+1}{d_A d_B^t} + \frac{\gamma_A \gamma_B (2  + \gamma_B-\gamma_A \gamma_B)}{d_A^2 d_B^{t}(1-\gamma_A \gamma_B)}.
\end{align}
Therefore, we can have the approximate form of the vanilla linear XEB in the noisy HRCS circuit as
\begin{align}
    \chi_{\rm HRCS}(t) &= 2^n 2^n \E_{U\in \rm Haar} P\left[\bmz(t), x(t)\right] \tilde{P}\left[\bmz(t), x(t)\right] -1 \nonumber\\
    &= 2^n\left[\frac{\gamma_A^t \gamma_B^t+1}{d_A d_B^t} + \frac{\gamma_A \gamma_B (2  + \gamma_B-\gamma_A \gamma_B)}{d_A^2 d_B^{t}(1-\gamma_A \gamma_B)}\right] -1\\
    &= \gamma_A^t \gamma_B^t + \frac{\gamma_A \gamma_B (2  + \gamma_B -\gamma_A \gamma_B)}{d_A(1-\gamma_A \gamma_B)}.
    \label{eq:noisy_XEB_asymp_app}
\end{align}

In the main text, we plot the normalized XEB for fair comparison across different time. The normalized XEB theory for noisy circuit becomes
\begin{align}
    \tilde{\chi}_{\rm HRCS}(t) &= \frac{\chi_{\rm HRCS}(t)}{2^n Z_{\rm HRCS}(t) -1} \simeq \frac{1}{2\exp\left(\frac{t-1}{d_A}\right)-1} \left[\gamma_A^t \gamma_B^t + \frac{\gamma_A \gamma_B (2  + \gamma_B -\gamma_A \gamma_B)}{d_A(1-\gamma_A \gamma_B)}\right] \label{eq:XEBnorm_app}.
\end{align}
In the early time, the normalized XEB presents exponentially decay as $\sim \gamma_A^t \gamma_B^t$ as conventional RCS, while in the late time for $\epsilon$-approximate AC, the decay of normalized XEB is mainly dominated by the growth of $Z_{\rm HRCS}(t)$. We also notice that in Eq.~\eqref{eq:XEBnorm_app}, the product of noisy parameter $\gamma_A \gamma_B$ determines the leading-order scaling.

We next discuss the relation between noisy XEB and state fidelity.
XEB is known to become a proxy for the state fidelity in conventional RCS though with limitations~\cite{gao2024limitations}. Here we examine whether the XEB in HRCS can also work as a proxy for state fidelity. It is easy to see that the XEB does not characterize the fidelity of physical output conditional state $\ket{\psi_\bmz}_A$ as the support dimension does not match. Instead, we consider the fidelity of the equivalent joint state $U\ket{0}_A \ket{0}_{\bm B}$ supported on system and equivalent expanded baths (see Supplementary Fig.~\ref{fig:noisy_circuit} bottom). The fidelity between noisy state and ideal state is
\begin{align}
    F &= \braket{\psi|\rho|\psi} = \tr\left(\tilde{\calU}\left[\ketbra{\bm 0}{\bm 0}_{A\bm B}\right] \calU\left[\ketbra{\bm 0}{\bm 0}_{A\bm B}\right]\right)\\
    &=\tr\left(\left(\calU\otimes \tilde{\calU}\right)\left[\ketbra{\bm 0}{\bm 0}_{A\bm B}^{\otimes 2}\right] \tau_A \tau_{\bm B} \right)\\
    &= \tr\left[\left({\bm \calN}_A\circ\calU_t^{\otimes 2}\right) \circ \cdots \left({\bm \calN}_A\circ\calU_1^{\otimes 2}\right) \left[\ketbra{\bm 0}{\bm 0}_{A\bm B}^{\otimes 2}\right] \tau_A \otimes_{j=1}^t {\bm \calN_{B_j}}^\dagger\left(\tau_{B_j}\right) \right].
    \label{eq:noisy_fidelity_def}
\end{align}
Starting from $t=1$, the unitary twirling of $U_1$ on $\ketbra{\bm 0}{\bm 0}_{A B_1}^{\otimes 2}$ followed by local noisy channels, we have
\begin{align}
    &{\bm \calN}_{B_1}^\dagger\left(\bbra{\hat{\tau}}_{B_1}\right) {\bm \calN}_{A}\left[\E_{U_1\in \rm Haar}\calU_1^{\otimes 2} \left(\kett{\bm 0}_{AB_1}\right)\right] \nonumber\\
    &= \frac{1}{d_A d_B(d_A d_B + 1)}{\bm \calN}_{B_1}^\dagger\left(\bbra{\hat{\tau}}_{B_1}\right) \left[{\bm \calN}_{A}\left(\kett{\hat{e}}_A\right) \kett{\hat{e}}_{B_1} + {\bm \calN}_{A}\left(\kett{\hat{\tau}}_A\right) \kett{\hat{\tau}}_{B_1}\right]\\
    &= \frac{1}{d_A d_B(d_A d_B + 1)} \left(\gamma_B \bbra{\hat{\tau}}_{B_1} + \frac{1-\gamma_B}{d_B}\bbra{\hat{e}}_{B_1}\right) \left[\kett{\hat{e}}_A\kett{\hat{e}}_{B_1} + \gamma_A \kett{\hat{\tau}}_A\kett{\hat{\tau}}_{B_1} + \frac{1-\gamma_A}{d_A}\kett{\hat{e}}_A \kett{\hat{\tau}}_{B_1} \right]\\
    &= \frac{1}{d_A d_B(d_A d_B + 1)} \left[\left(d_B + \frac{(1-\gamma_A)(1-\gamma_B+\gamma_B d_B^2)}{d_A}\right)\kett{\hat{e}}_A + \gamma_A(1-\gamma_B +\gamma_B d_B^2) \kett{\hat{\tau}}_A\right].
    \label{eq:noisy_f_0}
\end{align}
We next calculate the noisy twirling on identity and swap operator separately.
\begin{align}
    &{\bm \calN}_{B_2}^\dagger\left(\bbra{\hat{\tau}}_{B_2}\right) {\bm \calN}_{A}\left[\E_{U_2\in \rm Haar}\calU_2^{\otimes 2} \left(\kett{\hat{e}}_A\kett{\bm 0}_{B_2}\right)\right] \nonumber\\
    &= \left(\gamma_B \bbra{\hat{\tau}}_{B_2} + \frac{1-\gamma_B}{d_B} \bbra{\hat{e}}_{B_2}\right) \left[\frac{d_A^2 d_B -1}{d_B\left(d_A^2 d_B^2-1\right)}  {\bm \calN}_{A}\left(\kett{\hat{e}}_A\right)\kett{\hat{e}}_{B_2} + \frac{d_A(d_B-1)}{d_B\left(d_A^2 d_B^2-1\right)} {\bm \calN}_{A}\left(\kett{\hat{\tau}}_A\right)\kett{\hat{\tau}}_{B_2}\right]\\
    &= \frac{d_A^2 d_B -1}{d_B\left(d_A^2 d_B^2-1\right)} \left(\gamma_B d_B + \frac{1-\gamma_B}{d_B} d_B^2\right)\kett{\hat{e}}_A + \frac{d_A(d_B-1)}{d_B\left(d_A^2 d_B^2-1\right)} \left(\gamma_B d_B^2 + \frac{1-\gamma_B}{d_B} d_B\right) \left(\gamma_A \kett{\hat{\tau}}_A + \frac{1-\gamma_A}{d_A}\kett{\hat{e}}_A\right) \\
    &= \frac{1}{d_A^2 d_B^2 -1}\left[d_A^2 d_B -1 + \frac{(d_B-1)(1-\gamma_A)}{d_B} \left(1-\gamma_B + \gamma_B d_B^2\right)\right] \kett{\hat{e}}_A + \frac{d_A(d_B-1)\gamma_A}{d_B\left(d_A^2 d_B^2-1\right)} \left(1-\gamma_B + \gamma_B d_B^2\right)   \kett{\hat{\tau}}_A\\
    &\coloneqq s_{00} \kett{\hat{e}}_A + s_{10} \kett{\hat{\tau}}_A.
    \label{eq:noisy_f_1}
\end{align}

\begin{align}
    &{\bm \calN}_{B_2}^\dagger\left(\bbra{\hat{\tau}}_{B_2}\right) {\bm \calN}_{A}\left[\E_{U_2\in \rm Haar}\calU_2^{\otimes 2} \left(\kett{\hat{\tau}}_A\kett{\bm 0}_{B_2}\right)\right] \nonumber\\
    &= \left(\gamma_B \bbra{\hat{\tau}}_{B_2} + \frac{1-\gamma_B}{d_B} \bbra{\hat{e}}_{B_2}\right) \left[\frac{d_A(d_B-1)}{d_B\left(d_A^2 d_B^2-1\right)}  {\bm \calN}_{A}\left(\kett{\hat{e}}_A\right)\kett{\hat{e}}_{B_2} +  \frac{d_A^2 d_B -1}{d_B\left(d_A^2 d_B^2-1\right)} {\bm \calN}_{A}\left(\kett{\hat{\tau}}_A\right)\kett{\hat{\tau}}_{B_2}\right]\\
    &= \frac{d_A(d_B-1)}{d_B\left(d_A^2 d_B^2-1\right)}\left(\gamma_B d_B + \frac{1-\gamma_B}{d_B} d_B^2\right)  \kett{\hat{e}}_A + \frac{d_A^2 d_B -1}{d_B\left(d_A^2 d_B^2-1\right)} \left(\gamma_B d_B^2 + \frac{1-\gamma_B}{d_B} d_B\right)  \left(\gamma_A \kett{\hat{\tau}}_A + \frac{1-\gamma_A}{d_A}\kett{\hat{e}}_A\right)  \\
    &= \frac{1}{d_A^2 d_B^2-1} \left[ d_A(d_B-1) + \frac{(d_A^2 d_B-1)(1-\gamma_A)}{d_A d_B}\left(1-\gamma_B + \gamma_B d_B^2\right)\right] \kett{\hat{e}}_A + \frac{\left(d_A^2 d_B -1\right)\gamma_A}{d_B\left(d_A^2 d_B^2-1\right)} \left(1-\gamma_B + \gamma_B d_B^2\right)  \kett{\hat{\tau}}_A\\
    &\coloneqq s_{01} \kett{\hat{e}}_A + s_{11} \kett{\hat{\tau}}_A.
    \label{eq:noisy_f_2}
\end{align}
The noisy transition matrix can also be formulated as a $2\times 2$ matrix
\be
    S = \begin{pmatrix}
        s_{00} & s_{01} \\
        s_{10} & s_{11}
    \end{pmatrix}.
\ee
The ensemble average fidelity in Eq.~\eqref{eq:noisy_fidelity_def} then becomes
\begin{align}
    F &= \frac{1}{d_A d_B(d_A d_B + 1)} \left(d_A\bra{\underline{0}} + d_A^2 \bra{\underline{1}}\right) S^{t-1}   \left[\left(d_B + \frac{(1-\gamma_A)(1-\gamma_B+\gamma_B d_B^2)}{d_A}\right) \ket{\underline{0}} + \gamma_A\left(1-\gamma_B + \gamma_B d_B^2\right) \ket{\underline{1}}\right] \label{eq:fid_exact_app} \\
    &\simeq (\gamma_A \gamma_B)^t,
    \label{eq:fid_asymp_app}
\end{align}
where we provide a simple asymptotic form of the fidelity. This asymptotic results can be interpreted as the counting of the number of depolarizing channels in the noisy HRCS model.


\begin{figure}[t]
    \centering
    \includegraphics[width=0.65\textwidth]{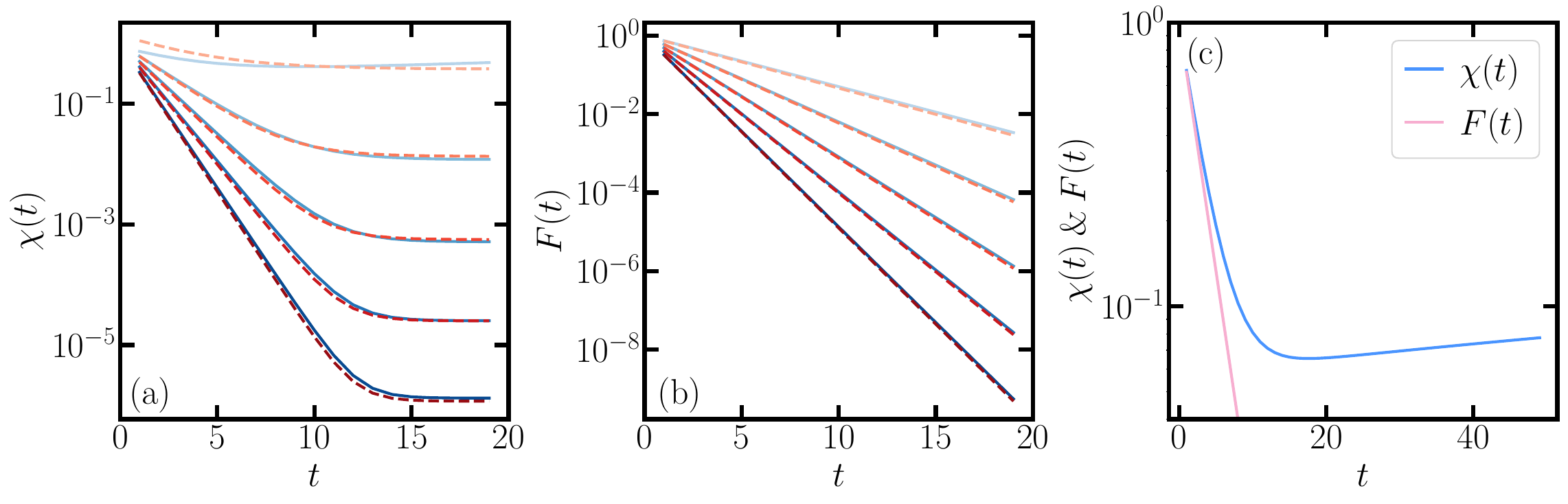}
    \caption{Noisy theory for XEB and fidelity. (a) Exact theory (blues solid lines) and asymptotic theory (red dashed lines) of vanilla XEB from Eq.~\eqref{eq:full_noisy_xeb_app} and Eq.~\eqref{eq:noisy_XEB_asymp_app}. (b) Exact theory (blues solid lines) and asymptotic theory (red dashed lines) of fidelity from Eq.~\eqref{eq:fid_exact_app} and Eq.~\eqref{eq:fid_asymp_app}. 
    In both cases, we choose $N_B = 2$ and various $N_A=4, 8, 12, 16, 20$ (top to bottom). The depolarizing channel parameters are set as $\gamma_A=q^{N_A}, \gamma_B = q^{N_B}$ with $q = 0.95$. In (c), we compare the vanilla noisy XEB and fidelity with $N_A = 6, N_B = 2, q=0.95$.}
    \label{fig:noisy_XEB}
\end{figure}

In Supplementary Fig.~\ref{fig:noisy_XEB}a, we plot the exact vanilla XEB from Eq.~\eqref{eq:full_noisy_xeb_app} in a system of $N_B = 2$ and different $N_A$ (blue solid lines). Here we model the depolarizing channel as $\gamma_A = q^{N_A}$ and $\gamma_B = q^{N_B}$ with $q$ to be a free parameter. 
With the increase of system size $N_A$, the exact theory and asymptotic one (red dashed lines) in Eq.~\eqref{eq:noisy_XEB_asymp_app} agrees well despite small discrepancy in late time.
In Supplementary Fig.~\ref{fig:noisy_XEB}b, the noisy theory of fidelity also agrees very well with the asymptotic theory of Eq.~\eqref{eq:fid_asymp_app}, which also aligns the physical intuition in RCS. Therefore, through a simple comparison in Fig.~\ref{fig:noisy_XEB}c, the vanilla XEB can be regarded as a fidelity proxy as $(\gamma_A \gamma_B)^t$ in the early time of $t \ll N_A$ but fails at the late time, due to the heavy accumulated noise in the circuit, which is also identified in Ref.~\cite{gao2024limitations}.

In the end of this section, we present the evaluation of XEB in the experiments of Fig.~\ref{fig:experiment}b with patch circuit.
As the two patches are disjoint circuits, the joint sampling distribution can be factorized as $P[\bmz_1(t), x_1(t),\bmz_2(t), x_2(t)] = P[\bmz_1(t), x_1(t)]P[\bmz_2(t), x_2(t)]$. Therefore, when each patch has $N_A'$ system and $N_B'$ bath qubits, the ideal XEB fidelity of the two-patch circuit becomes
\begin{align}
    \calF_{\rm XEB,patch}(t) = \left(d_{A^\prime} d_{B^\prime}^t\right)^2\left[Z_{\rm HRCS}(t)\right]^2 -1 = \frac{4d_{A^\prime}^2 d_{B^\prime}^{2t} (d_{A^\prime} + 1)^{2(t-1)}}{(1+d_{A^\prime} d_{B^\prime})^{2t}} -1,
\end{align}
where $Z_{\rm HRCS}(t)$ is the ensemble-averaged CP in HRCS with dimensions $d_{A^\prime} = 2^{N_A'}, d_{B^\prime}=2^{N_B'}$ (see Eq.~\eqref{eq:HRCS_CP_spt_app}). In the large system limit of $N_A'\gg 1$, the ideal XEB fidelity scales as $\sim 4\exp\left[\frac{2t\left(1-d_{B^\prime}^{-1}\right)+d_{B^\prime}^{-1}-1}{d_{A^\prime}}\right]-1$, which still exponentially grows with temporal steps $t$ but suppressed by the system dimension $d_{A^\prime}$. Similarly, the noisy theory of XEB for patch circuit becomes
\begin{align}
    \chi_{\rm 2-patch}(t) =  \left[1+\chi_{\rm HRCS}(t)\right]^2 -1,
\end{align}
where $\chi_{\rm HRCS}(t)$ is the noisy theory of XEB for single HRCS in Eq.~\eqref{eq:noisy_XEB_asymp_app} with dimenions $d_{A^\prime} = 2^{N_A'}$, $d_{B^\prime} =2^{N_B'}$.

\section{Additional details of experiments}
\label{app:experiment}
In this section, we provide additional details of the experiments on IBM Quantum devices.

In the first experiment of Fig.~\ref{fig:experiment}a in the main text, we implement a circuit of $N_A = 5, N_B = 5$ qubits on the IBMQ Torino and in each step, we apply a $L=8$ layer parameterized circuit. Due to the automatic qubit allocation from Qiskit, for temporal steps $t\le 10$, our circuit uses qubits shown in Supplementary Fig.~\ref{fig:circuit_5_5}a (blue circuits) while for $t=13, 16$, it uses qubits shown in Supplementary Fig.~\ref{fig:circuit_5_5}b. In both cases, the circuit is applied on a line of 1D qubits and thus the two-qubit gates are only applied on nearest neighbors as we requested. 
In Supplementary Fig.~\ref{fig:gates}a, we show the number of native gates of IBMQ Torino used in our experiments for different $t$ thus effective qubit number $N_{\rm eff}(t)$, and the number of gates increases linearly with $t$ as expected. The number of gates in each step is approximately the same, and for the largest temporal steps of $t=16$, we use $5049$ SX gates, $3684$ RZ gates and $1152$ CZ gates. For the larger-scale experiment of Fig.~\ref{fig:experiment}b in the main text, we show the number of native gates in the circuits in Supplementary Fig.~\ref{fig:gates}b. As we implement two patch circuits, the number of gates used in this experiment is roughly two times the one in Supplementary Fig.~\ref{fig:gates}a for the same $t$.

In Supplementary Fig.~\ref{fig:error_5_5}, we present the typical circuit operation error rates and qubit lifetime in our experiments of Fig.~\ref{fig:experiment}a in the main text, provided by IBM Quantum. The top and bottom panels correspond to the experiments with $t\le 10$ and $t =13, 16$. The detailed gate and readout error for each qubit can be found in the subplots, and we list the average error rate over used qubits, which are also presented in Methods. For $t\le 10$, the average SX gate error is $e_{\rm SX} = (3.07\pm 1.27) \times 10^{-4}$, CZ gate error is $e_{\rm CZ} = (2.64 \pm 0.673)\times 10^{-3}$ and readout error $e_{\rm read} = (2.14\pm 1.48)\times 10^{-2}$. For $t = 13, 16$, the average SX gate error is $e_{\rm SX} = (3.38\pm 1.86) \times 10^{-4}$, CZ gate error is $e_{\rm CZ} = (2.39 \pm 0.382)\times 10^{-3}$ and readout error $e_{\rm read} = (2.36\pm 1.55)\times 10^{-2}$. The average qubit relaxation time for $t\le 10$ is $T_1 = 146\pm 51.2 \mu{\rm s}$  and the average dephasing time is $T_2 = 120 \pm 48.4 \mu {\rm s}$. For $t = 13, 16$, they are $T_1 = 180\pm 25.4 \mu{\rm s}, T_2 = 131 \pm 49.9 \mu {\rm s}$.

Finally, we list the XEB results in conventional RCS experiments up to now in Table~\ref{tab:xeb_list} and our HRCS results in Table~\ref{tab:xeb_hrcs}. Note that Quantinuum utilizes a different method of mirror benchmarking (MB) to estimate the state fidelity. 
Here we do {\em not} aim to compare the performance in terms of state/device fidelity which is unfair across different platforms, circuit architectures and sampling algorithms. Instead, we report the experimental data here for reference in terms of sampling on the majority of classical distribution revealed by XEB.

\begin{figure*}[t]
    \centering
    \includegraphics[width=0.55\textwidth]{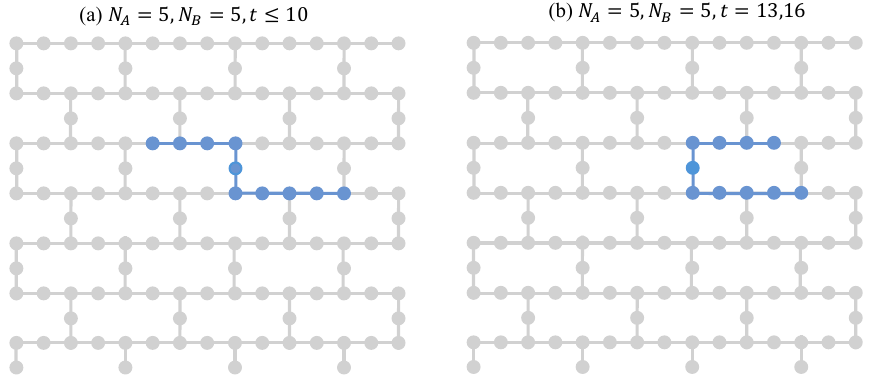}
    \caption{Circuit details in the experiment of Fig.~\ref{fig:experiment}. In (a) and (b), we plot the qubits that are used for the experiment with $t\le 10$ and $t=13,16$ on IBMQ Torino.}
    \label{fig:circuit_5_5}
\end{figure*}

\begin{figure*}[t]
    \centering
    \includegraphics[width=0.55\textwidth]{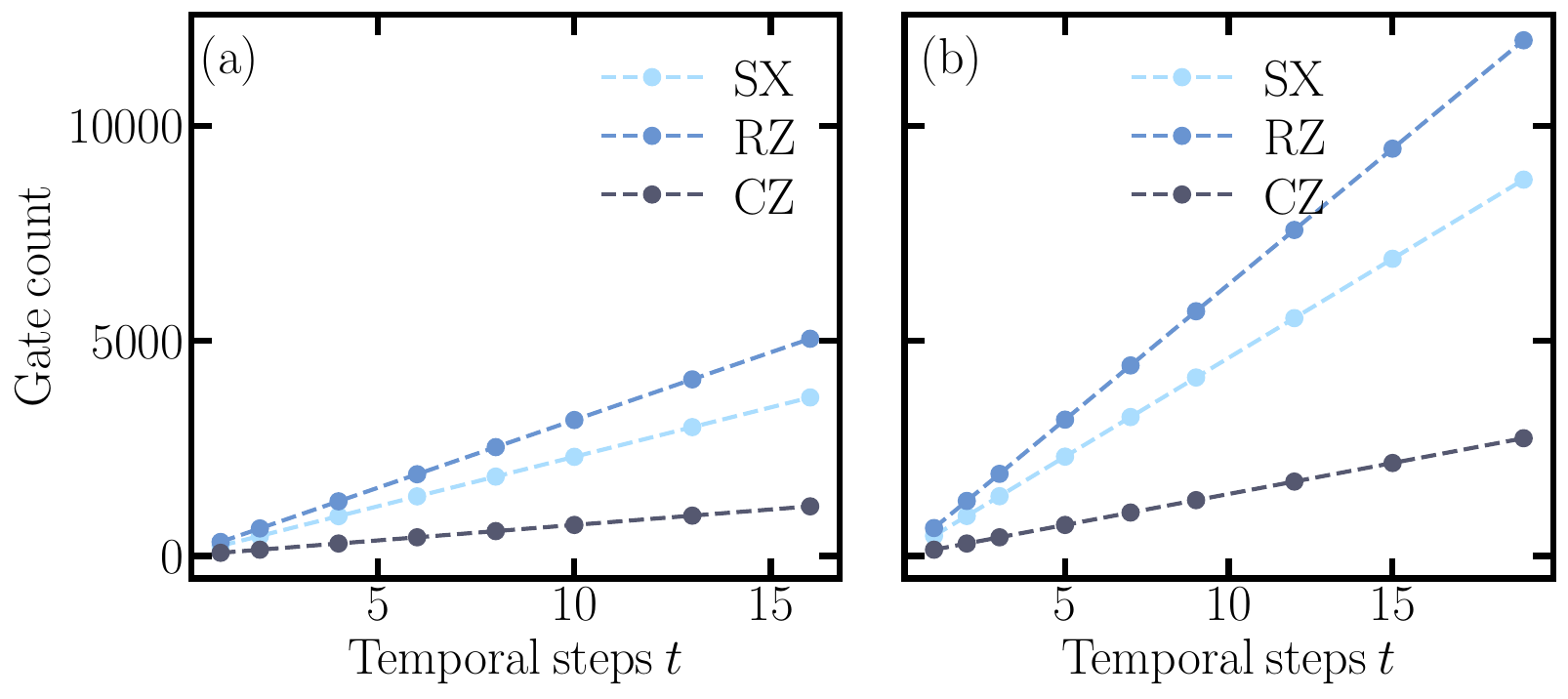}
    \caption{Number of SX, RZ and CZ gates used in experiments of Fig.~\ref{fig:experiment} (a) and (b) are shown in subplots a and b separately.}
    \label{fig:gates}
\end{figure*}

\begin{figure*}[t]
    \centering
    \includegraphics[width=\textwidth]{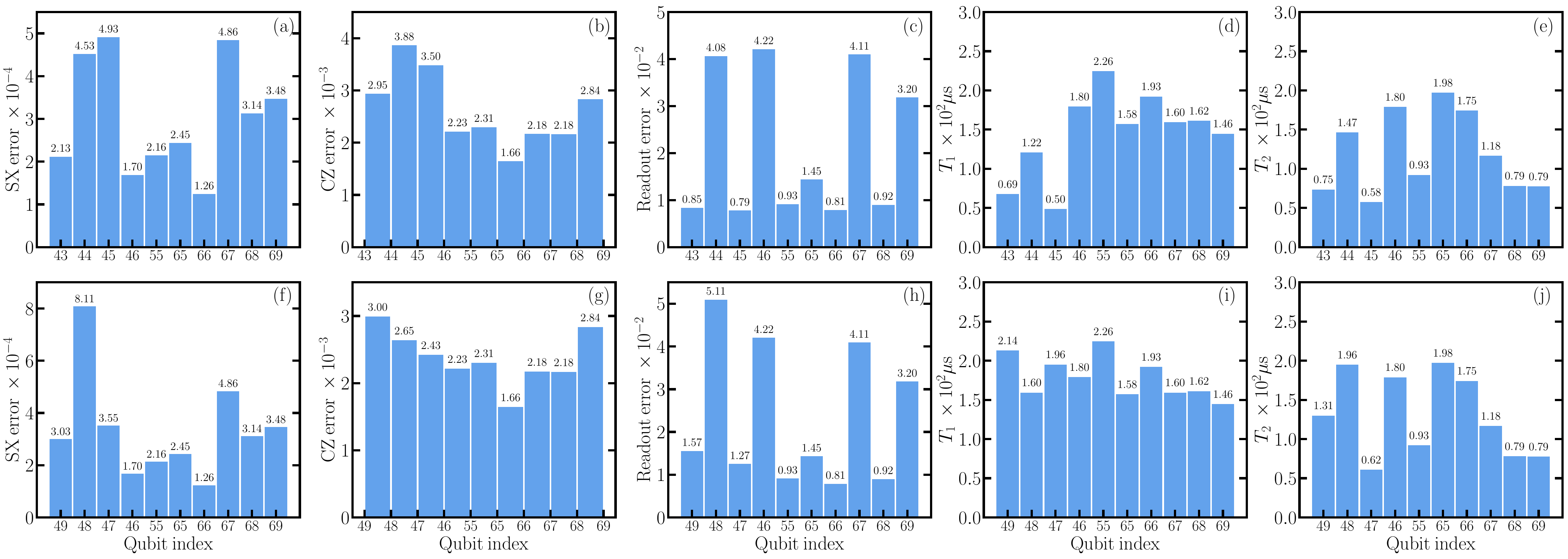}
    \caption{Circuit operation error rates and qubit lifetime in the experiment of Fig.~\ref{fig:experiment}a. We list the SX error, CZ error, readout error, relaxation time $T_1$ and dephasing time $T_2$ of qubits. (a)-(e) correspond to experiments with $t\le 10$ and (f)-(j) correspond to experiments with $t=13, 16$. In (a), (c), (f), (h), the bars show the error of the single qubit and in (b) and (g), each bar shows the error of the CZ gate connecting the qubit of left and right edges of the bar.}
    \label{fig:error_5_5}
\end{figure*}

\begin{table*}
    \centering
    \renewcommand{\arraystretch}{1.3} 
    \begin{tabular}{|c|c|c|c|}
    \hline
    Author \& Platform  &  \# of qubits 
    & Circuit & Fidelity (estimator) \\
    \hline
    Google 2019 & \multirow{2}{*}{53} & 2D circuit (elided) & \multirow{2}{*}{0.00224 (XEB)}\\
    Sycamore & & 20 cycles &  \\
    \hline
    USTC 2021 & \multirow{2}{*}{56} & 2D circuit (elided) & \multirow{2}{*}{0.000662 (XEB)}\\
    {\em Zuchongzhi} & & 20 cycles & \\
    \hline
    USTC 2022 & \multirow{2}{*}{60} & 2D circuit (elided) & \multirow{2}{*}{0.000366 (XEB)}\\
    {\em Zuchongzhi} 2.1 & & 24 cycles & \\
    \hline
    Google 2024 & \multirow{2}{*}{70} & 2D phased matched circuit (2-patch) & \multirow{2}{*}{0.0006 (XEB)}\\
    SYC-70 & & 26 cycles & \\
    \hline
    USTC 2025 & \multirow{2}{*}{83} & 2D circuit (4-patch) & \multirow{2}{*}{0.0003 (XEB)}\\
    {\em Zuchongzhi} 3.0 & & 32 cycles & \\
    \hline
    Quantinuum 2025 & \multirow{2}{*}{56} & randomly connected circuit & \multirow{2}{*}{0.18 (MB)}\\
    H2 & & 20 cycles & \\
    \hline
    Google 2025 & \multirow{2}{*}{103} & 2D circuit (3-patch) & \multirow{2}{*}{0.001 (XEB)}\\
    Willow & & 40 cycles & \\
    \hline
    Quantinuum 2025 & \multirow{2}{*}{98} & randomly connected circuit & \multirow{2}{*}{0.0449 (MB)}\\
    Helios & & 26 cycles & \\
    \hline
    \end{tabular}
    \caption{List of conventional RCS experiments~\cite{arute2019quantum, wu2021strong, zhu2022quantum, morvan2024phase, deCross2025, gao2025establishing, abanin2025observation, ransford2025helios}. Google and USTC implemented circuits on their superconducting quantum chip of a two-dimensional grid. Quantinuum implemented their circuits on their trapped-ion chip with arbitrary geometric connectivity. Google and USTC utilize XEB to becnhamrk state fidelity while Quantinuum turns to mirror benchmarking (MB).}
    \label{tab:xeb_list}
\end{table*}

\begin{table*}
    \centering
    \renewcommand{\arraystretch}{1.3} 
    \begin{tabular}{|c|c|c|c|}
    \hline
    \# of classical bits & circuit & vanilla XEB & normalized XEB \\
    \hline
    10 & full & 0.528 & 0.529\\
    \hline
    55 & full & 0.0721 & 0.0447\\
    \hline
    85 & full & 0.0443 & 0.0208 \\
    \hline
    \hline
    20 & 2-patch & 1.21 & 0.405 \\
    \hline
    100 & 2-patch & 0.118 & 0.0218 \\
    \hline
    200 & 2-patch & 0.0593 & 0.00556 \\
    \hline
    \end{tabular}
    \caption{Experimental result of our HRCS on IBM Quantum Torino. In the top panel of experiment in Fig.~\ref{fig:experiment}a, we implement one-dimensional circuit with $10$ qubits, and in the botttom panel of experiment in Fig.~\ref{fig:experiment}b, we implement two patch of one-dimensional circuit with a total of $20$ qubits.}
    \label{tab:xeb_hrcs}
\end{table*}

\end{widetext}

\end{document}